\documentclass[draft,12pt,reqno,a4paper]{amsart}
\usepackage[margin=3cm,footskip=1cm]{geometry}

\usepackage{amssymb} \usepackage{amsmath} \usepackage{amsfonts}
\usepackage{amsthm} \usepackage{mathtools}
\usepackage{enumerate} \usepackage[latin1]{inputenc}
\usepackage[shortlabels]{enumitem}

\usepackage{color} 
\usepackage{graphicx} \usepackage{verbatim}

\DeclareMathOperator*{\wlim}{w--lim}
\DeclareMathOperator*{\vGlim}{{\mathcal G}--lim}

\DeclareMathOperator*{\wvGlim}{{w--\mathcal G}--lim}
\DeclareMathOperator*{\wBstarlim}{{w^\star--B^*}--lim}
\DeclareMathOperator*{\wBlim}{{w--B}--lim}

 \newcommand{\cs}{{\rm the Cauchy-Schwarz inequality }}
\newcommand{\cS}{{\rm the Cauchy-Schwarz inequality}}

\newcommand{\N}{{\mathbb{N}}} 
\newcommand{\R}{{\mathbb{R}}} 
\newcommand{\C}{{\mathbb{C}}}

 \renewcommand{\c}{{\rm c}}
\newcommand{\e}{{\rm e}} \newcommand{\ess}{{\rm ess}}
 \renewcommand{\i}{{\rm i}}
\renewcommand{\d}{{\rm d}}

\newcommand{\diag}{{\rm diag}}

\renewcommand{\Re}{{\rm Re}\,} \renewcommand{\Im}{{\rm Im}\,}
 
\newcommand{\intR}{{-\!\!\!\!\!\int_R\,}}

\DeclarePairedDelimiter\inp\langle\rangle
\newcommand\parb[2][]{#1 \big ( #2#1\big )} \newcommand\parbb[2][]{#1
  \Big ( #2#1\Big )} 
 \newcommand{\grad}{{\rm grad }\,}
 
 \renewcommand{\exp}{{\rm exp}}
\newcommand{\mand}{\text{ and }} \newcommand{\mfor}{\text{ for }}

\newcommand{\vA}{{\mathcal A}} \newcommand{\vB}{{\mathcal B}}
 \newcommand{\vD}{{\mathcal D}}
\newcommand{\vE}{{\mathcal E}} 
\newcommand{\vG}{{\mathcal G}} \newcommand{\vI}{{\mathcal I}}
 \newcommand{\vH}{{\mathcal H}}
 
 \newcommand{\vN}{{\mathcal N}}

\theoremstyle{plain}%default
\newtheorem{thm}{Theorem}[section]
\newtheorem{proposition}[thm]{Proposition}
\newtheorem{lemma}[thm]{Lemma} \newtheorem{corollary}[thm]{Corollary}
\theoremstyle{definition} 
 \newtheorem{example}[thm]{Example}

 \newtheorem{cond}[thm]{Condition}

\newtheorem{remarks}[thm]{Remarks}
 \newtheorem*{remarks*}{Remarks}
\newtheorem*{remark*}{Remark}

\numberwithin{equation}{section}

\title{Stationary scattering theory on manifolds, I\hspace{-.1em}I}

\thanks{ K.I. is supported by JSPS KAKENHI grant nr. 25800073.\, 
E.S. is supported by  DFF grant nr.  4181-00042.
}

\author{K. Ito}
\address[K. Ito]{Department of Mathematics, Kobe University\\
1-1 Rokkodai, Nada, Kobe, 657-8501, Japan}
\email{ito-ken@math.kobe-u.ac.jp}

\author{E. Skibsted} \address[E. Skibsted]{Institut for Matematiske
  Fag \\
  Aarhus Universitet\\ Ny Munkegade 8000 Aarhus C, Denmark}
\email{skibsted@imf.au.dk}

\begin{document}
\begin{abstract} 
Based on our previous study  \cite {IS3} 
we develop fully the stationary scattering theory for the 
Schr\"odinger operator on a manifold possessing an  escape
  function. A particular class of examples are manifolds  with
  Euclidean and/or hyperbolic ends, possibly with unbounded and
  non-smooth obstacles.
We develop  the theory largely along the classical
lines \cite{Sa, Co}  and derive in particular WKB-asymptotics of
  a class of minimal  generalized eigenfunctions. 
As an application we  prove  a conjecture of \cite{HPW} on
cross-ends transmissions 
 in its natural and
strong form within the  framework of our theory.
\end{abstract}

\maketitle
\tableofcontents

\section{Introduction}\label{sec:intr} 
Let $(M,g)$ be a connected Riemannian manifold.
In this paper we study  stationary scattering  theory for the geometric Schr\"odinger operator
\begin{align*}
H=H_0+V; \quad
H_0=-\tfrac12\Delta=\tfrac12p_i^*g^{ij}p_j,\ p_i=-\mathrm i\partial_i,
\end{align*}
on the Hilbert space ${\mathcal H}=L^2(M)$. The potential $V$ is
real-valued and bounded, and the self-adjointness of $H$ is realized
by the Dirichlet boundary condition. We shall develop a long-range stationary
scattering theory to a large   extent along the lines of \cite{Sa,
  Co} on $\R^d$ or  exterior domains of $\R^d$. In particular our
theory relies on the existence of an intrinsic
escape function. For other previous works on
scattering with unbounded obstacles we refer to  \cite{Il1,Il2,
  Il3}. For previous time-dependent short-range scattering theories on manifolds we
refer to \cite{IN,IS1}. We shall develop the time-dependent scattering
theory   from the stationary theory of this paper elsewhere
\cite{IS4}. 

Our main results are asymptotic completeness Theorem
\ref{thm:dist-four-transf}, i.e. the existence of
unitarily diagonalizing distorted Fourier transforms, and a
characterization of an associated class of minimal generalized
eigenfunctions in terms of (zeroth order) WKB-asymptotics  Theorem    \ref{thm:char-gener-eigenf-1}.  As an
application we prove a conjecture of \cite{HPW} on cross-ends
transmissions. It is stated in a strong form as   Corollary \ref{cor:transmission}. The results of
the paper are obtained in terms of an intrinsic escape function
geometrically controlled by parameters. At the border of our
parameter constraints we construct an example for which the minimal
generalized eigenfunctions do not have WKB-asymptotics  Example
\ref{ex:countex}. Whence 
our somewhat technical conditions are more natural than a first
reading  might indicate and in a  sense 
optimal.

\subsection{Setting and results from
  \cite{IS3}}\label{subsec:preliminary result}
Our  paper is a direct
continuation of  \cite{IS3}, and we start by recalling the setting and 
various results there 
partly to fix notation and terminologies.
This subsection exhibits only a minimal review,
and we refer to \cite[Subsection~1.1]{IS3} for 
more details and 
to \cite[Subsection~1.2]{IS3} for 
several examples of manifolds satisfying the abstract conditions appearing below.

\subsubsection{Basic setting}

We assume an \textit{end} structure on $M$ in a somewhat disguised form.\begin{cond}\label{cond:12.6.2.21.13}
Let $(M,g)$ be a connected Riemannian manifold  of dimension $d\ge 1$.
There exist a function $r\in C^\infty(M)$ with image $r(M)=[1,\infty)$
and constants 
$c>0$ and $r_0\ge 2$ such that:
\begin{enumerate}
\item\label{item:12.4.7.19.40} 
The gradient vector field $\omega=\mathop{\mathrm{grad}} r\in\mathfrak X(M)$ is
forward complete in the sense that the integral curve of $\omega$ is defined
for any initial point $x\in M$ and any non-negative time parameter $t\ge 0$.
\item\label{item:12.4.7.19.40b}
The bound $|\mathrm dr|=|\omega|\ge c$ holds on $\{x\in M\,|\, r(x)> r_0/2\}$. 
\end{enumerate}
\end{cond}

Under Condition~\ref{cond:12.6.2.21.13}
each component of the subset 
$E=\{x\in M\,|\, r(x)>r_0\}$
is called an \emph{end} of $M$,
and, along with Condition~\ref{cond:12.6.2.21.13a} below, 
the function $r$ may model a distance function there.
We note that by 
Condition~\ref{cond:12.6.2.21.13} (\ref{item:12.4.7.19.40b})
and the implicit function theorem  
the $r$-spheres 
\begin{align*}
S_R=\{x\in M\,|\,r(x)=R\};\quad R> r_0/2,
\end{align*} 
are submanifolds of $M$.
We will construct the \textit{spherical coordinates} on $E$
in Subsection \ref{subsec:Limiting Hilbert space}.

Let us impose more conditions on the geometry of $E$ in terms of the radius function $r$.
Choose $\chi\in C^\infty(\mathbb{R})$ such that 
\begin{align}
\chi(t)
=\left\{\begin{array}{ll}
1 &\mbox{ for } t \le 1, \\
0 &\mbox{ for } t \ge 2,
\end{array}
\right.
\quad
\chi\ge 0,
\quad
\chi'\le 0,\quad \sqrt{1-\chi}\in C^\infty,
\label{eq:14.1.7.23.24}
\end{align}
and set 
\begin{align}
\eta=1-\chi(2r/r_0),\quad
\tilde\eta=|\mathrm dr|^{-2}\eta 
=|\mathrm dr|^{-2}\bigl(1-\chi(2r/r_0)\bigr)
.
\label{eq:13.9.23.5.54b}
\end{align}
We introduce a  ``radial'' differential operator $A$:
\begin{align}
A
=\mathop{\mathrm{Re}} p^r
=\tfrac12\bigl( p^r+( p^r)^*\bigr);\quad 
p^r=-\mathrm i \nabla^r,\ 
\nabla^r=\nabla_{\omega}=g^{ij}(\nabla_i r)\nabla_j,
\label{eq:13.9.23.2.24}
\end{align}
and also the ``spherical'' tensor $\ell$ 
and the associated differential operator $L$: 
\begin{align}
\ell=g-\tilde\eta\,\mathrm dr\otimes \mathrm dr,\quad 
L=p_i^*\ell^{ij}p_j
.
\label{eq:13.9.23.5.54}
\end{align}
As we can see easily in the spherical coordinates
introduced in Subsection \ref{subsec:Limiting Hilbert space}, the tensor $\ell$
may be identified with the pull-back of $g$ to the $r$-spheres. We
call $L$ the spherical part of $-\Delta$. Note that if $|\d r|=1$ then
$-L$ acts as  the Laplace--Beltrami operator on $S_r$ (in general as a kind of
perturbation of this operator, see \eqref{eq:lsubr}).
 We remark that the tensor $\ell$ clearly
satisfies
\begin{align}
0\le \ell\le g,\quad
\ell^{\bullet i}(\nabla r)_i=(1-\eta)\mathrm dr,
\label{eq:14.12.23.13.49}
\end{align}
where the first bounds of \eqref{eq:14.12.23.13.49} are
understood as quadratic form estimates on the fibers
of the tangent bundle of $M$.
The quantities  of \eqref{eq:13.9.23.5.54} will play a
major role in this paper. 

Let us recall a local expression of the Levi--Civita connection $\nabla$:
If we denote the Christoffel symbol by 
$\Gamma^k_{ij}=\tfrac12g^{kl}(\partial_i g_{lj}+\partial_j g_{li}-\partial_lg_{ij})$,
then for any smooth function $f$ on $M$
\begin{align}
(\nabla f)_i&=(\nabla_if)=(\mathrm d f)_i=\partial_if,\quad
(\nabla^2 f)_{ij}
=\partial_i\partial_j f-\Gamma^k_{ij}\partial_kf.
\label{eq:14.3.5.0.14}
\end{align}
Note that $\nabla^2f $ is the geometric Hessian of $f$. 

\begin{cond}\label{cond:12.6.2.21.13a}
There exist constants $\tau,C>0$ such that 
globally on $M$
\begin{subequations}
\begin{align}
\begin{split}
|\nabla|\mathrm dr|^2|\leq Cr^{-1-\tau/2}
,\quad |\nabla^k r|\leq C\text{ for }\,k\in\{1,2\},\quad 
\bigl|\ell^{\bullet i}\nabla_i\Delta r\bigr|\leq Cr^{-1-\tau/2}.
\end{split}\label{eq:13.9.28.13.28}
\end{align} There exists $\sigma'>0$ such that for all   $R>r_0/2$, and as  quadratic forms on fibers of the tangent bundle of $S_R$, 
\begin{align}
R\,\iota^*_R\nabla^2r\ge \tfrac12\sigma' |\mathrm dr|^2\iota^*_R g,
\label{eq:13.9.5.3.30}
\end{align} where $\iota_R\colon S_R\hookrightarrow M$ 
is the inclusion map.
\end{subequations}
\end{cond}

We note that 
Condition~\ref{cond:12.6.2.21.13a} and the identity
\begin{align}
(\nabla^2r)^{ij}(\nabla r)_j=\tfrac12(\nabla|\mathrm dr|^2)^i
\label{eq:14.12.27.3.22}
\end{align} was used in \cite{IS3} to obtain the  more practical version
of \eqref{eq:13.9.5.3.30}: For any $\sigma\in(0,\sigma')$ and $\tau$
as in Condition~\ref{cond:12.6.2.21.13a}  there exists $C>0$
such that globally on $M$
\begin{align}
r\parb{\nabla^2r-\tfrac 12\tilde \eta^2(\nabla^r|\mathrm dr|^2)\d r
  \otimes\d r} \ge \tfrac12\sigma |\mathrm dr|^2 \ell- Cr^{-\tau}g.
\label{eq:13.9.5.3.30C}
\end{align}

Next we introduce an effective potential:
\begin{align}
q=V+\tfrac18\tilde\eta \bigl[(\Delta r)^2+2\nabla^r\Delta r\bigr].
\label{eq:14.12.10.22.48}
\end{align}
\begin{cond}\label{cond:10.6.1.16.24}
  There exists a splitting by real-valued functions:
  \begin{align*}
  q=q_1+q_2;\quad q_1\in C^1(M)\cap L^\infty(M),\  q_2\in L^\infty(M),
  \end{align*}
such that for some $\rho',C>0$ the following bounds hold globally on $M$:
\begin{align}\label{eq:60}
\nabla ^rq_1\leq Cr^{-1-\rho'},\quad
|q_2|\le Cr^{-1-\rho'}.
\end{align}
\end{cond}
We remark that in this  paper only derivatives of $r$ of order at most
five are used quantitatively.

Now let us mention the self-adjoint realizations of $H$ and $H_0$.
Since $(M,g)$ can be incomplete, the operators $H$ and $H_0$ are not necessarily 
essentially self-adjoint on $C^\infty_{\mathrm{c}}(M)$.
We realize $H_0$ as a self-adjoint operator by imposing the 
Dirichlet boundary condition, i.e.\ 
$H_0$ is the unique self-adjoint
operator associated with the closure of the quadratic form
\begin{align*}
\langle H_0\rangle_\psi=\langle \psi, -\tfrac12 \Delta \psi\rangle,\quad \psi
\in C^\infty_{\mathrm{c}}(M).
\end{align*}
We denote the form closure and the self-adjoint realization by the same symbol $H_0$. 
Define the associated Sobolev spaces $\mathcal H^s$ by
\begin{align}\mathcal H^s=(H_0+1)^{-s/2}{\mathcal H}, \quad s\in\mathbb{R}.
\label{eq:13.9.5.1.2}
\end{align}
Then $H_0$ may be understood as a
closed quadratic form on $Q(H_0)=\mathcal H^1$.
Equivalently, $H_0$ makes sense also as a bounded operator 
$\mathcal H^1\to\mathcal H^{-1}$, whose action coincides with 
that for distributions.  
By the definition of the Friedrichs extension 
the self-adjoint realization of $H_0$ 
is the restriction of such distributional
$H_0\colon \mathcal H^1\to\mathcal H^{-1}$ to the domain:
\begin{align*} 
{\mathcal D}(H_0)=\{\psi\in\mathcal H^1\,|\, H_0\psi\in {\mathcal H}\}\subseteq {\mathcal H}.
\end{align*}
Since $V$ is bounded and self-adjoint by Conditions~\ref{cond:12.6.2.21.13}--\ref{cond:10.6.1.16.24}, 
we can realize the self-adjoint operator $H=H_0+V$ simply as
\begin{align*}
H=H_0+V,\quad \mathcal D(H)=\mathcal D(H_0).
\end{align*}

In contrast to (\ref{eq:13.9.5.1.2}) we introduce the Hilbert spaces $\mathcal H_s$
and $\mathcal H_{s\pm}$ with configuration weights:
\begin{align*}
\mathcal H_s=r^{-s}{\mathcal H}, \quad 
{\mathcal H_{s+}=\bigcup_{s'>s}\mathcal H_{s'},\quad
\mathcal H_{s-}=\bigcap_{s'<s}\mathcal H_{s'},\quad}
s\in\mathbb{R}.
\end{align*}
We consider the $r$-balls $B_R=\{ r(x)<R\}$ and  
the characteristic functions 
\begin{align}
\begin{split}
F_\nu=F(B_{R_{\nu+1}}\setminus B_{R_\nu}),\ R_\nu=2^{\nu},\ \nu\ge 0,
\end{split}\label{eq:13.9.12.17.34}
\end{align}
 where $F(\Omega)$ is used for sharp characteristic
 function of a subset $\Omega\subseteq M$. 
Define the associated Besov spaces $B$ and $B^*$ by
\begin{align}
\begin{split}
B&=\{\psi\in L^2_{\mathrm{loc}}(M)\,|\,\|\psi\|_{B}<\infty\},\quad 
\|\psi\|_{B}=\sum_{\nu=0}^\infty R_\nu^{1/2}
\|F_\nu\psi\|_{{\mathcal H}},\\
B^*&=\{\psi\in L^2_{\mathrm{loc}}(M)\,|\, \|\psi\|_{B^*}<\infty\},\quad 
\|\psi\|_{B^*}=\sup_{\nu\ge 0}R_\nu^{-1/2}\|F_\nu\psi\|_{{\mathcal H}},
\end{split}\label{eq:13.9.7.15.11}
\end{align}
respectively.
We also define $B^*_0$ to be the closure of $C^\infty_{\mathrm c}(M)$ in $B^*$.
Recall the nesting:
\begin{align*}
\mathcal H_{1/2+}\subsetneq B\subsetneq \mathcal H_{1/2}
\subsetneq\vH
\subsetneq\mathcal H_{-1/2}\subsetneq B^*_0\subsetneq B^*\subsetneq \mathcal H_{-1/2-}.
\end{align*}

Using the function $\chi\in C^\infty(\mathbb R)$ of \eqref{eq:14.1.7.23.24},
define $\chi_n,\bar\chi_n,\chi_{m,n}\in C^\infty(M)$ for 
 $n>m\ge 0$ by
\begin{align}
\chi_n=\chi(r/R_n),\quad \bar\chi_n=1-\chi_n,\quad
\chi_{m,n}=\bar\chi_m\chi_n.
\label{eq:11.7.11.5.14}
\end{align}  
 Let us introduce an auxiliary space:
\begin{align*}
\mathcal N=\{\psi\in L^2_{\mathrm{loc}}(M)\,|\, \chi_n\psi\in \mathcal
H^1\mbox{ for all }n\ge 0\}.
\end{align*}
This is a  space of  functions that satisfy the
  Dirichlet boundary condition,
possibly with infinite $\mathcal H^1$-norm on $M$.
 Note that under Conditions~\ref{cond:12.6.2.21.13}--\ref{cond:10.6.1.16.24} the manifold $M$
  may be, e.g.\ a half-space in the Euclidean space, and there could be 
a ``boundary'' even for large $r$, which is ``invisible'' from inside $M$. 

\subsubsection{Review of the previous results}
Now we gather and review the main results of \cite{IS3}.
Note that all the theorems in this subsection are already proved there.

Our first theorem is Rellich's theorem, the absence of $B^*_0$-eigenfunctions 
with eigenvalues above a certain ``critical energy''
$\lambda_0\in\mathbb R$ defined by 
\begin{align}
\lambda_0
=\limsup_{r\to\infty}q_1
=
\lim_{R\to\infty}
\bigl(\sup\{q_1(x)\,|\, r(x)\ge R\}\bigr).
\label{eq:13.9.30.6.8}
\end{align}
For the Euclidean and the hyperbolic spaces and many other examples 
the critical energy $\lambda_0$ can be computed 
explicitly, and the essential
   spectrum is given by $\sigma_\ess(H)= [\lambda_0,\infty).$
  The latter is  usually seen in terms of Weyl sequences, see  \cite{K1}.

\begin{thm}\label{thm:13.6.20.0.10}
Suppose Conditions~\ref{cond:12.6.2.21.13}--\ref{cond:10.6.1.16.24},
and let $\lambda>\lambda_0$.
If a function  $\phi\in  L^2_{\mathrm{loc}}(M)$ satisfies that
\begin{enumerate}
\item\label{item:13.7.29.0.27}
$(H-\lambda)\phi=0$ in the distributional sense, 
\item\label{item:13.7.29.0.26}
$\bar\chi_m\phi\in \mathcal N \cap B_0^*$ for all $m\ge 0$ large enough,
\end{enumerate}
then $\phi=0$ in $M$.
\end{thm}

Next we discuss the limiting absorption principle
and the radiation condition
related to the resolvent $R(z)=(H-z)^{-1}.$ 
We state a locally uniform bound for the resolvent as
a map: $B\to B^*$.
For that we need a compactness condition.
\begin{cond}\label{cond:12.6.2.21.13b}
In addition to
Conditions~\ref{cond:12.6.2.21.13}--\ref{cond:10.6.1.16.24}, 
there exists an open subset $\mathcal I\subseteq (\lambda_0,\infty)$
such that for any $n\ge 0$ and  compact interval $I\subseteq \mathcal I$ 
the mapping 
\begin{align*}
  \chi_nP_H(I)\colon \mathcal H\to\mathcal H
\end{align*}
is  compact, where $P_H(I)$ denotes the spectral projection onto $I$ for $H$. 
\end{cond}

Due to Rellich's compact embedding theorem 
\cite[Theorem X\hspace{-.1em}I\hspace{-.1em}I\hspace{-.1em}I.65]{RS},
``boundedness'' of $r$-balls provides a criterion for Condition~\ref{cond:12.6.2.21.13b}:
If each $r$-ball $B_R$, $R\geq 1$, is 
isometric to a bounded subset of a complete manifold, 
Condition~\ref{cond:12.6.2.21.13b} is satisfied for $\mathcal I=(\lambda_0,\infty)$.
Condition~\ref{cond:12.6.2.21.13b} in fact includes more general situations 
where $M$ has multiple ends of different critical energies and $r$-balls are unbounded
as in \cite{K4}.

We fix any $\sigma\in (0,\sigma')$ and then 
large enough $C>0$ in agreement with \eqref{eq:13.9.5.3.30C}, 
and introduce the positive quadratic form
\begin{align*}
h:=\nabla^2r-\tfrac 12\tilde \eta^2(\nabla^r|\mathrm dr|^2)\d r
  \otimes\d r+2Cr^{-1-\tau}g\ge \tfrac12\sigma r^{-1}|\mathrm dr|^2 \ell+Cr^{-1-\tau}g.
\end{align*} 
For any  subset $I\subseteq \mathcal I$ we denote
\begin{align*}
I_\pm=\{z=\lambda\pm \mathrm i\Gamma\in \mathbb C\,|\,\lambda\in I,\ \Gamma\in (0,1)\},
\end{align*}
respectively. 
We also use the notation 
$\langle T\rangle_\phi=\langle\phi,T\phi\rangle$.

\begin{thm}\label{thm:12.7.2.7.9}
Suppose Condition~\ref{cond:12.6.2.21.13b}
and let $I\subseteq \mathcal I$ be a compact interval.
Then there exists $C>0$ such that 
for any $\phi=R(z)\psi$ with $z\in I_\pm$ and $\psi\in B$
\begin{align}
\|\phi\|_{B^*}+\|p^r\phi\|_{B^*}
+\langle p_i^*h^{ij}p_j\rangle_\phi^{1/2}
+\|H_0\phi\|_{B^*}
\le C\|\psi\|_B.
\label{eq:13.8.22.4.59c}
\end{align}
\end{thm}

In our theory the Besov boundedness \eqref{eq:13.8.22.4.59c} 
does not immediately imply the limiting absorption principle,
and for the latter  we need also radiation condition bounds implied by
minor 
additional regularity conditions.

\begin{cond}\label{cond:12.6.2.21.13bbb}
In addition to Condition~\ref{cond:12.6.2.21.13b} 
there exist splittings $q_1=q_{11}+q_{12}$ and 
$q_2=q_{21}+q_{22}$ 
by real-valued functions
\begin{align*}
q_{11}\in C^2(M)\cap L^\infty(M),\quad
q_{12}, 
q_{21}\in C^1(M)\cap L^\infty(M),\quad
q_{22}\in L^\infty(M)
\end{align*}
and constants $\rho,C>0$
such that for $\alpha=0,1$
\begin{align*}
%  \label{eq:pot}
|\nabla^rq_{11}|&\le Cr^{-(1+\rho/2)/2},&
|\ell^{\bullet i}\nabla_i q_{11}|&\le Cr^{-1-\rho/2},&
|\mathrm d\nabla^rq_{11}|&\le Cr^{-1-\rho/2},\\
|\mathrm dq_{12}|&\le Cr^{-1-\rho/2},&
|(\nabla^r)^\alpha q_{21}|&\le Cr^{-\alpha-\rho},&
q_{21}\nabla^r q_{11}&\leq Cr^{-1-\rho},\\
|q_{22}|&\le Cr^{-1-\rho/2}.&&
\end{align*}
\end{cond}

Our  radiation condition bounds are
stated in terms of the  distributional radial differential operator 
$A$ defined in \eqref{eq:13.9.23.2.24}
and an asymptotic complex phase $a$ given below.
Pick a
  smooth decreasing function $r_\lambda\geq 2r_0$ of
  $\lambda>\lambda_0$ such that
  \begin{align}
    \lambda+\lambda_0-2q_1\geq 0\mfor r\geq r_\lambda/2,
  \label{eq:14.6.29.23.46}
  \end{align}
  and that  $r_\lambda= r_0$ for all $\lambda$ large enough.
Then we set 
\begin{align*}
\eta_\lambda=1-\chi(2r/r_\lambda),
\end{align*}
and for $z=\lambda \pm \i \Gamma \in\mathcal I\cup\mathcal I_\pm$
\begin{subequations}
 \begin{align}
b&=\eta_\lambda|\mathrm dr|\sqrt{2(z-q_1)},\qquad\quad\ \  
\tilde b=\tilde\eta b,\label{eq:13.9.5.7.2300}\\
a&=b\pm\tfrac14\eta_\lambda (p^rq_{11})\big/(z-q_1),\quad 
\tilde a=\tilde\eta a,\label{eq:13.9.5.7.23}
\end{align} 
\end{subequations}
respectively, where the branch of square root is chosen such that 
$\mathop{\mathrm{Re}}\sqrt w>0$ for $w\in \mathbb C\setminus
(-\infty,0]$. Note that for $z\in \mathcal I$ there are two values of
$a$ (and similarly of course for $\tilde a$) which could be denoted $a_\pm$. For convenience we prefer to use
the shorter notation. 
Note also that the phase $a$ of \eqref{eq:13.9.5.7.23} 
is an approximate solution to the radial Riccati equation
\begin{align}
\pm p^ra+a^2-2|\mathrm dr|^2(z-q_1)=0
\label{eq:15.3.11.19.35}
\end{align} 
in the sense that it makes the quantity on the left-hand side of
\eqref{eq:15.3.11.19.35} small for large $r\ge 1$.  
The quantity $b$ of \eqref{eq:13.9.5.7.2300} alone already gives an
approximate solution to the same equation, however with the second
term of \eqref{eq:13.9.5.7.23} a better approximation is obtained, cf.\
Lemma~\ref{lem:13.9.2.7.18}. 
%The
%definition of $a$  in  \eqref{eq:13.9.5.7.23} differs only slightly from the one used in
%\cite{IS3}, and the result Theorem \ref{cond:12.6.2.21.13bbb} below is a
%direct consequence of the analogous result in \cite{IS3}.  
Set
\begin{align}
\beta_c
=\tfrac12\min\{\sigma,\tau,\rho\}.\label{eq:10c}
\end{align} Here and whenceforth we consider $\sigma \in (0,\sigma')$
as a fixed parameter. 
\begin{thm}
  \label{prop:radiation-conditions} 
  Suppose Condition~\ref{cond:12.6.2.21.13bbb}, and let $I\subseteq \mathcal I$
be a compact interval.
  Then for all  $\beta\in [0,\beta_c)$
  there exists $C>0$ such that 
  for any $\phi=R(z)\psi$ with $\psi\in r^{-\beta}B$ and $z\in I_\pm$
\begin{align}
\|r^\beta(A\mp a)\phi\|_{B^*}
+\langle p_i^*r^{2\beta} h^{ij}p_j\rangle_{\phi}^{1/2}
&\leq C\|r^\beta\psi\|_B,\label{eq:14cccCff}
\end{align} 
respectively.
\end{thm}

The limiting absorption principle reads.

\begin{corollary}\label{cor:12.7.2.7.9b}
Suppose Condition~\ref{cond:12.6.2.21.13bbb},
and let $I\subseteq \mathcal I$ be a compact interval.
For any $s>1/2$
and $\epsilon\in (0,\min\{(2s-1)/(2s+1),\beta_c/(\beta_c+1)\})$ 
there exists $C>0$ such that for $\alpha=0,1$ and any $z,z'\in I_+$ or $z,z'\in I_-$ 
\begin{align}
\|p^\alpha R(z)-p^\alpha R(z')\|_{\mathcal B(\mathcal H_s,\mathcal H_{-s})}\le C|z-z'|^\epsilon.
\label{eq:14.12.30.21.52}
\end{align}
In particular, the operators $p^\alpha R(z)$, $\alpha=0,1$,  attain
uniform limits as $I_\pm \ni z \to \lambda \in I$ in the norm topology of 
${\mathcal B}(\mathcal H_s,\mathcal H_{-s})$, say denoted  
\begin{align}
p^\alpha R(\lambda\pm\mathrm i0):=\lim_{I_\pm \ni z\to \lambda}p^\alpha R(z), \quad \lambda\in I,
\label{eq:14.12.30.21.53}
\end{align}
respectively. 
These limits $p^\alpha R(\lambda\pm\mathrm i0)\in{\mathcal B}(B,B^*)$,
and  $R(\lambda\pm\mathrm i0):B\to \vN\cap B^*$.
\end{corollary}

Given  the limiting resolvents $R(\lambda\pm\mathrm i0)$
  the radiation condition bounds for real spectral parameters
follow directly from Theorem~\ref{prop:radiation-conditions}.

\begin{corollary}
  \label{cor:radiation-conditions}
  Suppose Condition~\ref{cond:12.6.2.21.13bbb},  and let 
  $I\subseteq \mathcal I$ be a compact interval.
  Then for all $\beta \in [0,\beta_c)$
  there exists $C>0$ such that 
  for any $\phi=R(\lambda\pm\mathrm i0)\psi$ with $\psi\in r^{-\beta}B$ and 
  $\lambda\in I$ 
\begin{align}
\|r^\beta(A\mp a)\phi\|_{B^*}
+\langle p_i^*r^{2\beta} h^{ij}p_j\rangle_{\phi}^{1/2}
&\leq C\|r^\beta\psi\|_B,\label{eq:14cccCa} 
 \end{align} 
respectively.
\end{corollary}

For the Euclidean and the hyperbolic spaces without potential $V$ we
have $\beta_c \ge 1$. Hence in these cases the bound
\eqref{eq:14cccCa} hold for any $\beta\in [0,1)$.  We remark that for
the Euclidean space and a sufficiently regular potential the bound
\eqref{eq:14cccCa} is well-known for $\beta\in [0,1)$, cf.
\cite{Is,Sa,HS1}.  However in this case one can actually allow
$\beta\in [1,2)$, cf.  \cite{HS1}. If $\beta>1$ is allowed the
existence of the distorted Fourier transform follows easy, cf.
\cite{HS1, HS2, Sk}. This is demonstrated in Subsection \ref{subsec:easy
  case}.

As another application of the radiation condition bounds 
we have  characterized  the limiting resolvents $R(\lambda\pm\mathrm i0)$. 
For the Euclidean space such characterization
is usually referred to as the Sommerfeld uniqueness result, see
for example  \cite{Is}. 
\begin{corollary}\label{cor:13.9.9.8.23}
  Suppose Condition \ref{cond:12.6.2.21.13bbb}, and let
  $\lambda\in\mathcal I$, $\phi\in L^2_{\mathrm{loc}}(M)$ and
  $\psi\in r^{-\beta}B$ with  $\beta\in [0,\beta_c)$.
Then 
$\phi=R(\lambda\pm\mathrm i0)\psi$ holds if and only if
both of the following conditions hold:
\begin{enumerate}[(i)]
\item\label{item:13.7.29.0.29}
$(H-\lambda)\phi=\psi$ in the distributional sense.
\item\label{item:13.7.29.0.28}
$\phi\in \mathcal N\cap
r^\beta B^*$  and $(A\mp a)\phi\in r^{-\beta}B^*_0$.
\end{enumerate}
\end{corollary}

\subsection{Limiting Hilbert space}\label{subsec:Limiting Hilbert space}

To state the main results of the paper in Subsection~\ref{subsec:Distorted Fourier transform}
here we introduce the spherical coordinates and the limiting Hilbert space.

\subsubsection{Abstract construction}
Let us begin with an abstract theory.
We construct the spherical coordinates on $E$ under Condition~\ref{cond:12.6.2.21.13}
Using $\tilde \eta$ of \eqref{eq:13.9.23.5.54b}, 
define the \textit{normalized} gradient vector field $\tilde \omega\in \mathfrak X(M)$ by 
\begin{align*}
\tilde \omega=\tilde\eta\omega
.
\end{align*}
We let  
\begin{align}
\tilde y\colon 
\widetilde{\mathcal M}\to M,\ 
(t,x)\mapsto\tilde y(t,x)=\exp(t\tilde\omega)(x), 
\label{eq:14.12.10.6.32}
\end{align}
denote  the maximal flow generated by 
the vector field $\tilde \omega$ (whence by
Condition~\ref{cond:12.6.2.21.13}  $[0,\infty)\times M\subseteq
\widetilde{\mathcal M}$). 
By definition it satisfies, in local coordinates,
\begin{align*}
\partial_t\tilde y^i(t,x)=\tilde \omega^i(\tilde y(t,x))
=(\tilde \eta g^{ij}\nabla_jr)(\tilde y(t,x)),\quad 
\tilde y(0,x)=x.
\end{align*}
This implies  in particular that for  $r(x)\geq r_0$ and $t\ge 0$ 
\begin{align}
r(\tilde y(t,x))=r(x)+t,
\label{eq:13.9.21.16.44}
\end{align}
and hence the semigroup \eqref{eq:14.12.10.6.32} induces a family of 
diffeomorphic embeddings
\begin{align}
\iota_{R,R'}=\tilde y(R'-R,{}\cdot{})_{|S_R}\colon S_R\to S_{R'};\quad r_0\le R\le R',
\label{eq:13.9.21.18.9}
\end{align}
satisfying 
\begin{align}
\iota_{R',R''}\circ \iota_{R,R'}=\iota_{R,R''};\quad r_0\le R\le R'\le R''.
\label{eq:13.9.21.18.9b}
\end{align}
Through \eqref{eq:13.9.21.18.9} and \eqref{eq:13.9.21.18.9b}
we may regard $S_R\subseteq S_{R'}$ for any $R\le R'$ in a well-defined manner.
Such inclusions naturally induce a manifold structure on the union
\begin{align}
S=\bigcup_{R> r_0} S_R. 
\label{eq:13.9.21.18.9c}
\end{align}
In fact, the manifold $S$ can be  attained as an inductive limit, but we
do not get into  technical details since we are going to use a concrete
simple procedure (which is facilitated  by an additional condition).  The
manifold $S$ may in any case be
considered as a boundary of $M$ at infinity.  Let $\sigma$ be any
local coordinates on $S$.  We can define $\sigma(x)$ for $x\in E$ by
considering $x\in S_{r(x)}\subseteq S$.  Then the spherical
coordinates of a point $x\in E$, written slightly inconsistently, are
the components of $(r,\sigma)=(r(x),\sigma(x))\in (r_0,\infty)\times
S$.  We shall refer to $r$ as the \textit{radius function}, and
$S_R\subseteq S$ as the \textit{angular} or \textit{spherical
  manifolds}.  Note that in such coordinates $E$ is identified with an
open subset of the half-infinite cylinder $(r_0,\infty)\times S$ whose
$r$-sections are monotonically increasing and exhausting $S$.

Regarding $r\geq r_0$ just as a parameter and letting $\d \mathcal A_r$ 
be the naturally induced measure on $S_r$, we introduce the Hilbert space
  \begin{align}
  \mathcal G_r=L^2(S_r,\mathrm d\tilde{\mathcal A}_r);\quad 
  \d \tilde{\mathcal A}_r 
  =
  |\mathrm dr|^{-1} \d \mathcal A_r
  =(\det g)^{1/2}\,\mathrm d\sigma^2 \cdots  \mathrm d\sigma^d,
\label{eq:13.9.20.0.17}
  \end{align}
where the last equality holds in the spherical coordinates 
with any local coordinates for $S_r\subseteq S$. 
As for the measure $\mathrm d\tilde{\mathcal A}_r$ we note that the co-area formula (cf.\ \cite[Theorem
C.5]{Ev}) is valid for all integrable functions $\phi$ supported in $E$:
In the spherical coordinates
\begin{equation}
  \label{eq:co_area}
  \int_E \phi(x) \bigl(\det g(x)\bigr)^{1/2} \,\d x
=\int_{r_0}^\infty \d r \int_{S_r}\phi(r,\sigma)
  \,\d \tilde{\mathcal A}_r(\sigma).
\end{equation}
Noting \eqref {eq:co_area}, we can construct isometric embeddings
\begin{align}
i_{r,r'}\colon \mathcal G_r\to\mathcal G_{r'}\quad
\text{for }r_0\leq r\le r'
\label{eq:1600125}
\end{align}
as follows.
For any $\xi_r\in \mathcal G_r$ we define 
$i_{r,r'}\xi_r=\xi_{r'}\in
\mathcal G_{r'}$ by letting, in the spherical coordinates,
\begin{align} \label{eq:transm}
\begin{split}
\xi_{r'}(\sigma)
&=\left({\det g(r,\sigma)}\big/{\det g(r',\sigma)}\right)^{1/4}\xi_{r}(\sigma)
\quad\text{for }
(r,\sigma)\in E,
\end{split}
\end{align}
and $\xi_{r'}(\sigma)=0$ for $(r,\sigma)\notin E$.
 Indeed \eqref{eq:1600125} are isometric embeddings
satisfying 
\begin{align*}
i_{r',r''}\circ i_{r,r'}=i_{r,r''}\quad\text{for }r_0\le r\le r'\le r''.
\end{align*} 
Then in parallel to \eqref{eq:13.9.21.18.9c} 
we may regard $\mathcal G_r\subseteq\mathcal G_{r'}$ for $r_0\le r\le r'$,
and these inclusions naturally induce a pre-Hilbert space structure on the union
\begin{align*}
\mathcal G_\infty=\bigcup_{r> r_0}\mathcal G_r.
\end{align*}
We can define
the ``limiting Hilbert space'' $\vG$ as the completion  of $\mathcal G_\infty$.

We remark that if $\omega$ is also backward complete  
all the above embeddings are in fact equalities,
i.e.\ $S_r\cong S_{r'}$ and $\mathcal G_r\cong \mathcal G_{r'}$ for $r_0\le r\le r'$,
and we may construct the limiting objects just by letting 
$S=S_{r_0}$ and $\mathcal G=\mathcal G_{r_0}$.
This is a motivation for the following concrete construction.

\subsubsection{Concrete construction}

In this paper we impose an additional geometric condition  
under which the set $S$ and the associated limiting
Hilbert space $\mathcal G$ can be realized more concretely than  above. 
\begin{cond}\label{cond:altG}
  There exists an extended  Riemannian manifold $(M^{\rm ex}, g^{\rm
    ex})$ of dimension $d$ in
  which $(M, g)$ is isometrically embedded. The previous 
  Condition \ref{cond:12.6.2.21.13} is also fulfilled
  for $(M^{\rm ex}, g^{\rm ex})$ (with the same constants $c$ and $
  r_0$) possibly without the connectedness assumption. In addition
  the extended vector field, say denoted
  by $\omega^{\rm ex}$, is backward
  complete in $M^{\rm ex}$ (that is complete in $M^{\rm ex}$). 
\end{cond}

Under
Condition \ref{cond:altG} we define
\begin{align}\label{eq:extend}
  \begin{split}
  S^{\rm ex}(M)&=\bigl\{x\in S_{r_0}^{\rm ex}\,\big|\,\{\tilde y^{\rm
    ex}(t,x)\,|\,t\ge 0\}\cap M\neq\emptyset\bigr\},
\\\vG^{\rm ex}&= L^2(
  S^{\rm ex}(M),\d \tilde \vA^{\rm ex}_{r_0})\subseteq L^2(
  S_{r_0}^{\rm ex},\d \tilde \vA^{\rm ex}_{r_0}).  
  \end{split}
\end{align} This leads to the isometrical embedding $\vG_r\subseteq \vG^{\rm ex}$, $r\geq r_0$,
by mapping $\vG_r\ni \xi\to \xi^{\rm ex}\in \vG^{\rm ex}$ 
given in the spherical coordinates by
\begin{align}\label{eq:trans2}
  \xi^{\rm ex}(\sigma)=
\left({\det g^{\rm ex}(r,\sigma)}\big/{\det g^{\rm ex}(r_0,\sigma)}\right)^{1/4}
\xi(\sigma)\quad 
\text{for }(r,\sigma)\in S_r,
\end{align} and $\xi^{\rm ex}=0$ at other points in $S^{\rm ex}(M)$.

If $\omega$ is  forward and 
backward complete obviously we do not need extended objects and $\vG=L^2(
  S_{r_0},\d \tilde \vA_{r_0})$. 

To study scattering on $M$ it is convenient (although not necessary
for all of our results) to use $\vG^{\rm ex}$ and integrals of
extended orbits (as appearing above). Although we shall not elaborate
our methods should have the potential of  some similar results as in
this paper with different
conditions than Condition \ref{cond:altG}. For convenience we shall
in the paper  drop the superscript ``ex'' and write 
$\vG=\vG^{\rm ex}$, $\tilde y(t, x)=\tilde y^{\rm ex}(t, x)$, 
etc., whenever it follows from the context that these objects are
``extended''. The notation $S$ and $\d\tilde\vA$ will exclusively be used for
$S^{\rm ex}(M)$ and $\d\tilde \vA^{\rm ex}_{r_0}$,
respectively. Whence $\vG=L^2(S,\d \tilde \vA)$.

The above formulas
\eqref{eq:transm} and \eqref{eq:trans2} can be
understood in terms of translations on $\vH$ or on 
$\vH^\mathrm{ex}$.
 We introduce  \textit{normalized radial
  translations} $\tilde T(t)\colon {\mathcal H}\to{\mathcal H}$, $t\in
\mathbb R$, as follows.  Recall the notation of the normalized maximal
flow \eqref{eq:14.12.10.6.32}.  Then $\tilde T(t)\psi$,
$\psi\in\mathcal H$, is defined by
\begin{align}
\begin{split}
(\tilde T(t)\psi)(x)
&=\tilde J(t,x)^{1/2}
\left({\det g(\tilde y(t,x))}\big/{\det g(x)}\right)^{1/4}\psi(\tilde y(t,x))\\
&=\exp \left(\int_0^t\tfrac12(\mathop{\mathrm{div}}\tilde\omega)(\tilde y(s,x))\,\mathrm{d}s\right)\psi(\tilde y(t,x))
\end{split}
\label{eq:12.9.24.22.46cc}
\end{align}
if $(t,x)\in\widetilde{\mathcal M}$, and $(\tilde T(t)\psi)(x)
=0$ otherwise, 
where $\tilde J(t,{}\cdot{})$ is the Jacobian of the mapping
$\tilde y(t,{}\cdot{})\colon M\to M$
and $\mathop{\mathrm{div}}\tilde\omega=\mathop{\mathrm{tr}}\nabla\tilde\omega$.
 Note that  $\tilde J=1$ in the spherical coordinates, clearly showing
 a relationship to \eqref{eq:transm}. The well-definedness and the equivalence of the two expressions in \eqref{eq:12.9.24.22.46cc} 
can be verified similarly to the unnormalized flow in \cite{IS3}.
Here we only note that for any $\psi\in\mathcal H$ the former expression of 
\eqref{eq:12.9.24.22.46cc} and a change of variables imply 
\begin{align}
\|\tilde T(t)\psi\|
=\biggl(\int_{\tilde M(t)}|\psi(x)|^2\bigl(\det g(x)\bigr)^{1/2}\,\mathrm dx\biggr)^{1/2};\quad 
\tilde M(t)=\tilde y(\max\{t,0\},M).
\label{eq:15.2.7.1.22}
\end{align}
Hence the operators $\tilde T(t)=\e^{\i t \tilde A_+}$, $t\ge 0$, and the operators $\tilde T(-t)=\e^{-\i t \tilde A_-}$, $t\ge 0$,
form  strongly continuous one-parameter semigroups of surjective
partial isometries and  isometries, respectively. 
The operators 
 are the adjoints of each other in the sense that  
$\tilde T(t)^*=\tilde T(-t)$, however in general
 the family $(\tilde T(t))_{t\in \R}$ 
 does  not form a one-parameter group. This is in contrast to the
 similarly defined  quantity for  $ M^{\rm ex}$, say denoted $(\tilde T^{\rm
   ex}(t))_{t\in \R}=
(\e^{\i t \tilde A^{\rm ex}})_{t\in \R}$,  which indeed is a group on
$\vH^\mathrm{ex}$ with self-adjoint generator
\begin{align*}
\tilde A^\mathrm{ex}=\mathop{\mathrm{Re}}\parb{-\mathrm i
\nabla_{\tilde \omega^\mathrm{ex}}},
\end{align*}
where 
\begin{align*}
\quad\nabla_{\tilde \omega^\mathrm{ex}}
= ({\tilde \omega^\mathrm{ex}})^i\nabla_i,\quad 
\tilde \omega^\mathrm{ex}&=\tilde \eta^\mathrm{ex} (\nabla r^\mathrm{ex}),\quad \tilde \eta^\mathrm{ex}=\eta^\mathrm{ex}|\mathrm dr^\mathrm{ex}|^{-2}.
\end{align*}
 Note  the
 relationship $\tilde T(t)=1_M\tilde T^{\rm
   ex}(t)1_M$ for all $t\in \R$. We shall use the notation $\tilde A$
 as a generic notation for the generators $\tilde A_+, \tilde A_-$ and
 $\tilde A^\mathrm{ex}$ without distinction, 
if it is not confusing from the context.

\subsection{Main results}\label{subsec:Distorted
  Fourier transform}

\subsubsection{Distorted Fourier transform}

We need  additional assumptions. The following one
suffices for constructing the \textit{distorted Fourier transform}. 
\begin{cond}\label{cond:12.6.2.21.13c}
  Along with Condition~\ref{cond:altG},
  Condition~\ref{cond:12.6.2.21.13bbb} holds  with
  \begin{align}\label{eq:BAScond}
    2\beta_c=\min\{\sigma,\tau,\rho\}>1.
  \end{align}
The
  function $\tilde b=\tilde b(\lambda,x)$ has a real $C^1$-extension
  to $\vI \times M^{\rm ex}$, say denoted by ${\tilde b}^{\rm ex}$ (or
  by $\tilde b$ again for short). The following
  bound holds uniformly in $x\in E$ and locally uniformly in $\lambda\in\vI$: 
\begin{align}
   \sup_{r_0\leq \check r\leq r(x)}\bigg |\nabla'\int_{\check r-r(x)}^{0}\tilde b^{\rm ex}(\tilde y^{\rm ex}(t,x))\,\d t\bigg | \le 
C 
r(x)^{-1/2},
  \label{eq:13.9.23.16.41b2}
  \end{align} where $\nabla'=\ell^{\bullet i}\nabla_i$ denotes the covariant derivative  for the
 $r$-sphere
$S_{r(x)}$ (with induced  Riemannian metric).
\end{cond}
\begin{remarks*}
  If $M^{\rm ex}=M$ the technical bound \eqref{eq:13.9.23.16.41b2}
  follows from Lemma \ref{lemma:L_10b}. More generally we can \textit{verify}
  \eqref{eq:13.9.23.16.41b2} assuming Condition~\ref{cond:altG} and
  that various of the requirements in Conditions
  \ref{cond:12.6.2.21.13a} and \ref{cond:12.6.2.21.13bbb} for
  quantities on $M$ hold as well for the extended quantities (on
  $M^{\rm ex}$) with \eqref{eq:BAScond}. This follows from our proof
  of Lemma \ref{lemma:L_10b}. The
  bound is only used  in   the proof of Lemma \ref{lem:normBasi}, and we note
  that it is not needed if we impose the strengthening
  \eqref{eq:13.9.12.13.46bB} of \eqref{eq:BAScond} (however we do need
  it for the alternative Condition \ref{cond:12.6.2.21.13bb}
  \eqref{item:14.5.1.8.31}).
\end{remarks*}

For any $\psi\in\mathcal H_{1+}$ and $r\geq r_0$ 
we introduce a function $\xi(r)\in \vG $ using
the mapping \eqref{eq:trans2} (and omitting  the superscript ``ex'')
and noting the expression in \eqref{eq:12.9.24.22.46cc}.
We let 
\begin{align}\label{eq:defRes}
  \xi(r)(\sigma)= \exp \biggl(\int_{r_0}^{r}\parbb{\mp \i\tilde
    b+ \tfrac
    12\mathop{\mathrm{div}}\tilde\omega}(s,\sigma)\,\mathrm{d}s\biggr)
[\sqrt b R(\lambda\pm\mathrm
  i0)\psi](r,\sigma),
\end{align} or, alternatively, 
\begin{align}
  \xi(r)=\e^{\i
  (r-r_0)(\tilde A^{\rm ex }\mp\tilde
      b^{\rm ex })}\bigl [\sqrt b R(\lambda\pm\mathrm
  i0)\psi\bigr
    ]_{|S_r}=\e^{\i
  (r-r_0)(\tilde A\mp\tilde
      b)}\bigl [\sqrt b R(\lambda\pm\mathrm
  i0)\psi\bigr
    ]_{|S_r}.
\label{eq:160126}
\end{align}
Then we would like to define the ``distorted Fourier transform'' by
\begin{align}\label{eq:disF}
  F^\pm(\lambda)\psi =\vGlim_{r\to \infty} \xi(r);\quad 
\psi\in\mathcal H_{1+}.
\end{align}
 By definition the function  $F^\pm(\lambda)\psi\in \vG=L^2(S,\d \tilde \vA)$, and we
note that  our  construction of
$F^\pm(\lambda)\psi$ is 
non-canonical   primarily due to the freedom in choosing $\vG$. In fact for
$M^{\rm ex}=M$ the only non-canonical feature comes from  the   dependence
of $r_0$ (determining  $\vG$ in that case), while in general there is an additional freedom in choosing
extended functions.

 Of course we need to justify the definition \eqref{eq:disF}.

  \begin{thm}
\label{thm:strong} 
Suppose Condition~\ref{cond:12.6.2.21.13c}.  Then for any $\psi\in
{\mathcal H_{1+}}$  there exist the limits
\eqref{eq:disF}. The maps $\mathcal
I\ni\lambda\mapsto F^\pm(\lambda)\psi\in \mathcal G$ are
   continuous. Moreover
 the identities
  \begin{align}
  \label{eq:fund}
  \|F^\pm(\lambda)\psi\|^2=2\pi\inp{\psi, \delta(H-\lambda)\psi};\quad
\delta(H-\lambda):=\pi^{-1}\mathop{\mathrm{Im}}R(\lambda+\i 0),
\end{align}
hold. 
  \end{thm}

  Due to \eqref{eq:fund} the operators $F^\pm(\lambda)$ extend as
  continuous operators $B\to \vG$, and for any $\psi\in B$ the maps
  $F^\pm({}\cdot{})\psi\in \mathcal G$ are continuous. In Proposition
  \ref{prop:dist-four-transf}  stated below we give a formula for these
  extensions.  

Introduce 
\begin{align*}
\mathcal H_{\mathcal I}=P_H(\mathcal I)\mathcal H,\quad 
  \widetilde \vH_{\mathcal I} =L^2(\mathcal I, (2\pi)^{-1}\d \lambda;\vG),
\end{align*} 
set $H_\mathcal I=H P_H(\mathcal I)$ and let $M_\lambda$ be the operator of multiplication by $\lambda$ on
$\widetilde\vH_{\mathcal I}$.
We define
\begin{align*}
  F^\pm=\int_{\mathcal I} \oplus F^\pm(\lambda)\,\d \lambda\colon
  B\to C(\mathcal I;\mathcal G).
\end{align*} 
These  operators can be extended to proper spaces which is stated as
the first  part of
the following result.
\begin{proposition}
  \label{prop:dist-four-transf} 
  Suppose Condition~\ref{cond:12.6.2.21.13c}. 
The operators   $F^\pm$ considered as  maps $B\cap\mathcal
H_{\mathcal I}\to \widetilde{\mathcal H}_{\mathcal I}$ extend uniquely
to  
isometries $\mathcal H_{\mathcal
  I}\to \widetilde{\mathcal H}_{\mathcal I}$. These  extensions obey $F^\pm H_\mathcal I\subseteq
M_\lambda F^\pm$.
Moreover  for  any $\psi\in B$ the vectors $F^\pm(\lambda)\psi $ are given as 
averaged limits. More precisely  introducing for any such $\psi$ the integral
$\intR \xi(r)\,\d r:=R^{-1}\int_{R}^{2R} \xi(r)
\,\d r$, these  vectors  are given as
\begin{align}\label{eq:extbF2}
\begin{split} &F^\pm(\lambda)\psi =\vGlim_{R\to
    \infty} {-\!\!\!\!\!\!\int_R} \xi(r)\,\d r
\\&=\vGlim_{R\to
    \infty}-\!\!\!\!\!\!\int_R \exp \biggl( \int_{r_0}^{r}\parbb{\mp \i\tilde
    b+ \tfrac
    12\mathop{\mathrm{div}}\tilde\omega}(s,{}\cdot{})\,\mathrm{d}s\biggr)
[\sqrt b R(\lambda\pm\mathrm
  i0)\psi](r,{}\cdot{})
\,\d r,
\end{split}
\end{align} and the limits
\eqref{eq:extbF2} are attained locally
uniformly in $\lambda\in\mathcal I$.
\end{proposition}

The above extended isometries 
$F^\pm\colon \mathcal H_{\mathcal I}\to\widetilde{\mathcal H}_{\mathcal I}$
are actually unitary under an additional condition,
and for  this  reason 
we call them the Fourier transforms associated with $H_{\mathcal I}$. 
The new condition consists of two alternatives. The first one is a partial 
  strengthening of Condition~\ref{cond:12.6.2.21.13c}. The other one
  is primarily a set of
   bounds on  higher order derivatives of various quantities defined on $M$.
\begin{cond}\label{cond:12.6.2.21.13bb}
In addition to  Condition~\ref{cond:12.6.2.21.13c} one of the
following properties holds:
\begin{enumerate}
\item\label{item:14.5.1.8.30} 
\begin{align}
\min\{\sigma,\tau,\rho\}>2.\label{eq:13.9.12.13.46bB}
\end{align}
\item\label{item:14.5.1.8.31}
The extension ${\tilde b}^{\rm ex}$ of
Condition \ref{cond:12.6.2.21.13c} is in $C^2$. The restriction $q_1{}_{|S_r}$ belongs to $C^2(S_r)$ for  $r\geq r_0$,  and 
 there exists  $C>0$ such that
\begin{subequations}
\begin{align}
 \bigl|\iota^*_R \nabla^3r\bigr|&\le
CR^{{-1-\tau/2}}\text{ for }  R\geq r_0,\label{eq:smootrA}\\
|\nabla'{}^2 q_1{}_{|S_r}|&\le Cr^{-1-\rho} \text{ for }    r\geq r_0,\label{eq:smootV}
\end{align}
and 
\begin{align}\label{eq:smootr}
  \begin{split}
      \bigl|\nabla'{}^2|\mathrm dr|^2_{|S_r}\bigr| &\leq Cr^{-1-\tau},\quad 
  \bigl|\nabla'{}^2(\nabla_{\omega}|\mathrm dr|^2)_{|S_r}\bigr|\le
  Cr^{-1-\tau},\\
  \bigl|\nabla'{}^2(\Delta r)_{|S_r}\bigr|&\le Cr^{-1-\tau} \text{ for }    r\geq r_0, 
  \end{split}
\end{align}
\end{subequations}
where $\nabla'$ denotes the Levi--Civita connection associated with the
induced Riemannian metric $\iota^*_rg$ on the $r$-sphere
$S_{r}$. 
\end{enumerate}
\end{cond}
We remark that for any $f\in C^\infty(M)$ the Hessian
$\nabla'{}^2(f_{|S_R})= \iota^*_R t_f$ where $t_f=\nabla^2f-(\tilde
\omega ^j\partial _j f)\nabla^2r$. The bounds \eqref{eq:smootV} and
\eqref{eq:smootr}  allow us to estimate  $\nabla'{}^2 \bigl(\pm\mathrm i\tilde b-\tfrac1{2}
\mathop{\mathrm{div}}\tilde \omega
\bigr)=O(r^{-1-\min\{\tau\,\rho\}})$ (to be used 
in our verification of \eqref{eq:14.5.1.16.14}).

\begin{thm}
  \label{thm:dist-four-transf} 
Suppose Condition~\ref{cond:12.6.2.21.13bb}. 
Then  the operators  
$F^\pm\colon \mathcal H_{\mathcal I}\to \widetilde{\mathcal H}_{\mathcal I}$
are unitarily diagonalizing transforms for $H_{\mathcal I}$,
that is, they are unitary and
 \begin{align*}
%  \label{eq:diagBb}
   F^\pm H_{\mathcal I}=M_\lambda F^\pm,
\end{align*}
respectively.
\end{thm}

\begin{remark*}
For the conclusion of Theorem~\ref{thm:dist-four-transf}
it suffices to assume Condition~\ref{cond:12.6.2.21.13c}  and 
\eqref{eq:14.5.1.16.14} of Lemma \ref{lem:14.5.1.16.30}. 
In fact, 
Condition~\ref{cond:12.6.2.21.13bb} is   here and henceforth used only for the verification of \eqref{eq:14.5.1.16.14}. 
\end{remark*}

\subsubsection{Scattering matrix and generalized eigenfunctions}
 Next for any $\xi\in \mathcal G$ let us
introduce purely outgoing/incoming  approximate  generalized eigenfunctions  $\phi^\pm[\xi]\in
B^*$ by, using the spherical coordinates, 
\begin{align}\label{eq:gen1B}
\begin{split}
   \phi^\pm[\xi](r,\sigma)&= \eta_\lambda
  \bigl[2|\mathrm dr|^2 (\lambda-q_1)\bigr]^{-1/4} 
\\&\phantom{{}={}}
\cdot \exp \biggl(\int_{r_0}^{r}\parbb{\pm \i\tilde
    b- \tfrac
    12\mathop{\mathrm{div}}\tilde\omega}(s,\sigma)
\,\mathrm{d}s\biggr)\xi(\sigma),
\end{split}
\end{align}  
cf.\ Theorem~\ref{thm:strong}.  We remark that formulas like
\eqref{eq:gen1B} in the context of Schr\"odinger operators are
referred to as (zeroth order) WKB-approximations. If we denote the oscillatory part
of the phase by $S(x)=\pm \int_{r_0}^{r(x)}\tilde
b(s,\sigma(x))\,\mathrm{d}s$ then
Condition \ref{cond:12.6.2.21.13bbb} and   \eqref{eq:BAScond}  of
Condition \ref{cond:12.6.2.21.13c} 
 imply that for  $M^{\rm ex}=M$
\begin{align*}
  \tfrac 12|\d S(x)|^2+q_1-\lambda=O(r^{-1-\epsilon})\text{ for all } \epsilon< 2\beta_c-1,
\end{align*} see Lemma \ref{lemma:L_10b}. Whence in this case
$S(\cdot)$ is an approximate solution to the eikonal equation with the
effective potential $q_1$ and a short-range error. In the general case
of Condition \ref{cond:12.6.2.21.13c} the bound
\eqref{eq:13.9.23.16.41b2} is barely too weak to give a uniform
short-range error, however due to Lemma \ref{lemma:L_10b} we still
have pointwise short-range bounds (i.e. short-range bounds that are
not uniform in $\sigma\in S$). Although such property is basic for the
WKB-method (in particular for obtaining higher order expansions) it
will only be used in a disguised form in this paper.  We remark that
under Condition \ref{cond:12.6.2.21.13c} for any $\xi\in C_\c^\infty
(S)\subseteq\vG$ the vectors $\phi^\pm[\xi]\in \vN$ (here possibly
needed cutoff further at infinity), and under
Condition~\ref{cond:12.6.2.21.13bb} they are approximate generalized
eigenfunctions in the sense $R(\i)(H-\lambda)\phi^\pm[\xi] \in B \cap
\vH^1$ which is a consequence of \eqref{eq:14.5.1.16.14} (cf. the
proof of Lemma \ref{lem:14.5.4.17.5}).
      \begin{example}\label{ex:countex} Consider a  subset
        $M\subseteq \R^2$ equipped with the Euclidean metric and given
        with an end bounded by the ``interior'' of a parabola, say
        $x^2<y$, and $r^2:=x^2/2+y^2>r^2_0$.  We consider only
        $V=0$. The orbits of $\omega=\grad r$ are the branches of
        parabolas $c y^{1/2}=x$ where $-1<c<1$, and Condition
        \ref{cond:12.6.2.21.13bbb} is fulfilled with $\sigma=\tau=1$
        and $\rho=2$ (and similarly for Condition \ref
        {cond:12.6.2.21.13bb} (\ref{item:14.5.1.8.31})), in particular
        $\beta_c=1/2$ is fulfilled. However the barely stronger condition
        \eqref{eq:BAScond} is not fulfilled and whence the example is not covered by
        the theory of this paper (in contrast to \cite{IS3}). Moreover 
        we can in fact show that the generalized eigenfunctions in $\mathcal N
        \cap B^*$ are {\it not} of WKB-type as in the theorem stated
        below, see Subsection \ref{subsec:Parabolic example}. Let us
        here note, as an indication of this result, that for any
        $0\neq \xi\in C_\c^\infty (S)\subseteq\vG$
      \begin{align*}
        (H-\lambda)\phi^\pm[\xi]
      \in \vH_{ 1/2-}\setminus B, 
      \end{align*}
which technically prohibits us to construct  WKB-solutions. 
      \end{example}

For some examples for which our theory  applies we refer the reader to \cite[Subsection 1.2]{IS3}.

Under Condition~\ref{cond:12.6.2.21.13bb} and for any $\lambda\in\mathcal I$ the {\it scattering matrix} 
$S(\lambda)\colon \mathcal G\to\mathcal G$ is defined by the
identity 
\begin{align}
    \label{eq:1S} F^+(\lambda)\psi=S(\lambda)F^-(\lambda)\psi;\quad\psi\in  B.
\end{align}
It follows from \eqref{eq:rep3}  that 
$C^\infty_{\mathrm c}(S)\subseteq \mathop{\mathrm{Ran}}
F^\pm(\lambda)$, and hence, with
Theorem \ref{thm:strong}, Proposition~\ref{prop:dist-four-transf} and a density argument,
$S(\cdot)$ is a well-defined strongly continuous unitary operator.  We obtain a
 characterization of the generalized eigenfunctions in $\mathcal N
 \cap  B^*$, i.e. the elements of 
\begin{equation*}
    \vE_\lambda:=\{\phi\in\mathcal N \cap  B^*\,|\, (H-\lambda)\phi=0\}.
  \end{equation*} Due to Theorem \ref{thm:13.6.20.0.10} these
  eigenfunctions  may be called \emph{minimal}.

\begin{thm}
  \label{thm:char-gener-eigenf-1}
Suppose Condition~\ref{cond:12.6.2.21.13bb}.
Then for any $\lambda \in\mathcal I$  the following assertions hold.
\begin{subequations}
  \begin{enumerate}[(i)]
  \item\label{item:14.5.13.5.40} For any one of $\xi_\pm \in \vG$ or $\phi\in \vE_\lambda$ the
    two  other  quantities in $\{\xi_-,\xi_+,  \phi\}$ uniquely exist
    such that
    \begin{align}\label{eq:gen1}
      \phi -\phi^+[\xi_+]+\phi^-[\xi_-]\in
      B_0^*.
    \end{align}

  \item \label{item:14.5.13.5.41} The  correspondences  in  \eqref{eq:gen1} are  given  by  the
    formulas (recall \eqref{eq:160126})
    \begin{align}\label{eq:aEigenfw}
     \phi&=\i F^\pm(\lambda)^*\xi_\pm,\quad \xi_+=S(\lambda)\xi_-,\\
\xi_\pm&=2^{-1} \vGlim_{R\to \infty}
-\!\!\!\!\!\!\int_R \e^{\i(r-r_0)(\tilde A^{\rm ex }\mp\tilde
      b^{\rm ex })}
\bigl [{b^{-1/2}}(  A\pm b)\phi\bigr]_{|S_r}\,\d r.\label{eq:aEigenfwB}
    \end{align}
    In  particular  the  wave  matrices  $F^\pm(\lambda)^*\colon\vG\to
    \vE_\lambda$ are linear isomorphisms.

  \item\label{item:14.5.13.5.42}  The wave matrices
    $F^\pm(\lambda)^*\colon\vG\to \vE_\lambda\,(\subseteq B^*)$
    are bi-continuous. In fact
    \begin{align}\label{eq:aEigenf2w}
      2\|\xi_\pm\|_{\vG}^2=\lim_{R\to \infty}R^{-1}\int_{B_{2R}\setminus B_{R}}
      |{b}^{1/2}\phi|^2\,(\det g)^{1/2}\d x.
    \end{align}

  \item\label{item:14.5.14.4.17}   The   operators    $F^\pm(\lambda)\colon   B\to   \vG$   and
    $\delta(H-\lambda)\colon B\to \vE_\lambda$ are onto.
  \end{enumerate}
   \end{subequations}

\end{thm}

We remark  that parts of this theorem overlap with \cite{ACH, AH, Co,
  GY,Me,Va}. 

Finally we give  an  application of our results to 
channel scattering theory addressed,  but
treated very differently,  in \cite{HPW}.
Suppose $M^\mathrm{ex}$ has $N\geq 2$ number of ends, i.e.\  
$E^\mathrm{ex}=\{x\in M^\mathrm{ex}\,|\, r^\mathrm{ex}(x)>r_0\}$ has $N\geq 2$ components $E_i$, $i=1,\dots,N$.
Then the Hilbert space $\mathcal G$ splits as 
\begin{align*}
\mathcal G=\mathcal G_1\oplus\dots\oplus\mathcal G_N;\quad \mathcal G_i=L^2(S_i),\ 
S_i=S\cap \overline{E_i},
\end{align*}
and, accordingly, the scattering matrix $S(\lambda)$ has a matrix representation
\begin{align*}
S(\lambda)=(S_{ij}(\lambda))_{1\leq i,j\leq N},\quad
S_{ij}(\lambda)\in \mathcal B(\mathcal G_j,\mathcal G_i).
\end{align*}
\begin{corollary}\label{cor:transmission}
   Suppose  under Condition~\ref{cond:12.6.2.21.13bb} that 
 $E^\mathrm{ex}$ has $N$ number of ends.  Decomposing      as above  for any $\lambda
 \in\mathcal I$ the scattering matrix $S(\lambda)$ into components the
 off-diagonal ones, $S_{ij}(\lambda)$ with  $i\neq j$, are
  one-to-one mappings.
\end{corollary}
\begin{proof}
 If  $\xi_-=(\xi_-^1, \dots, \xi^N_-)\in\vG$ is given with
$\xi^j_-=0$ for $j\neq 2$ and $\xi_+^1=0$ then $\phi=\i
F^+(\lambda)^*\xi_-$ obeys that $1_{E_1\cap M}\phi\in B^*_0$. By using a suitable cutoff of
the function $r$ (essentially defined by making it vanish in
$E_j\cap M$ for $j\geq 2$) we then obtain from Theorem \ref{thm:13.6.20.0.10} that
$\phi=0$. For example we could redefine $r$ and $r_0$ of Condition
\ref{cond:12.6.2.21.13} as follows
(using the notation \eqref{eq:14.1.7.23.24}):
First replace  $r$ by the function $r1_{E_1\cap M}\parb{1-
  \chi({r}/{r_0})}$ and  then replace 
 the parameter $r_0$ by $4 r_0$. With these modifications Conditions
\ref{cond:12.6.2.21.13}--\ref{cond:10.6.1.16.24} are fulfilled (with
 the other parameters there unchanged), and therefore indeed Theorem \ref{thm:13.6.20.0.10}
applies. In particular we deduce  that
$\xi^2_-=0$, showing that $\ker S_{12}(\lambda)=\{0\}$. We
can argue in the same way for all other off-diagonal components of the scattering matrix.
\end{proof}

We note that Corollary \ref{cor:transmission} may be seen as a
 stationary solution to conjectures of \cite{HPW}, see \cite[Remark
5.7]{HPW}. We shall develop the time-dependent version of our results
in \cite{IS4}. In particular this includes  a time-dependent version of
Corollary \ref{cor:transmission} directly proving conjectures of \cite{HPW} in
a strong form.

\section{Preliminaries}

\subsection{Elementary tensor analysis}\label{sec:12.10.11.16.40}
Here we fix our convention for the covariant derivatives.
We formulate and use them always in local expressions, 
but for a coordinate-independent representation, see \cite[p. 34]{Chavel}.

\subsubsection{Derivatives of functions}

We shall denote two tensors by the same symbol if they are related to  each other 
through the canonical identification $TM\cong T^*M$,
and distinguish them by super- and subscripts.
We denote $TM\cong T^*M$ by $T$ for short, and set $T^p=T^{\otimes p}$.
The covariant derivative $\nabla$ acts as 
a linear operator $\Gamma(T^p)\to \Gamma(T^{p+1})$ 
and is  defined for $t\in \Gamma(T^p)$
by 
\begin{align}
(\nabla t)_{ji_1\cdots i_p}
&=\nabla_jt_{i_1\cdots i_p}
=\partial_{j} t_{i_1\cdots i_p}
-\sum_{s=1}^p \Gamma^{k}_{ji_s}t_{i_1\cdots k\cdots  i_p}.
\label{eq:12.9.27.1.20}
\end{align}
Here $\Gamma^k_{ij}=\tfrac12g^{kl}(\partial_i g_{lj}+\partial_j g_{li}-\partial_lg_{ij})$ 
is the Christoffel symbol and $t$ is considered as a section of the
$p$-fold cotangent bundle, and we adopt the convention that a new subscript is always added to the left as 
in \eqref{eq:12.9.27.1.20}.
By the identification $TM\cong T^*M$ it suffices to discuss an expression 
only for the subscripts.
In fact, we have the compatibility condition
\begin{align}
\nabla_i g_{jk}=\partial_ig_{jk}-\Gamma_{ij}^lg_{lk}-\Gamma_{ik}^lg_{jl}=0,
\label{eq:12.9.27.1.21}
\end{align}
and then by \eqref{eq:12.9.27.1.20} and \eqref{eq:12.9.27.1.21}
the covariant derivative can be computed for the tensors of any type. 
For example, for $t\in \Gamma(T)=\Gamma(T^1)$
\begin{align}\label{eq:covtan}
  \begin{split}
  (\nabla t)_j{}^i
&=g^{ik}(\nabla t)_{jk}
=g^{ik}\bigl(\partial_j t_k-\Gamma_{jk}^l t_l\bigr)
\\&=g^{ik}\bigl(\partial_j g_{kl}t^l-\Gamma_{jk}^l g_{lm}t^m\bigr)
=\partial_j t^i+\Gamma_{jk}^it^k,  
  \end{split}
\end{align}
and this extends to the general case with ease.
The covariant derivative acts as a derivation with respect to tensor product,
i.e.\ for $t\in \Gamma(T^p)$ and $u\in \Gamma(T^q)$
\begin{align}\label{eq:11.3.22.6.18}
(\nabla (t\otimes u))_{ji_1\cdots i_{p+q}}
=(\nabla t)_{ji_1\cdots i_p}
u_{i_{p+1}\cdots i_{p+q}}
+
t_{i_1\cdots i_p}
(\nabla u)_{ji_{p+1}\cdots i_{p+q}}.
\end{align}
The formal adjoint $\nabla^*\colon\Gamma(T^{p+1})\to\Gamma(T^p)$ is 
defined to satisfy
\begin{align*}
\int 
\overline{u_{ji_1\cdots i_p}}
(\nabla t)^{ji_1\cdots i_p}
(\det g)^{1/2}\,\mathrm{d}x
=\int \overline{(\nabla^* u)_{i_1\cdots i_p}}
t^{i_1\cdots i_p}(\det g)^{1/2}\,\mathrm{d}x
\end{align*}
for $u\in \Gamma(T^{p+1})$ and $t\in \Gamma(T^p)$ 
 compactly  supported in a coordinate neighbourhood.
Actually we can write it in a divergence form: For $u\in \Gamma(T^{p+1})$
\begin{align*}
(\nabla^* u)_{i_1\cdots i_p}
=-(\mathop{\mathrm{div}} u)_{i_1\cdots i_p}
=-(\nabla u)_{j}{}^j{}_{i_1\cdots i_p}
=-g^{jk}(\nabla u)_{jki_1\cdots i_p}.
\end{align*}

Finally let us give several remarks.
It is clear that for any function $f\in \Gamma(T^0)=C^\infty (M)$
the second covariant derivative $\nabla^2f=\nabla\nabla f$ is symmetric, i.e.\ 
\begin{align}
\begin{split}
(\nabla^2 f)_{ij}=(\nabla^2 f)_{ji}=\partial_i\partial_jf-\Gamma_{ij}^k\partial_kf,
\end{split}
\label{eq:11.3.22.6.20}
\end{align}
and we have expressions for the Laplace--Beltrami operator $\Delta$:
\begin{align*}
\Delta f=(\nabla^2f)_i{}^i=g^{ij}(\nabla^2f)_{ij}=\mathop{\mathrm{tr}}\nabla^2f
=\mathop{\mathrm{div}}\nabla f.
\end{align*} 
We  note that covariant differentiation and
  contraction are commuting operations. Whence we have, for example,
  for $t\in\Gamma(T)$ and  $u\in\Gamma(T^{p+1})$
  \begin{align}\label{eq:contract}
      \begin{split}
       \nabla_k t^j u_{ji_1\cdots i_p}
	   &=(\nabla t)_k{}^ju_{ji_1\cdots i_p}
	   +t^j (\nabla u)_{kji_1\cdots i_p},\\
       \nabla_j(\nabla t)_{i}{}^i&=(\nabla^2 t)_{ji}{}^{i}=g^{ik}(\nabla^2 t)_{jik} .
\end{split}
  \end{align}

\subsubsection{Derivatives of mappings}
Next let us present a short description of the derivatives of a mapping
(not of a function).
Let $y\colon M\to N$ be a general mapping
from a Riemannian manifold $(M,g)$ to another $(N,h)$. 
In geometric literatures the $k$-th derivatives $\nabla^k y$, $k=1,2,\dots$, are 
defined to satisfy the ``chain rule''. For instance, 
the derivatives $\nabla y$ and $\nabla^2y$ are required 
to satisfy in local coordinates that for any function $f\in C^\infty(N)$
\begin{align*}
[\nabla(f(y))]_i&=(\nabla y)^\alpha{}_i(\nabla f)_\alpha (y) ,\\
[\nabla^2(f(y))]_{ij}
&=(\nabla^2 y)^\alpha{}_{ij}(\nabla f)_\alpha (y)
+(\nabla y)^\alpha{}_i(\nabla y)^\beta{}_j(\nabla^2 f)_{\alpha\beta} (y).
\end{align*}
Here we used the Roman and the Greek alphabets to denote 
the indices of coordinates $x\in M$ and $y=y(x)\in N$, respectively.
Although we are not going to verify this,
the above definition is indeed well-justified, 
and we have the following local expressions for such derivatives:
\begin{align}
(\nabla y)^\alpha{}_i
&=\partial_iy^\alpha,\quad
(\nabla^2y)^\alpha{}_{ij}
=\partial_i\partial_j y^\alpha
-\Gamma_{ij}^k\partial_k y^\alpha
+\Gamma_{\beta\gamma}^\alpha
(\partial_i y^\beta)
(\partial_j y^\gamma).
\label{eq:160220}
\end{align}
Note that we adopted the same convention on the Roman and Greek indices as above:
In particular, $\Gamma^k_{ij}$ and $\Gamma^\alpha_{\beta\gamma}$
denote the Christoffel symbols for $(M,g)$ and $(N,g)$, respectively.

\subsection{Decomposition of Hamiltonian}
\label{subsec:Decomposition of Hamiltonian} 
Throughout the remaining part of the paper we extensively use the
notation 
\begin{align*}
\kappa=\min\{1+\tau/2,1+\rho/2,\rho\}
\end{align*} and $\tilde\eta$ of \eqref{eq:13.9.23.5.54b}.  Let us
recall two results, \cite[(1.9)]{IS3} and \cite[Lemma 5.1]{IS3},
respectively. (Recall for Lemma \ref{lem:13.9.2.7.18} that $a$ has two
values for $z\in I$, say  $a=a_\pm$.)

\begin{lemma}\label{lem:15.1.15.15.16aa}
Suppose Conditions~\ref{cond:12.6.2.21.13}--\ref{cond:10.6.1.16.24}.
Then, as quadratic forms on $\mathcal H^1$,
\begin{align*}
\begin{split}
H
&=\tfrac 12A\tilde\eta A
+\tfrac12 L +q_1+q_4;\quad q_4=q_2+\tfrac14(\nabla^r\tilde\eta)(\Delta r).
\end{split}%\label{eq:15.1.16.12.45aa}
\end{align*}
\end{lemma}

\begin{lemma}\label{lem:13.9.2.7.18}
Let $I\subseteq \mathcal I$ be a compact interval.
There exist $C>0$ such that uniformly in $z\in I\cup I_+$ or $z\in I\cup I_-$
\begin{align*}
%\begin{split}
& |a|\le C,\quad 
\bigl|\pm p^ra+a^2-2|\mathrm dr|^2(z-q_1)\bigr|
+\bigl|\ell^{\bullet i}\nabla_i a\bigr|\le Cr^{-\kappa}.
%\end{split}
%\label{eq:13.9.2.7.18}
\end{align*}
\end{lemma}

We may consider Lemma \ref{lem:15.1.15.15.16aa} as a decomposition  of $H$ into a sum of radial and
spherical components (see the discussion at the end of
the section). In the next 
section we shall use similar decompositions:

\begin{lemma}\label{lem:15.1.15.15.16ff}
Let $I\subseteq \mathcal I$  be a compact interval.
Then 
 as a quadratic form on $\bar \chi _n\mathcal H^1\subseteq \mathcal H^1$  for  any large $ n$
 and uniformly in $z=\lambda \pm
\i \Gamma \in I\cup I_\pm $
\begin{subequations}
\begin{align}\label{eq:decH}
H-z
&=
\tfrac 12(A\pm a)\tilde\eta(A\mp a)+\tfrac12 L +O(r^{-\kappa}),\\
H-z &=
\tfrac 12 b^{1/2} (\tilde A\pm \tilde b)b^{-1/2}(A\mp a)+\tfrac12 L +O(r^{-\kappa})(A\mp a)+O(r^{-\kappa}),\label{eq:decH2}\\
H-z &=
\tfrac 12\tfrac a {\sqrt b}  (\tilde A\pm \tilde b)\tfrac  {\sqrt b} a(A\mp a)+\tfrac12 L +O(r^{-\kappa})(A\mp a)+O(r^{-\kappa}).\label{eq:decH3}
 \end{align}  
\end{subequations}
\end{lemma}
\begin{proof}
Using Lemma~\ref{lem:15.1.15.15.16aa} we can write 
\begin{align*}
\begin{split}
H-z
&=
\tfrac 12(A\pm a)\tilde\eta(A\mp a)
\pm \tfrac12(p^r\tilde\eta a)
+\tfrac12\tilde\eta a^2
\\&\phantom{{}={}}
+\tfrac12 L 
+q_1+q_2
+\tfrac14(\nabla^r\tilde\eta)(\Delta r)
-z.
\end{split}
\end{align*}
Hence the first identity  \eqref{eq:decH} is
obtained  applying Lemma~\ref{lem:13.9.2.7.18} 
to  the remainder written 
\begin{align*}
&\tfrac 12\tilde\eta\bigl[
\pm (p^ra)+a^2
-2|\mathrm dr|^2(z-q_1)
\bigr]\\
&\phantom{{}={}}-(1-\eta)(z-q_1)
+q_{2}
+\tfrac14(\nabla^r\tilde\eta)(\Delta r\mp 2\mathrm ia)=O(r^{-\kappa}).
\end{align*} This is valid with or without the factor $\bar
\chi _n$. However   for \eqref{eq:decH2} and \eqref{eq:decH3} we need
this  factor to avoid dividing by zero.
 We use  \eqref{eq:decH} and the 
  identities 
  \begin{subequations}
  \begin{align}\label{eq:AtildeA}
  A\tilde\eta&= \tilde A -\tfrac \i 2 (\nabla^r \tilde\eta),\\
(A\pm
    a)\tilde\eta&= b^{1/2} (\tilde A\pm \tilde b)b^{-1/2}
 +O(r^{-\kappa})\label{eq:a_com2a},\\
(A\pm
    a)\tilde\eta&=a b^{-1/2} (\tilde A\pm \tilde b)a^{-1}b^{1/2}
 +O(r^{-\kappa})\label{eq:a_com3a}.
  \end{align}
\end{subequations}     
\end{proof}

The identities  
\begin{align}\label{eq:a_com1a}( A\mp 
    b)b^{1/2}=b^{1/2}( A\mp  b -\tfrac \i
    2\nabla^r \ln b)=b^{1/2}( A\mp 
    a+O(r^{-\kappa}))
\end{align} would provide more symmetric versions 
\eqref{eq:decH2} and \eqref{eq:decH3}, however these are not useful under our conditions.

We note the natural identification in spherical coordinates,
cf. \eqref{eq:co_area}, 
\begin{align*}
L^2(E)\cong L^2([r_0,\infty)_r;\mathcal G_r),\quad
 \langle\check \phi,\phi\rangle_{L^2(E)}
 =  \int_{r_0}^\infty\langle\check\phi,\phi\rangle_{\mathcal G_r}\,\mathrm dr.
 % \label{eq:13.9.23.16.40}
  \end{align*}
  Recalling \eqref{eq:13.9.23.5.54}, i.e. 
\begin{align*}
L=p_i^*\ell^{ij}p_j,\quad 
\ell=g-\tilde\eta\,\mathrm dr\otimes \mathrm dr\in \Gamma(T^{2}),
\end{align*} we can  write correspondingly (in the form sense)
\begin{align*}
  L\cong 
\int^\infty_{r_0}
\oplus L_r \,\d r.
\end{align*} Note here  the orthogonal splittings $g=\diag(|\d
r|^{-2},g')$ and $\ell=\diag(0,g')$, where $g'=g_r$ is the induced metric on $S_r\subseteq
M$. Indeed explicitly 
\begin{align*}
\inp{\check\phi,L_r\phi}_{\mathcal G_r}
=\int_{S_r} \overline{(p_i\check\phi)}g_r^{ij}(p_j\phi)\,\mathrm d\tilde{\mathcal A}_r=\langle p'\check\phi,p'\phi\rangle_{\mathcal G_r},
\end{align*} where  $\i p'$ is to the covariant derivative on
$S_r$. With Condition~\ref{cond:altG}  we can at this point 
use  local coordinates of $S$ to define and do the integral, in any
case  clearly the
radial derivative $\partial_r$ does not enter.

We may consider $L_r$ as an  operator, more precisely as the operator 
defined by   the
Friedrichs extension from $C_\c^\infty (S_r)\subseteq \mathcal
G_r$ of the expression
\begin{align}\label{eq:lsubr}
  L_r= -|\d r| (\det g_r)^{-1/2}\partial_i |\d r|^{-1}(\det
  g_r)^{1/2}g_r^{ij} \partial_j=-\Delta_r+\tfrac 12 (\partial_i \ln |\d r|^2)\ell ^{ij}\partial_j,
\end{align} where \begin{align*}
  \Delta_r=\Delta ^{\rm LB}_r=(\det g_r)^{-1/2}\partial_i (\det g_r)^{1/2}g_r^{ij} \partial_j
\end{align*} denotes the Laplace--Beltrami
operator on $S_r$.

By an approximation argument it follows that  for  any $\phi \in \mathcal H^1$ the
restriction $\phi_{|S_r}\in
\vD (L_r^{1/2})=\vD (p')$ for almost all $r\geq r_0$, in fact  for  all $r\geq r_0$
\begin{align*}
\int_{r_0}^{r}
\|p'\phi\|^2_{\mathcal G_s}\,\mathrm ds=\int_{B_{r}\setminus B_{r_0}}
\overline{(p_i \phi)}\ell^{ij}(p_j\phi) (\det g)^{1/2}\,\d x\leq \|\phi\|^2_{{\vH}^1}.
\end{align*}

\section{Distorted Fourier transform and stationary scattering
  theory}\label{sec:Fourier transform and stationary scattering
  theory}

In this section we impose Condition \ref{cond:12.6.2.21.13c}.  Recall
from Section \ref{subsec:Limiting Hilbert space} that 
$\mathrm e^{\mathrm it\tilde A}$, $t\leq 0$, naturally induces an
isometry
\begin{align*}
\mathrm e^{\mathrm it\tilde A}\colon
\mathcal G_{r+t}\to \mathcal G_{r};\quad 
\mathcal G_r=L^2(S_r,\mathrm d\tilde{\mathcal A}_r),\, r+t\geq r_0.
\end{align*} For $t\geq 0$ this operator is in general only a partial
isometry. 
Let us modify the exponent
and  consider the semigroups  $\mathrm{e}^{\mathrm{i}t(\tilde A\mp \tilde b)}$
generated by the operators $\tilde A\mp \tilde b$,
respectively. 
For any $\lambda>\lambda_0$ we have, cf.   (\ref{eq:13.9.5.7.2300}), 
\begin{align}
\tilde b=\tilde b_\lambda=\eta_\lambda|\mathrm dr|^{-1}\sqrt{2 (\lambda-q_1)}
=\tilde\eta b_\lambda.
\label{eq:13.9.5.7.23cc}
\end{align}
With the expressions (\ref{eq:12.9.24.22.46cc})
it is easy to verify that 
\begin{align}
\begin{split}
(\mathrm{e}^{\mathrm{i}t(\tilde A\mp \tilde b)}\phi)(x)
&=\exp \left(\mp \mathrm i\int_0^t\tilde b(\tilde y(s,x))\,\mathrm{d}s\right)
(\mathrm e^{\mathrm it\tilde A}\phi)(x).
\end{split}\label{eq:12.6.7.1.10cc}
\end{align} A similar formula holds for
$(\mathrm{e}^{\mathrm{i}t(\tilde A^{\rm ex}\mp \tilde b^{\rm
    ex})}$. It follows from \eqref{eq:12.6.7.1.10cc} that  the operators
$\mathrm{e}^{\mathrm{i}(r-r')(\tilde A\mp \tilde b)}$, $r'\geq r\geq
r_0$, induce isometries $\mathcal G_{r}\to \mathcal G_{r'}$. 
\begin{comment}We have $|\partial_tb(\tilde y(t,x))|\le Ct^{-\min\{\kappa/2,1+\tau'\}}$,
and hence there exist the limits
\begin{align*}
\lim_{r\to\infty}  a(\tilde{\mathrm e}^rx)=\lim_{r\to\infty}  b(\tilde{\mathrm e}^rx)
\end{align*}
uniformly in $x\in M\setminus B_{r_0}$ {and $\lambda\in I$}. 
 Whence an
equivalent ``definition'' of the distorted Fourier transform is given
by  introducing 
\begin{align}\label{eq:defResbC}
  \xi(r):= \exp \biggl(\int_{0}^{r(\cdot,r)}\parbb{\mp \i\tilde
    b+ \tfrac
    12\mathop{\mathrm{div}}\tilde\omega}
(\tilde y(-s,{}\cdot{}))\,\mathrm{d}s\biggr)
[\sqrt a R(\lambda\pm\mathrm i0)\psi]_{|S_{r}},
\end{align} 
\end{comment}  
\begin{align}\label{eq:defResb}
  \xi(r):= \exp \biggl(\int_{r_0}^{r}\parbb{\mp \i\tilde
    b^{\rm ex}+ \tfrac
    12\mathop{\mathrm{div}}\tilde\omega^{\rm ex}}(s,{}\cdot{})\,\mathrm{d}s\biggr)[\sqrt b R(\lambda\pm\mathrm
  i0)\psi](r,\cdot)\in \vG,
\end{align} 
 and that   the distorted Fourier transform is
``given'' by
\begin{align}\label{eq:disFb}
  F^\pm(\lambda)\psi =\vGlim_{r\to \infty} \xi(r).
\end{align} 
The first problem of  justifying this is to show that for each fixed $r$
indeed  $\xi(r)\in \vG$. This is in fact doable for $\psi\in B$.  More
generally under Condition 
\ref{cond:12.6.2.21.13c} the following four results are  valid  with  $\phi=R(\lambda\pm
\mathrm i 0)\psi$  for any  $\lambda>\lambda_0$  and   $\psi\in B$.
\begin{lemma}\label{lem:xi}
  For all $\psi\in B$ and $r\geq r_0$ the quantity $\xi(r)\in \vG$.
\end{lemma}
\begin{proof} Introduce $\xi\in C^\infty_\c(S_{r})$, $0\leq \xi\leq
  1$ and look at  the push-forward given by $\xi_{r'}=\xi(
\tilde y(r-r',{}\cdot{}))\in C^\infty_\c(S_{r'})$,  $r'\geq r\geq r_0$.

Note that the $\vG_{r}$-valued function 
\begin{align*}
  u(r'):=\mathrm e^{\mathrm i(r'-r)\tilde A}\bigl\{\xi_{r'}[\sqrt{b}\phi]_{|S_{r'}}\bigr\};\,r'\in [r,\infty),
\end{align*} is a  well-defined absolutely continuous function. In
particular by  the
fundamental theorem of calculus
\begin{align*}
  u(r)=\int^{r+1}_{r}u(s)\d s-\int^{r+1}_{r}\int^s_{r}\tfrac{ \d}{\d r'}
  u(r')\d r'\d s,
\end{align*} yielding upon computing the derivative
\begin{align*}
 \tfrac{ \d}{\d r'}
  u(r')=  \mathrm e^{\mathrm i(r'-r)\tilde A}\bigl\{\xi_{r'}[\i \tilde
  A\sqrt b\phi]_{|S_{r'}}\bigr\},
\end{align*}
 taking the norm inside and using \cs the bound
\begin{align}\label{eq:bnunix}
  \|\xi [\sqrt{b}\phi]_{|S_{r}}\|_{\vG_{r}}\leq \| 1_{B_{r+1}}\sqrt{b}\phi\|+3^{-1/2}\|
  1_{B_{r+1}}\tilde A\sqrt{b}\phi\|.
\end{align} By taking $\xi\nearrow
     1$ we obtain a concrete bound of the trace
$[\sqrt{b}\phi]_{|S_{r}}\in \vG_{r}$.
  \end{proof}

 In the above proof we only used the property that
  $\check \phi:=\sqrt{b}\phi\in \vN$ which follows from the fact that
  $\phi\in \vN$. The latter property suffices for  for
  the next result too. Note that for any such $\check \phi$ and $r\geq
  r_0$ we may  for any $R_\nu>r+1$  approximate $\chi_\nu\check \phi\in
  \vH^1$ by a sequence $(\check
   \phi_n)\subseteq C^\infty_\c(M)\subseteq \vH^1$. Then it follows from 
  \eqref{eq:bnunix} that $[\check \phi_n]_{|S_{r}}\to
  \check \phi_{|S_{r}}$ in  $\vG_{r}$ for  $n\to \infty$.

\begin{lemma}\label{lem:ac}
  The quantity $\xi(\cdot)\in \vG$ (possibly considered for an
  arbitrary $\phi\in \vN$) is an absolutely continuous
  $\vG$-valued function on $[r_0,\infty)$.
\end{lemma}
\begin{proof} We fix any $r_1>r_0$ and write for $r\in[r_0,r_1]$
\begin{align*}
  \xi(r)=\e^{\i
  (r-r_0)(\tilde A^{\rm ex }\mp\tilde
      b^{\rm ex })}\bigl [\sqrt b R(\lambda\pm\mathrm
  i0)\psi\bigr
    ]_{|S_r}=\e^{\i
  (r_1-r_0)(\tilde A^{\rm ex }\mp\tilde
      b^{\rm ex })} \e^{\i
  (r-r_1)(\tilde A\mp\tilde
      b)}\bigl [\sqrt b \phi\bigr
    ]_{|S_r}.
\end{align*} It suffices to show that the $\vG_{r_1}$-valued function
\begin{align*}
 [r_0,r_1]\ni r\to v(r) = \e^{\i
  (r-r_1)(\tilde A\mp\tilde
      b)}\bigl [\sqrt b \phi\bigr
    ]_{|S_r}
\end{align*} is absolutely continuous. Formally
\begin{align}\label{eq:fromAC}
  v'(r)=\e^{\i
  (r-r_1)(\tilde A\mp\tilde
      b)}\bigl [\i
  (\tilde A \mp \tilde b)\sqrt b \phi\bigr
    ]_{|S_r},
\end{align} i.e. $v(r)=v(r_1)-\int_{r}^{r_1} v'(s)\,\d s$ with $v'(s)$
given by this formula.
 If we replace $\sqrt{b}\phi=:\check \phi\in \vN$ by an
  approximating sequence $(\check \phi_n)\subseteq  C^\infty_\c(M)$ as
  in the remark preceding the lemma (used with  $r=r_1$) indeed
  \eqref{eq:fromAC} holds true for all $n$. Therefore
  \eqref{eq:fromAC} also holds in the
  limit $n\to \infty$.
  \begin{comment}
   Formally this is true since
\begin{align*}
  \mathrm e^{\mathrm i(r-r_1)\tilde A}\exp \biggl(\int_{r_0}^{r}\mp \i\tilde
    b(s,{}\cdot{})\,\mathrm{d}s\biggr)=\exp \biggl(\int_{r_0}^{r_1}\mp \i\tilde
    b(s,{}\cdot{})\,\mathrm{d}s\biggr)\mathrm e^{\mathrm i(r-r_1)(\tilde
      A\mp \tilde b)}.
  \end{align*}  
  \end{comment}
\end{proof}

The above proof gives the following formula for the
derivative (omitting  ``ex''):
\begin{align}\label{eq:defResbb}
  \xi'(r)=\tfrac{\d}{\d r}\xi(r)= \exp \biggl( \int_{r_0}^{r}\parbb{\mp \i\tilde
    b+ \tfrac
    12\mathop{\mathrm{div}}\tilde\omega}(s,{}\cdot{})\,\mathrm{d}s\biggr)[\i
  (\tilde A \mp \tilde b)\sqrt b \phi](r,\cdot)\in \vG.
\end{align}

\begin{lemma}\label{lem:fundaM} The following limit exists and is given
  as 
  \begin{align}
    \label{eq:limF}
    \lim_{R\to\infty}-\!\!\!\!\!\!\int_R\|[\sqrt{b}\phi]_{|S_{r}}\|^2_{\vG_{r}}\,\d
    r=\pm 2\mathop{\mathrm{Im}}\langle\psi,\phi\rangle.
  \end{align}
\end{lemma}
\begin{proof} Consider for convenience only the upper sign.
  \begin{align*}
    \int_R^{2R}\|[\sqrt{b}\phi]_{|S_{r}}\|^2_{\vG_{r}}\,\d
    r=\int_R^{2R} \parb{
\mathop{\mathrm{Re}}\langle \phi,(b-A) \phi\rangle_{\mathcal G_r}
+\mathop{\mathrm{Im}}\langle \phi,\nabla_{\omega}\phi\rangle_{\mathcal G_r}}\,\d r.
\end{align*} By Corollaries \ref{cor:12.7.2.7.9b} and 
\ref{cor:radiation-conditions}
\begin{align*}
  -\!\!\!\!\!\!\int_R 
\mathop{\mathrm{Re}}\langle \phi,(b-A) \phi\rangle_{\mathcal
  G_r}\,\d r \to 0.
\end{align*} Next  we introduce for any  $R\ge r_0$ 
 a smooth approximation of the characteristic function
of  the ball $B_R$  of the form  (employing \eqref{eq:14.1.7.23.24})
\begin{align*}
  \chi_{\epsilon,s}(r)=\chi((r-R-s)/\epsilon);\quad  \epsilon>0,\,
  s\in [0,R].
\end{align*} 
We compute a Green's identity
\begin{align}
\begin{split}
&\int_R^{2R}\mathop{\mathrm{Im}}\inp{\phi,\nabla_{\omega}\phi}_{\vG_r} \,
\d r\\
&=\lim _{\epsilon\to 0 }\mathop{\mathrm{Im}}
\int_{r_0}^\infty \parbb{\chi(r-2R)/\epsilon)-\chi((r-R)/\epsilon}\inp{\phi,  \nabla_{\omega}\phi}_{\vG_r}\,\d r\\
&=-\lim _{\epsilon\to 0 }\mathop{\mathrm{Im}}
\int_{r_0}^\infty \parbb{\int_0^R\chi'_{\epsilon,s}(r)\,\d s} \inp{\phi,  \nabla_{\omega}\phi}_{\vG_r}\,\d r\\
&=-\lim _{\epsilon\to 0 }\mathop{\mathrm{Im}}
\int_0^R \int_M (\partial_i\chi_{\epsilon,s}) \bar{\phi}  g^{ij}(\partial_j\phi)(\det g)^{1/2}\,\d x\d s\\
&=-\lim _{\epsilon\to 0 }\mathop{\mathrm{Im}}
\int_0^R \int_M (\partial_i\chi_{\epsilon,s}\bar{\phi})
g^{ij}(\partial_j\phi)(\det g)^{1/2}\,\d x \d s\\
&=-2\lim _{\epsilon\to 0 }\int_0^R \mathop{\mathrm{Im}}
\langle\chi_{\epsilon,s} \phi, (H-\lambda)\phi\rangle_{\mathcal H}\,\d
s\\
&=-2 \int_0^R\mathop{\mathrm{Im}}
\langle 1_{B_{R+s}} \phi, \psi\rangle_{\mathcal H}\,\d s.
  \end{split}\label{eq:14.4.30.7.48a}
  \end{align}  
 It follows that
\begin{align*}
  -\!\!\!\!\!\!\int_R 
\mathop{\mathrm{Im}}\langle \phi,\nabla_{\omega}\phi\rangle_{\mathcal G_r}\,\d r\to 2\mathop{\mathrm{Im}}\langle\psi,\phi\rangle.
\end{align*}
\end{proof}

We remark  that  a small modification of the computation
\eqref{eq:14.4.30.7.48a} yields the more familiar
  Green's identity
  \begin{subequations}
  \begin{align}\label{eq:greens}
    \mathop{\mathrm{Im}}\inp{\phi,\nabla_{\omega}\phi}_{\vG_r}=-2\mathop{\mathrm{Im}}
    \langle 1_{B_{r}} \phi, \psi\rangle_{\mathcal H}\text { for almost
      all }r\geq r_0,
  \end{align} yielding in particular that $r\to
  \mathop{\mathrm{Im}}\inp{\phi,\nabla_{\omega}\phi}_{\vG_r}$ is
  absolutely continuous. A similar computation shows that in fact $r\to
  \inp{\check\phi,\nabla_{\omega}\phi}_{\vG_r}$ is absolutely
  continuous for any $\check\phi\in \vN$ as it follows from the
  resulting Green's identity
  \begin{align}\label{eq:greens2}
    \inp{\check\phi,\nabla_{\omega}\phi}_{\vG_r}=\langle 1_{B_{r}} p_i\check\phi, g^{ij}p_j\phi\rangle+2
     \big\langle 1_{B_{r}} \check\phi, (V-\lambda)\phi-\psi
     \big\rangle\text { for }r\geq r_0.
  \end{align}  However, in comparison, we do not know continuity or
  even local boundedness of the function 
$r\to \|\nabla_{\omega}\phi\|_{\vG_r}$.  
     \end{subequations}

     The following technical result will play a major role (see the
     proofs of Lemmas \ref{lem:normBasi}--\ref{lemma:s1}). Recall that the factor $\bar \chi
     _n$ of Lemma \ref{lem:15.1.15.15.16ff} was introduce to avoid
     zeros of $a$ and $b$. This is also the role of the factor $\bar
     \chi _n$ below.

\begin{lemma}\label{lem:fundaM2} Let $I\subseteq \mathcal I$  be a compact interval.
Then 
 we introduce  for  any large $ n$ a function $f_{\check r}(r)$,
 $r\geq \check r$, 
 depending on 
 any $\check r\geq r_0$ as well as on any  $\lambda\in I$   and $\check \phi\in \vN$
 as follows: Using
   spherical coordinates we define for $r\geq \check r$
\begin{align*}
\check e&=\exp \biggl(
    \int_{\check  r}^{r}\pm 2\i\tilde
    b(s,{}\cdot{})\,\mathrm{d}s\biggr),\\  
D \xi(r)&= \exp \biggl( \int_{r_0}^{r}\parbb{\mp \i\tilde
    b+ \tfrac
    12\mathop{\mathrm{div}}\tilde\omega}(s,{}\cdot{})
\,\mathrm{d}s\biggr)[\sqrt b \i(A\mp a)\phi](r,\cdot)\in\vG,\\
\check \xi(r)&= \exp \biggl( \int_{r_0}^{r}\parbb{\mp \i\tilde
    b+ \tfrac
    12\mathop{\mathrm{div}}\tilde\omega}(s,{}\cdot{})
\,\mathrm{d}s\biggr)[  \sqrt b\check \phi](r,\cdot)\in\vG,\\
f_{\check r}(r)&=\inp{\check\xi(r), ({\check
    e}b^{-1}\bar\chi_n )(r,\cdot)D\xi(r)}_{\vG}.
\end{align*}   

Then  the function $ f_{\check r}(\cdot)$
  is absolutely continuous on $[\check  r, \infty)$ with derivative
  \begin{align}\label{eq:T-1}
  \begin{split}
  f'_{\check r}(r)&=T_1+\cdots +T_5;\\
T_1&=\inp{(\tilde A \mp \tilde b)\sqrt b \check \phi,\check e b^{-1/2}\bar\chi_n (A\mp a)\phi}_{\vG_r},\\
T_2&=-2\inp{\check \phi,\check e \bar\chi_n \psi}_{\vG_r},\\
T_3&=\inp{p'\overline{\check e} \check \phi,\bar\chi_n p' \phi}_{\vG_r},\\
T_4&=\inp{\check \phi,O(r^{-\kappa})(A\mp a)\phi}_{\vG_r},\\
T_5&=\inp{\check \phi,O(r^{-\kappa})\phi}_{\vG_r},
  \end{split}
\end{align} where the bounds of $T_4$ and $T_5$ are uniform in
$\lambda\in I$ and $\check
r\geq r_0$.
\end{lemma}
\begin{proof}
 First we proceed   using \eqref{eq:decH2} somewhat
 unjustified. Compute (formally)
\begin{align*}
  f'_{\check r}(r)=\inp{(\tilde A \mp \tilde b)\sqrt b \check
    \phi,\check e b^{-1/2} \bar\chi_n (A\mp a)\phi}_{\vG_r}-\inp{\sqrt b \check \phi,  
  (\tilde A \mp \tilde b)\check e b^{-1/2}\bar\chi_n   (A\mp a)
  \phi}_{\vG_r},
\end{align*} and then substituting for the second term
\begin{align*}
\sqrt b(\tilde A \mp \tilde b)\check e b^{-1/2}\bar\chi_n  (A\mp a)
&= \check e\sqrt b(\tilde A \pm \tilde b)b^{-1/2} \bar\chi_n(A\mp a)\\
&=2\check e \bar\chi_n\parb{ H-\lambda-\tfrac 12 L}+O(r^{-\kappa})
 (A\mp a)+O(r^{-\kappa}).
\end{align*} 
  This yields  \eqref{eq:T-1}. Note that $T_3$ is
well-defined since in fact $\overline{\check e}\check \phi\in\vN$ due to
\eqref{eq:13.9.23.16.41b2}. By the product rule
\begin{align}\label{eq:comInt}
  p'\overline{\check e}\check \phi
=\mp\biggl(p'\int_{\check r}^{r}2\i\tilde b(s,{}\cdot{})\,\d s\biggr)
\overline{\check e}\check \phi+ \overline{\check e}p'\check
  \phi,
\end{align}  
yielding that $T_3\in L^1_{\mathrm{loc}}$ as a function of
$r$. Similarly for the other terms showing explicitly that $f'_{\check
  r}\in L^1_{\mathrm{loc}}$.  The required (pointwise) uniformity property for
$T_4$ and $T_5$ is trivial since the $\check r$-dependence is through
the oscillatory  factor  ${\check e}$ only.

Next we give a rigorous derivation of \eqref{eq:T-1} using
\eqref{eq:decH2} differently (this argument will not be repeated for
the derivation of similar formulas in the proof of  Lemmas \ref{lemma:w1} and
\ref{lemma:s1}). We already argued that all of the
above terms make sense and agree with the conclusion of the lemma. We
claim that indeed $ f_{\check r}$ is absolutely continuous.  Note that
due to \eqref{eq:greens2} and the fact that $\overline{\check e}\check
\phi\in\vN$ we have the representation
\begin{align}\label{eq:frr}
  \begin{split}
   f_{\check r}(r)=&\langle 1_{B_{r}} p_i\bar\chi_n\overline{\check e}\check\phi, g^{ij}p_j\phi\rangle+2
     \big\langle 1_{B_{r}} \bar\chi_n\overline{\check e}\check\phi, (V-\lambda)\phi-\psi
     \big\rangle\\& +\big\langle \bar\chi_n\overline{\check
       e}\check\phi, (\tfrac 12\Delta r\mp \i a)\phi
     \big\rangle_{\vG_r}. 
  \end{split}
\end{align} Clearly the first and second terms are absolutely
continuous, and  by  Lemma \ref{lem:ac} the last one is too. 

It is of course doable to compute $f'_{\check r}$ using
\eqref{eq:frr}. However the result is not immediately consistent with
the 
representation from our informal computation. Instead  we shall
proceed as follows:   It remains to show that for $r_1>\check r$
\begin{align}\label{eq:abreprre}
  f_{\check r} (r_1)=\int ^{r_1} (T_1+T_2+T_3+T_4+T_5)\,\d r.
\end{align} Let for $r_1>\check r$ 
\begin{align*}
  \chi_{\epsilon}(r)=\chi((r-r_1)/\epsilon);\quad  \epsilon>0.
\end{align*} We compute on one hand 
\begin{align*}
  \inp{(\tilde A\pm \tilde b)b^{1/2}&\chi_{\epsilon}\bar\chi_n \overline{\check e}{\check \phi}, b^{-1/2}(A\mp a)\phi
  }\\&=\inp{\chi'_{\epsilon}\bar\chi_n \overline{\check e}{\check \phi}, \i(A\mp a)\phi
  }+\inp{\chi_{\epsilon}(\tilde A\pm \tilde b)b^{1/2}\bar\chi_n \overline{\check e}{\check \phi}, b^{-1/2}(A\mp a)\phi
  }
\\&=\int \chi'_{\epsilon}(r)f_{\check r}(r)\,\d r+\inp{\chi_{\epsilon}(\tilde A\pm \tilde b)b^{1/2}\bar\chi_n \overline{\check e}{\check \phi}, b^{-1/2}(A\mp a)\phi
  }, 
\end{align*} and on the other hand using \eqref{eq:decH2} (considering
$L$ as a form)
\begin{align*}
  \inp{(\tilde A\pm \tilde b)b^{1/2}&\chi_{\epsilon}\bar\chi_n \overline{\check e}{\check \phi}, b^{-1/2}(A\mp a)\phi
  }\\&=\inp{\chi_{\epsilon}\bar\chi_n \overline{\check e}{\check \phi}, 2\psi-2\parb{\tfrac12 L +O(r^{-\kappa})(A\mp a)+O(r^{-\kappa})}\phi
  }.
\end{align*}  Since $f_{\check r} $ is continuous at $r_1$ we obtain
using that the right-hand sides are equal
and by letting $\epsilon\to 0$ that 
\begin{align*}
  &-f_{\check r}(r_1)+\inp{1_{B_{r_1}}(\tilde A\pm \tilde b)b^{1/2}\bar\chi_n \overline{\check e}{\check \phi}, b^{-1/2}(A\mp a)\phi
  }\\&=\inp{1_{B_{r_1}}\bar\chi_n \overline{\check e}{\check \phi}, 2\psi-2\parb{\tfrac12 L +O(r^{-\kappa})(A\mp a)+O(r^{-\kappa})}\phi
  }
\\&=2\inp{1_{B_{r_1}}{\check \phi}, {\check e}\bar\chi_n \psi}-
\inp{1_{B_{r_1}}p'\overline{\check e}
{\check \phi}, \bar\chi_n p'\phi
  }\\& -
\inp{1_{B_{r_1}}{\check \phi}, O(r^{-\kappa})(A\mp a)\phi
  }-\inp{1_{B_{r_1}}{\check \phi}, O(r^{-\kappa})\phi
  }.
\end{align*} Moreover 
\begin{align*}
  &\inp{1_{B_{r_1}}(\tilde A\pm \tilde b)b^{1/2}\bar\chi_n \overline{\check e}{\check \phi}, b^{-1/2}(A\mp a)\phi
  }\\&=\inp{1_{B_{r_1}}(\tilde A\mp \tilde b)b^{1/2}\bar\chi_n {\check \phi}, {\check e}b^{-1/2}(A\mp a)\phi
  }\\&=\inp{1_{B_{r_1}}(\tilde A\mp \tilde b)b^{1/2}{\check \phi}, {\check e}\bar\chi_n b^{-1/2}(A\mp a)\phi
  } +
\inp{1_{B_{r_1}}{\check \phi}, O(r^{-\kappa})(A\mp a)\phi
  }.
\end{align*} We conclude \eqref{eq:abreprre}.

\end{proof}

\subsection{Proof of Theorem~\ref{thm:strong}\  for easy
  case}\label{subsec:easy case} Suppose in addition to  Condition
\ref{cond:12.6.2.21.13c} that Condition \ref{cond:12.6.2.21.13bb}
~\eqref{item:14.5.1.8.30} holds and consider only $\psi\in \mathcal
H_{3/2+}$. We show that
\eqref{eq:disFb}  exists. For convenience we consider only the upper sign. Note that
the estimate of Corollary \ref{cor:radiation-conditions} holds for
some $\beta>1$.

  We compute,  cf. \eqref{eq:a_com1a} and \eqref{eq:AtildeA}, 
\begin{align*}
  (\tilde A-\tilde b)b^{1/2}&=b^{1/2}(\tilde A-\tilde b -\tfrac \i
    2\tilde \omega^i\nabla_i \ln b)\\&=b^{1/2}(\tilde A-\tilde
    a+O(r^{-2}))\\&=b^{1/2}\tilde\eta(A-
    a+O(r^{-2})).
  \end{align*}  Using  then in turn  \eqref{eq:defResbb} and 
  \cs we obtain for  $\beta$ slightly bigger than $1$
\begin{align*}
\int_{r_0}^\infty \| \tfrac{\d}{\d r}\xi(r)\|_{\mathcal G}\,\d r
&\leq C_\beta\bigg(\int_{r_0}^\infty r^{2\beta-1}
\bigl\|[\sqrt{b}|\d r|^{-2}
\parb{A-a+O(r^{-2})}\phi]_{|S_{r}}\bigr\|^2_{\mathcal G_{r}}\,\d r\bigg)^{1/2}\\
&\leq C_1 \|\parb{A-a}\phi\|_{\beta-1/2}+C_2 \\&\leq C_3<\infty.
\end{align*} Whence  the existence of \eqref{eq:disFb}  follows by integration.

The constant $C_3$ can be chosen locally uniform in
$\lambda>\lambda_0$ and arbitrary small if we replace
$\int_{r_0}^\infty$ by  $\int_{R}^\infty$, $R>r_0$ big. Whence the
limit \eqref{eq:disFb}  is attained locally
uniformly in $(\lambda_0,\infty)$. In addition, since 
for finite $r$ the map $\lambda\to \xi(r)$
is continuous (cf. \eqref{eq:bnunix}), we obtain
continuity of the map
$(\lambda_0,\infty)\ni\lambda\to F^+(\lambda)\psi\in
\mathcal G$.

Let us also note that due to  Lemma \ref{lem:fundaM}
\begin{align*}
  \|F^+(\lambda)\psi\|=
2\mathop{\mathrm{Im}}\langle\psi,\phi\rangle
\end{align*} follows from the computation
  \begin{align*}
     \|F^+(\lambda)\psi\|_{\vG}^2=\lim_{R\to\infty}R^{-1}\int_R^{2R}\|\xi(r)\|_{\vG}^2\,\d r
=\lim_{R\to\infty}-\!\!\!\!\!\!\int_R\|[\sqrt{b}\phi]_{|S_{r}}\|^2_{\vG_{r}}\,\d r.
\end{align*}

\subsection{Proof of Theorem~\ref{thm:strong}\ for general case}
\label{subsec:Fourier transform} 
In this subsection we prove the existence of the limit
\eqref{eq:disFb}  for  $\psi\in
{\mathcal H_{1+}}$  under Condition \ref{cond:12.6.2.21.13c} and then the
remaining assertions of Theorem~\ref{thm:strong}.  We shall
consider only the upper sign, since the lower sign can be dealt with
in parallel.  Throughout the subsection we fix any  compact
interval $I\subseteq \mathcal I$.

We  first investigate  properties of spherical derivatives. 
Let us introduce a  ``backwards hitting time'' by 
\begin{align*}
  r^{\rm bht}(x)=\sup\bigl\{s \leq r(x)-r_0\,\big| \,
\tilde y(-s,x)\in M\big\};\,
  x\in E.
\end{align*} 
%Note that $r^{\rm bht}$ is lower semi-continuous and whence a Borel function.

\begin{lemma}
  \label{lemma:L_10b}
There exists $C_1>0$ such that for all $x\in E$ and $t\in (-r^{\rm bht}(x),0]$ 
\begin{align}
\begin{split}
\ell_*(t,x)
:={}&
\Bigl(
\ell^{ij}(x)
[\partial_i\tilde y^\alpha(t,x)][\partial_j\tilde y^\beta(t,x)]
\Bigr)_{\alpha,\beta}
\\
\le{}& C_1\bigl[(r(x)+t)/r(x)\bigr]^{\sigma}
\ell(\tilde y(t,x))
\end{split}\label{eq:13:9:34:16:20}
\end{align}
as quadratic forms on the fibers of the cotangent bundle. 
In particular, for any given  $\check\sigma\le \sigma$ with 
$\check\sigma<\min\{\tau,\rho\}$ 
there exists $C_2>0$ independent of $\lambda\in I$ such that for all  
$x\in E$ 
  \begin{align}
   \int_{-r^{\rm bht}(x)}^{0}\big|p'\tilde b(\tilde y(s,x)) \big|\,\d s \le 
C_2r(x)^{-\check\sigma/2}.
  \label{eq:13.9.23.16.41b}
  \end{align}
\end{lemma}
\begin{proof}
We remark that the tensor $\ell_*(t,x)$ is the push-forward of 
$\ell(x)$ under $\tilde y(t,{}\cdot{})$.
To prove the  inequality of (\ref{eq:13:9:34:16:20})  we
consider  the trace
\begin{align*}
F(t)=
g_{\alpha\beta}(\tilde y(t,x))
\ell^{ij}(x)
[\partial_i\tilde y^\alpha(t,x)][\partial_j\tilde y^\beta(t,x)].
\end{align*}
 Note that we use the Roman and the Greek indices to denote 
quantities concerning $x$ and $\tilde y=\tilde y(t,x)$, respectively.
The derivative $F'(t)$ is computed as follows:
Differentiating the expression 
\begin{align*}
  (g^*)_{ij}(t,x):=g_{\alpha\beta}(\tilde y(t,x))
[\partial_i\tilde y^\alpha(t,x)][\partial_j\tilde y^\beta(t,x)]
\end{align*}
and 
using the compatibility condition (\ref{eq:12.9.27.1.21})
and the flow equation,
we  compute
\begin{align}
\begin{split}
\tfrac{\partial}{\partial t}(g^*)_{ij}
&= 
[\Gamma_{\gamma\alpha}^\delta g_{\delta\beta}
+\Gamma_{\gamma\beta}^\delta g_{\alpha\delta}]\tilde\omega^\gamma 
(\partial_i\tilde y^\alpha)(\partial_j\tilde y^\beta)
\\
&\phantom{={}}
+
g_{\alpha\beta}
(\partial_\gamma \tilde \omega^\alpha)(\partial_i\tilde y^\gamma)
(\partial_j\tilde y^\beta)+
g_{\alpha\beta}
(\partial_\gamma \tilde \omega^\beta)(\partial_i\tilde y^\alpha)
(\partial_j\tilde y^\gamma)\\
&=2(\nabla \tilde \omega)_{\alpha\beta}(\partial_i\tilde y^\alpha)
(\partial_j\tilde y^\beta),
\end{split}
\label{eq:12.10.12.8.38}
\end{align}
which yields
 \begin{align*}
F'(t)=
2\ell^{ij}(x)(\nabla\tilde\omega(\tilde y(t,x)))_{\alpha\beta}
[\partial_i\tilde y^\alpha(t,x)][\partial_j\tilde y^\beta(t,x)].
\end{align*}
Next we decompose $\nabla\tilde\omega =|\mathrm dr|^{-2}\nabla^2
r+(\mathrm d|\mathrm dr|^{-2})\otimes \mathrm dr$ and substitute
 in the above formula. The
second term does not contribute which follows easily using spherical
coordinates.  Using then \eqref{eq:13.9.5.3.30}  we obtain
\begin{align*}
F'(t)
\ge 
\sigma'(r(x)+t)^{-1}
F(t),
\end{align*}
so that for $t\in (-r^{\rm bht}(x),0]$
\begin{align*}
F(t)\le (d-1)\bigl[(r(x)+t)/r(x)\bigr]^{\sigma'}.\end{align*}
Hence we obtain 
\begin{align*}
\Bigl(
\ell^{ij}(x)
[\partial_i\tilde y^\alpha(t,x)][\partial_j\tilde y^\beta(t,x)]
\Bigr)_{\alpha,\beta}
\le (d-1)\bigl[(r(x)+t)/r(x)\bigr]^{\sigma'}
g(\tilde y(t,x)).
\end{align*}
If we write the last inequality in the spherical coordinates again,
then there appears no radial component to the left,
and hence we can remove the radial component from the right.
Thus, using also that $\sigma<\sigma'$, the inequality  (\ref{eq:13:9:34:16:20}) follows.

Now for  (\ref{eq:13.9.23.16.41b}) we first use the chain rule. Next
by \cS, (\ref{eq:13.9.21.16.44}),
Condition~\ref{cond:12.6.2.21.13bbb} and (\ref{eq:13:9:34:16:20})
with $\sigma$ replaced by $\check\sigma$ we can estimate 
\begin{align*}
&
\int_{-r^{\rm bht}(x)}^0
\big|[\partial_\bullet\tilde y^\alpha(s,x)]
\bigl(\partial_\alpha \tilde b
\bigr)(\tilde y(s,x))\big|\,\mathrm{d}s\\
&\le 
C_4 r(x)^{-\check\sigma/2}\int_{-r^{\rm bht}(x)}^0(r(x)+s)^{-1-\min\{\tau,\rho\}/2  +\check\sigma/2}\,\mathrm{d}s
\\&\le 
C_5r(x)^{-\check\sigma/2},
\end{align*} 
 showing  (\ref{eq:13.9.23.16.41b}).
\end{proof}

As we already noted the tensor $\ell_*(t,x)$ is the push-forward of 
$\ell(x)$ under  the map $\tilde y(t,{}\cdot{})$.
Actually Lemma~\ref{lemma:L_10b}
is applied only in the spherical coordinates,
in which the inequality \eqref{eq:13:9:34:16:20}
takes a more simplified form, as in the following corollary.
However, partly to avoid confusion concerning how we should 
compare two tensors with different base points,
we formulate it using a more convenient  geometric terminology.

\begin{corollary}
  \label{cor:C_10}
  Let $\check\sigma\le \sigma$ and $\check\sigma<\min\{ \tau,\rho\}$.
  Then there exists a constant $C>0$ such that 
  for all  $ \check r\ge r_0$  and  $u\in C^1_{\mathrm
    c}(S_{\check r})$, 
  the function 
  $u(r)=\mathrm e^{\mathrm i(\check r-r)\tilde A}u\in 
  C^1_{\mathrm c}(S_{r})$   for  $r\ge \check r$ and 
satisfies
\begin{subequations}
\begin{align}
  \|u(r)\|_{\mathcal G_r}&=\|u\|_{\mathcal G_{r'}},
    \label{eq:13.9.23.16.40}
\\
  \| p'u(r)\|_{\mathcal G_r}&\le 
C\Bigl ( (\check r/r)^{\sigma/2}
  \|p'u\|_{\mathcal G_{\check r}}+r^{-\check \sigma/2}\| u\|_{\mathcal G_{\check r}}
 \Bigr).
  \label{eq:13.9.23.16.41}
\end{align} More generally  the bounds \eqref{eq:13.9.23.16.40} and
\eqref{eq:13.9.23.16.41} are valid for $u\in\vD
(L_{\check r}^{1/2})$ in which case $u(r)\in\vD
(L_{r}^{1/2})$ for all $r\ge \check r$.
\end{subequations}
\end{corollary}
\begin{proof}
  Only \eqref{eq:13.9.23.16.41} needs justification (note that the last
  assertion follows by  a density argument). Using  the expression
  in the spherical coordinates
\begin{align*}
\begin{split}
u(r)(\sigma)
=(\mathrm{e}^{\mathrm{i}(\check r-r)\tilde A}u)(\sigma)
=\exp \biggl(\int_{\check r}^r-\tfrac12\mathop{\mathrm{div}}\tilde\omega(s,\sigma)
\,\mathrm{d}s\biggr) u(\sigma),
\end{split}
\end{align*}
 we compute with $t=\check r-r$ and $x\in E $ given with 
 coordinates $(r,\sigma)$
  \begin{align}
    \begin{split}
    &p_i u(r)(\sigma)
    =-\i \e^{ (\int\cdots)}[\partial_i\tilde y^\alpha(t,x)]
\bigl(\partial_\alpha u
\bigr)(\check r, \sigma)\\&
+\biggl(\int_t^{0}
      \bigl[\partial_i\tilde y^\alpha(s,x)]
\bigl(\tfrac\i{2}\partial_\alpha
      \mathop{\mathrm{div}}\tilde \omega
      \bigr)(r+s,\sigma) 
\,\mathrm{d}s \biggr)
    u(r)(\sigma).
  \end{split}\label{eq:14.3.5.11.0}
\end{align} The first term is estimated using \eqref{eq:13:9:34:16:20}, and the second term is estimated as in the
proof of \eqref{eq:13.9.23.16.41b}.
\begin{comment}
(In the estimation of the first term of \eqref{eq:14.3.5.11.0} 
we should note that $\sigma$ in the integrand has a base point
on $S_s$ differently from the other $\sigma$ on $S_{\check r}$.)  
\end{comment}
\end{proof}

For the remaining part of this section let  $\psi\in \mathcal H_{1+}$ 
and   $\phi=R(\lambda+\mathrm i0)\psi$, $\lambda\in I\subseteq \mathcal I$.
  
\begin{lemma}\label{lem:normBasi}  The following limit exists and
   is given as
  \begin{align}
    \label{eq:normInf}
    \lim_{r\to \infty}\|\xi(r)\|_{\vG}^2= 2\mathop{\mathrm{Im}}\langle\psi,\phi\rangle_{\mathcal H}.
  \end{align}
\end{lemma}
  \begin{proof} 
    We use Lemma \ref{lem:fundaM2} taking there $\check
    \phi=\phi$. By evaluating at
    $r=\check r$ and integrating the derivative on the interval
    $[\check r, \infty)$ we then obtain that
    \begin{align*}
      \lim_{r\to \infty}f_r(r)=0.
    \end{align*} At this point note that all of the terms $T_1,\dots,
    T_5$ are integrable due to \cS, Corollaries \ref{cor:12.7.2.7.9b}
    and \ref{cor:radiation-conditions},
    \eqref{eq:13.9.23.16.41b2} and \eqref{eq:comInt} (note that if
    $M^{\rm ex}=M$ the bound \eqref{eq:13.9.23.16.41b2} follows from 
    \eqref{eq:13.9.23.16.41b} with $\check\sigma=1$). Next by taking the imaginary part and using
    \eqref{eq:greens} we obtain
\begin{align*}
      0=\lim_{r\to \infty}\Im \inp{\phi,\nabla _\omega
          \phi-\i b\phi}_{\vG_r}=
        2\mathop{\mathrm{Im}}\langle\psi,\phi\rangle_{\mathcal
          H}-\lim_{r\to \infty}\|\sqrt b\phi(r)\|_{\vG_r}^2.
\end{align*} 
    Whence indeed 
\begin{align*}
      \lim_{r\to \infty}\|\xi(r)\|_{\vG}^2= \lim_{r\to \infty}\|\sqrt b\phi(r)\|_{\vG_r}^2=2\mathop{\mathrm{Im}}\langle\psi,\phi\rangle_{\mathcal H}.
    \end{align*}
\end{proof}

\begin{remark*}It follows from the above proof that the limit
  \eqref{eq:normInf} is attained uniformly in $\lambda\in I$. This
  property will be used in the proof of Proposition
  \ref{prop:dist-four-transf}.
  \end{remark*}
  
Next decompose 
\begin{align*}
  \xi=\xi(r)= \exp \biggl( \int_{r_0}^{r}\parbb{- \i\tilde
    b+ \tfrac
    12\mathop{\mathrm{div}}\tilde\omega}(s,{}\cdot{})
\,\mathrm{d}s\biggr)[\sqrt b \phi](r,{}\cdot{})
\end{align*} 
as
\begin{align}\label{eq:DECO}
  \begin{split}
  \xi&=a^{-1}\xi_++a^{-1}\xi_-;\\
  \xi_{\pm}&= 2^{-1}\exp \biggl( \int_{r_0}^{r}\parbb{-
    \i\tilde b+ \tfrac 12\mathop{\mathrm{div}}\tilde\omega}(s,{}\cdot{})
\,\mathrm{d}s\biggr)[\sqrt b(a\pm A) \phi](r,{}\cdot{}).  
  \end{split}
\end{align} At this point the reader is \textit {WARNED} about our use
of notation: The quantities $a$ and $\phi$ are considered with the
\textit {upper sign only} in this subsection, so for the cases $\pm$
in \eqref{eq:DECO}  these quantities are the \textit {same} (not
to be mixed up with the convention of Lemmas \ref{lem:15.1.15.15.16ff}
and \ref{lem:fundaM2}).

\begin{lemma}
  \label{lemma:w1} 
  There exists the weak limit
  \begin{align*}
F:=\wvGlim_{r\to
    \infty} \xi(r).
  \end{align*}
\end{lemma}
\begin{proof}
  Let $g\in C_\c^\infty(S)\subseteq \vG$ be given. Due to Lemma
  \ref{lem:normBasi} it suffices to show the existence of 
\begin{align*}
C_{\pm}:=\lim_{r\to
    \infty} \inp{g, a^{-1}(r)\xi_{\pm}(r)}_{\vG}.
  \end{align*} 
\textit{Step I.} $C_-=0$. Writing
\begin{align*}
  g=\exp \biggl( \int_{r_0}^{r}\tfrac
    12\mathop{\mathrm{div}}\tilde\omega(s,{}\cdot{})\,\mathrm{d}s\biggr) u(r)
\end{align*} we note that $u(r')\in C_\c^\infty(S_{r'})$ for
$r'\geq R_n$ with $n$   large enough. We can write $u(r)=\mathrm e^{\mathrm i(r'-r)\tilde A}u(r')\in 
  C^\infty_{\mathrm c}(S_{r})$ for $r\geq r'$. Let
  $\bar\chi_n=1-\chi(r/R_n)$ so that $\bar\chi_n u\in \vN$ (and
  such that  all zeros of $a$ and $b$ are  in $B_{R_n}$).
  Introduce then for $r\geq r_0$
\begin{align*}
  \check \phi&=\exp \biggl( \int_{r_0}^{r} \i\tilde
    b(s,{}\cdot{})\,\mathrm{d}s\biggr) b^{-1/2}\bar\chi_n u(r),\\
\check \xi&= \exp \biggl( \int_{r_0}^{r}\parbb{ -\i\tilde
    b+ \tfrac
    12\mathop{\mathrm{div}}\tilde\omega}(s,{}\cdot{})
\,\mathrm{d}s\biggr) [\sqrt b\check \phi](r, \cdot) \in \vG.
  \end{align*}  Note 
   that we can consider $\check \phi$ as an
  element of $\vN$ and that $\check \xi(r)=g$.  Also note that $\xi_-=\i
  2^{-1}D\xi $ in terms of notation of Lemma \ref{lem:fundaM2}. We introduce  as in Lemma \ref{lem:fundaM2}
  \begin{align*}
    \check e=\exp \biggl(
    \int_{\check  r}^{r} 2\i\tilde
    b(s,{}\cdot{})\,\mathrm{d}s\biggr);\,r\geq\check r\geq r_0.
  \end{align*}
By the proof of Lemma \ref{lem:fundaM2} with this choice of
  $\check \phi$ (and by using \eqref{eq:decH3} rather than
  \eqref{eq:decH2}) it follows by evaluating
  at $r=\check r\geq 2R_n$ and integrating the derivative
    \begin{align*}
      \tfrac 2 \i\tfrac \d{\d r}\inp{ \check \xi(r),   {\check e}  a^{-1}
\xi_-(r)}_{\vG}&= \inp{(\tilde A - \tilde b)
\sqrt b \check \phi,\check e  a^{-1}\sqrt b
(A-a)\phi}_{\vG_r}
\\&\phantom{{}={}}
-\inp{\sqrt b \check \phi,  
  (\tilde A - \tilde b)\check
e  a^{-1}\sqrt b( A
    -  a)
  \phi}_{\vG_r}\\& =
-\inp{\sqrt b \check \phi,  
  \check
e (\tilde A + \tilde b) a^{-1}\sqrt b
  ( A
    -  a)\phi}_{\vG_r}
    \end{align*}
    on the interval $[\check r, \infty)$ that $C_-=\lim_{r\to \infty}
    \inp{ g, a^{-1} \xi_-(r)}_{\vG}$ exists and in fact is given by
    $C_-=0$.  Note that the analogue of the expression $T_1$ of Lemma
    \ref{lem:fundaM2} of the derivative vanishes, and that the
    corresponding terms $T_2,\dots, T_5$ are integrable (uniformly in
    $\check r$).  For example it follows from Lemma
    \ref{lemma:L_10b}  and Corollary
    \ref{cor:C_10} that
\begin{align}
  \label{eq:angEst}
  \|p'\bar {\check e}{\tfrac{b}{\bar a}}\check \phi\|_{\vG_r}\leq C
  \parb{(r')^{1/2}\|p'u(r')\|_{\vG_{r'}}+\|u(r')\|_{\vG_{r'}}}r^{-1/2};\,r\geq
  r'.  
\end{align} (Here $C$  may depend on the support of $g$.) This estimate can be applied with  $r'=2R_n$
 (for example) to treat the analogue of the expression $T_3$ of Lemma \ref{lem:fundaM2}.

\smallskip
\noindent
\textit{Step I\hspace{-.1em}I.} $C_+$ exists. Similarly we
have
\begin{align*}
       2 \i  \tfrac \d{\d r}\inp{ \check \xi(r), 
a^{-1}\xi_+(r)}_{\vG}=-\inp{\sqrt b \check \phi,  
  (\tilde A - \tilde b)  a^{-1}\sqrt b( A +  a)
  \phi}_{\vG_r}
\end{align*} To show that $C_+$ exists it suffices to argue that the
derivative is integrable (since now there is no factor $\check e$ and
whence no   $\check r$-dependence to control). As
in Step I there are four terms to consider, say  $T_2,\dots,
    T_5$. More precisely these terms are the contributions
    from four terms arising by the following computation. We
compute using \eqref{eq:AtildeA}, \eqref{eq:a_com3a}  and
in  the last step \eqref{eq:decH}
    \begin{align*}
    &\tfrac a{\sqrt b}   
  (\tilde A - \tilde b)  \tfrac {\sqrt b} a( A +  a)\\
&= (A-\bar a)\tilde\eta(A+a) +O(r^{-\kappa})(A-a)+O(r^{-\kappa})\\
&= (A- a)\tilde\eta(A+a) +\tilde\eta4\i  (\Im a)a+\parb{\tilde\eta2\i(\Im a)
  +O(r^{-\kappa})}(A-a)+O(r^{-\kappa})\\
&= (A+ a)\tilde\eta(A-a) +\tilde\eta\parb{2(p^r a)+4\i  (\Im a)a}+O(r^{-\kappa/2})(A-a)+O(r^{-\kappa})\\
&= (A+ a)\tilde\eta(A-a) +O(r^{-\kappa/2})(A-a)+O(r^{-\kappa})\\
&= 2(H-\lambda)-L+O(r^{-\kappa/2})(A-a)+O(r^{-\kappa}).
\end{align*} We can now proceed as in Step I. In
particular we can   use
\eqref{eq:angEst} with $\check e=1$ to treat the term $T_3$ which is
the analogue of $T_3$ of Lemma \ref{lem:fundaM2}. 

\end{proof}

\begin{lemma}
  \label{lemma:s1} 
  The quantity $F$ is the  strong  limit
  \begin{align*}
F=\vGlim_{r\to
    \infty} \xi(r).
  \end{align*}
\end{lemma}
\begin{proof}
   Due to Corollary ~\ref{cor:radiation-conditions} there exists $C_1>0$ and a  sequence  $r_n\to \infty$ such that 
  \begin{align}\label{eq:Ra40}
    \|(A-a)\phi\|_{\vG_{r_n}}^2+\|p'\phi\|_{\vG_{r_n}}^2\leq C_1/r_n.
  \end{align} 

To show the existence of the strong limit it suffices to show that
\begin{align}\label{eq:bLi1}
  \lim_{n\to \infty}\inp{\xi(r_n), F-\xi(r_n)}_{\vG}=0.
  \end{align}
  Due to \eqref{eq:Ra40} it suffices in turn to show that 
\begin{align}\label{eq:bLi2}
  \lim_{n\to \infty}\lim_{m\to \infty}\inp{\xi(r_n), (a^{-1}\xi_+)(r_m)- (a^{-1}\xi_+)(r_n)}_{\vG}=0.
\end{align} Now we claim that proceeding as in Step I\hspace{-.1em}I of the proof of Lemma
\ref{lemma:w1} (replacing $g\to \xi(r_n)$ and now integrating from
$r_n$)   indeed 
\eqref{eq:bLi2} follows. Note for the analogue term $T_3$ that  Lemma
    \ref{lemma:L_10b}  and Corollary
    \ref{cor:C_10} yield
\begin{align*}
  \begin{split}
  \|p' {\tfrac{b}{\bar a}}\check\phi_n\|_{\vG_r}\leq C_2
  \parb{r_n^{1/2}\|p'\phi\|_{\vG_{r_n}}&+\|\phi\|_{\vG_{r_n}}}r^{-1/2};\\
  &r\geq r_n,\quad\check\phi_n=
  \mathrm e^{\mathrm
      i(r_n-r)(\tilde A-\tilde b)}[\phi]_{S_{r_n}}.
  \end{split}
\end{align*} In combination with \eqref{eq:Ra40} we then obtain
\begin{align*}
  \|p' {\tfrac{b}{\bar a}}\check\phi_n\|_{\vG_r}\leq C_3r^{-1/2},
\end{align*} which suffices for integrability  and  smallness
$o(n^0)$ of the integral (due to \cs and
Corollary ~\ref{cor:radiation-conditions}) essentially showing \eqref{eq:bLi2}.
\end{proof}

\begin{proof}[Proof of Theorem~\ref{thm:strong}]

  The definition \eqref{eq:disF} is justified by Lemma
  \ref{lemma:s1}.  Clearly (\ref{eq:fund}) follows from Lemma
  \ref{lem:normBasi}.
  
It remains to show  the continuity statement.  We shall basically follow the
scheme of \cite{Co}. Due to (\ref{eq:fund}), Corollary
\ref{cor:12.7.2.7.9b} and the density of $C^\infty_\c(S)\subseteq \vG$
the continuity of the map 
  $I\ni \lambda\to F^{+}(\lambda)\psi\in\vG $ follows if we can show continuity of
  $\inp{F^{+}(\cdot)\psi,g}_{\vG}$ for any  $\psi\in \vH_{1+}$
  and  $g\in
  C^\infty_\c(S)$.
    Let us for any such $g$ introduce approximate generalized
    eigenfunctions 
  $\phi^\pm[g]\in \vN\cap B^*$ by specifying in the spherical coordinates 
  \begin{align}\label{eq:apGen}
\begin{split}
   \phi^\pm[g](r,\sigma)= \bar \chi_n(r) b^{-1/2}(r,\sigma)
  \exp \biggl( \int_{r_0}^{r}\parbb{\pm \i\tilde
    b- \tfrac
    12\mathop{\mathrm{div}}\tilde\omega}(s,\sigma)\,\mathrm{d}s\biggr)g(\sigma).
\end{split} 
\end{align} The factor $\bar \chi_n $ is 
chosen as a cut-off function, possibly depending  on  $I$ and the support of $g$,
 to assure
 the property  $\phi^\pm[g]\in\vN$ (as in the proof of Lemma \ref{lemma:w1}). Note that these vectors  are essentially the same as  those
introduced at \eqref{eq:gen1B} (this is why we are using the same notation). 
  We shall use 
 the
previous notation $\xi,\xi_+$ and $\phi$. First calculate (for $m$
sufficiently large)
\begin{align*}
  &2\inp{\psi,
    \chi_m\phi^+[g]}=\\& -\i\inp{(A+\bar a)\phi,\chi'_m\phi^+[g]}
+\inp{(A+\bar a)\phi, \chi_m\tilde\eta(A-a)\phi^+[g]}
\\&\phantom{{}={}}
+\int^\infty _{r_0} \chi_m(r)\parb{\inp{p'\phi, p'\phi^+[g]}_{\vG_r}+\inp{\phi,
    O(r^{-\kappa})\phi^+[g]}_{\vG_r}}\,\d r,
\end{align*} cf. \eqref{eq:decH}. Note  for the first term to the
right that 
\begin{align*}
  \inp{(A+\bar a)\phi,\chi'_m\phi^+[g]}=\int_{r_0}^\infty
  \chi'_m(r) \inp{(A+\bar a)\phi, &\phi^+[g]}_{\vG_r}\,\d r,
\end{align*} and 
\begin{align*}
    \inp{(A+\bar a)\phi, &\phi^+[g]}_{\vG_r}=\inp{(\bar a-a)\phi, \phi^+[g]}_{\vG_r}+2\inp {(b^{-1}\xi_+)(r),
    g}_{\vG}\\
&=2\inp {(a^{-1}\xi_+)(r),
  g}_{\vG}+2\inp {((b^{-1}-a^{-1})\xi_+)(r), g}_{\vG} +2\i\inp{(\Im a)\phi, \phi^+[g]}_{\vG_r}.
  \end{align*}

Note for the second  term to the
right  that 
\begin{align}\label{eq:BasA}
  \begin{split}
  \parb{A-a+O(r^{-\kappa})}\phi^+[g]&=\parb{b^{-1/2}
  (A-b)b^{1/2}+\tfrac {\i} 2 \nabla_\omega \ln |\d r|^2}\phi^+[g]
\\&=O(r^{-\infty}),  
  \end{split}
\end{align} in fact  here the last term vanishes for $r$ large. These
considerations  allow
us to take   $m\to \infty$ and we obtain  (using that $\inp {(a^{-1}\xi_+)(r),
  g}_{\vG}\to\inp {F^+(\lambda)\psi,
  g}_{\vG}$ for $r\to \infty$) that 
 \begin{align*}
2\i\inp {F^+(\lambda)\psi,g}_{\vG}
&=2\inp{\psi, \phi^+[g]}
-\inp{(A+\bar a)\phi, \tilde\eta(A-a)\phi^+[g]}
\\&\phantom{{}={}}
-\int^\infty _{r_0} \parb{\inp{p'\phi, p'\phi^+[g]}_{\vG_r}+\inp{\phi,
    O(r^{-\kappa})\phi^+[g]}_{\vG_r}}\,\d r.
\end{align*} By tracing the $\lambda$-dependence we conclude  from this
representation that indeed $\inp {F^+(\cdot)\psi,
  g}_{\vG} $ is continuous. 
\end{proof}

The  last formula reads more compactly (although less precisely)
\begin{align}\label{eq:pair}
  \i\inp {F^+(\lambda)\psi, g}_{\vG} -\inp{\psi,
    \phi^+[g]}=-\inp{\phi, (H-\lambda)\phi^+[g]},
\end{align} where the right-hand side is given an interpretation
very similar to \eqref{eq:decH}.

\subsection{Properties of distorted Fourier transform}\label{subsec:Stationary wave operator} 

  We
first prove
Proposition~\ref{prop:dist-four-transf}. Throughout  this subsection we continue to consider only the upper sign.

\begin{proof}[Proof of Proposition~\ref{prop:dist-four-transf}]
  We first prove (\ref{eq:extbF2}) for $\psi\in B$. 
It is a direct consequence of Theorem~\ref{thm:strong} 
that (\ref{eq:extbF2}) holds for $\psi\in \mathcal H_{1+}$,
and we have already seen that the left-hand side extends
continuously in $\psi\in B$.
Hence it suffices to show the existence and continuity of the right-hand side in $B$.
By Theorem~\ref{thm:12.7.2.7.9} these matters  reduce to the following
estimate, a version of which appears in a  similar context in \cite{Sk}:
 For any
$\psi\in B$
\begin{align}
 \sup_{R>r_0}\left\| -\!\!\!\!\!\!\int_R
  \xi(r)\,\d r\right\|_{\mathcal G}\leq C \|\phi\|_{B^*}.
\label{eq:14.4.25.15.40}
\end{align}   
 To show \eqref{eq:14.4.25.15.40} we write $\xi(r)=\e^{\i
  (r-r_0)(\tilde A^{\rm ex }-\tilde
      b^{\rm ex })}u(r) $ and  note that 
\begin{align*}
  \left\| -\!\!\!\!\!\!\int_R
  \xi(r)\,\d r\right\|_{\mathcal G}\leq  -\!\!\!\!\!\!\int_{R}\left\|
  u(r)\right\|_{\mathcal G_{r}}\,\d r.
\end{align*}

Next for any $R>r_0$ we choose $n\ge 0$ such that $R_n\le 2R<R_{n+1}$
and use
the Cauchy--Schwarz inequality  to obtain
\begin{align*}
\left\| -\!\!\!\!\!\!\int_R
  \xi(r)\,\d r\right\|_{\mathcal G}^2
&\le
-\!\!\!\!\!\!\int_{R}\|u(r)\|_{\mathcal G_r}^2\,\d r\\
&\le
2\sum_{\nu=0}^n(R_\nu/R_{n})R_\nu^{-1}\|F_\nu \sqrt b \phi\|^2\\
&\le 4\|\sqrt b \phi\|_{B*}^2.
\end{align*}
Hence we have shown  \eqref{eq:14.4.25.15.40} and therefore  that 
(\ref{eq:extbF2}) holds for any $\psi\in B$. 

To show that $\intR \xi(r)\,\d r
\to F^\pm(\lambda)\psi$ locally
  uniformly in $\lambda$ we can assume that $\psi\in \mathcal
  H_{1+}$. Next  note that
  \begin{align*}
    \left\|F^+(\lambda)\psi--\!\!\!\!\!\!\int_{R_1} \xi(r_1)\,\d r_1\right\|_{\vG}^2 \leq \lim_{R_2\to
      \infty} -\!\!\!\!\!\!\int_{R_1} -\!\!\!\!\!\!\int_{R_2} \|\xi(r_2)-\xi(r_1)\|_{\vG}^2\,\d r_2 \d r_1.
  \end{align*} We look at  $R_2>2R_1$ and  write
  \begin{align*}
    \|\xi(r_2)-\xi(r_1)\|_{\vG}^2&=\|\xi(r_2)\|_{\vG}^2-\|\xi(r_1)\|_{\vG}^2-2\Re
    \inp{ \xi(r_1),\xi(r_2)-\xi(r_1)}_{\vG}\\
&=\|\xi(r_2)\|_{\vG}^2-\|\xi(r_1)\|_{\vG}^2
\\&\phantom{{}={}}
-2\Re\inp{ \xi(r_1),(a^{-1}\xi_-)(r_2)}_{\vG}+2\Re
    \inp{ \xi(r_1), (a^{-1}\xi_-)(r_1)}_{\vG}
\\&\phantom{{}={}}
-2\Re
    \inp{ \xi(r_1),(a^{-1}\xi_+)(r_2)- (a^{-1}\xi_+)(r_1}_{\vG}.
  \end{align*} The first term contributes to the $R_2$-limit by
  $2\mathop{\mathrm{Im}}\langle\psi,\phi\rangle$ due to Lemma
  \ref{lem:fundaM}, the second term by
  $--\!\!\!\!\!\!\int_{R_1}\|\xi(r_1)\|_{\vG}^2\,\d r_1$, the third
  term by $0$ (cf. Corollary
  \ref{cor:radiation-conditions}), the fourth
  term by $-\!\!\!\!\!\int_{R_1} 2\Re \inp{ \xi(r_1),
    (a^{-1}\xi_-)(r_1)}_{\vG}\,\d r_1$ and the last term by $o(R^0_1)$
  (by the proof of Lemma \ref{lemma:s1}). We readily check that
  $-\!\!\!\!\!\int_{R_1} 2\Re \inp{ \xi(r_1),
    (a^{-1}\xi_-)(r_1)}_{\vG}\,\d r_1\to 0$ locally uniformly in
  $\lambda$, and similarly that
  $-\!\!\!\!\!\int_{R_1}\|\xi(r_1)\|_{\vG}^2\,\d r_1\to
  2\mathop{\mathrm{Im}}\langle\psi,\phi\rangle$ and the
  quantity $o(R^0_1)\to 0$ locally uniformly in $\lambda$. This is by 
   Corollary \ref{cor:radiation-conditions} and the proofs of Lemmas
  \ref{lem:fundaM} and \ref{lemma:s1}, respectively.

  Next, we note that  $B\cap\mathcal H_{\mathcal I}$ is dense in $\mathcal H_{\mathcal
  I}$. In fact for any $\psi\in B$ and any $f\in C_\c^\infty(\vI)$ the
vector $f(H)\psi\in B$, cf. \cite [Theorem 14.1.4]{Ho2}. Due to Stone's formula and \eqref{eq:fund} we have
\begin{align*}
\|F^+\psi\|_{\widetilde{\mathcal H}_{\mathcal I}}=\|\psi\|_{\mathcal H_{\mathcal I}};\quad \psi\in B\cap\mathcal H_{\mathcal I},
\end{align*}
 so the operator $F^+$ extends as an isometry from $B\cap\mathcal H_{\mathcal
  I}\subseteq \mathcal H_{\mathcal I}$ to  an isometry
$\mathcal H_{\mathcal I}\to \widetilde{\mathcal H}_{\mathcal I}$
(denoted also by $F^+$).
 It remains to show that $F^+ H_\mathcal I\subseteq
M_\lambda F^+$ or equivalently that $F^+ (H_\vI- \i)^{-1}= (M_\lambda- \i)^{-1}F^+$.
  Whence it suffices to show  that 
\begin{align*}
  F^+ (H- \i)^{-1}\psi= (M_\lambda- \i)^{-1}F^+\psi\text{
    for any }\psi\in B\cap \mathcal H_{\mathcal I}.
\end{align*} Using the resolvent equations 
\begin{align}\label{eq:2equ}
  R(\lambda\pm\i 0)R(\i)=(\lambda - \i)^{-1}R(\lambda\pm\i
  0)-(\lambda - \i)^{-1}R(\i),
\end{align}
 we obtain 
\begin{align*}
  F^+(\lambda) R(\i)\psi=\lim_{R\to \infty }-\!\!\!\!\!\!\int_R (\lambda -
\i)^{-1}\xi(r)\,\d r=(\lambda -
\i)^{-1}F^+(\lambda)\psi.
\end{align*} Note that due to \cs the second term of \eqref{eq:2equ} does not
contribute to  the limit.
\end{proof}

Now we embark on   the proof of
{Theorem~\ref{thm:dist-four-transf}. Note that under the conditions of
  the lemma below a priori we can write
  \begin{align*}
    R(\mathrm i)L\phi^+[g] =\wBstarlim_{m\to \infty}R(\mathrm
  i)\chi_m L\phi^+[g] \in B^*,
  \end{align*} meaning that for any $\check \psi\in B$
\begin{align*}
  \langle R(\mathrm i)L\phi^+[g], \check \psi\rangle&=\langle
  p'\phi^+[g], p'R(-\mathrm i)\check \psi\rangle\\&=\lim _{m\to \infty}\int \chi_m(r)\langle
  p'\phi^+[g], p'R(-\mathrm i)\check \psi\rangle_{\vG_r}\,\d r.
\end{align*}
\begin{lemma}\label{lem:14.5.1.16.30} Suppose
  Condition~\ref{cond:12.6.2.21.13bb}. For any 
  $g\in
  C^\infty_\c(S)$ let $\phi^+[g]\in \vN\cap B^*
$ be  given by 
  \eqref{eq:apGen} (where $n$ is large, but locally independent of
  $\lambda>\lambda_0$). Then 
\begin{align*}
    R(\mathrm i)L\phi^+[g] =\wBlim_{m\to \infty}R(\mathrm
  i)\chi_m L\phi^+[g] \in B,
  \end{align*} meaning that the vector $\psi=R(\mathrm i)L\phi^+[g]$, a
  priori in $B^*$, actually  is in $B$ and that for all $\check \phi\in B^*$ 
\begin{align*}
  \langle \psi, \check \phi\rangle&=\langle
  p'\phi^+[g], p'R(-\mathrm i)\check \phi\rangle\\&=\lim _{m\to \infty}\int \chi_m(r)\langle
  p'\phi^+[g], p'R(-\mathrm i)\check \phi\rangle_{\vG_r}\,\d r.
\end{align*} 
In fact 
\begin{align}
R(\mathrm i)L\phi^+[g] \in B \cap\mathcal H^1.\label{eq:14.5.1.16.14}
\end{align} 
\begin{comment}
  Here and below, a priori $R(\mathrm i)L\phi^+[g] \in B^*$ is given by
the recipe  (for any  $\check \psi\in B$)
\begin{align*}
  \langle R(\mathrm i)L\phi^+[g], \check \psi\rangle&=\langle
  p'\phi^+[g], p'R(-\mathrm i)\check \psi\rangle\\&=\lim _{m\to \infty}\int \chi_m(r)\langle
  p'\phi^+[g], p'R(-\mathrm i)\check \psi\rangle_{\vG_r}\,\d r.
\end{align*}
\end{comment}
\end{lemma}
When 
  Condition~\ref{cond:12.6.2.21.13bb} (\ref{item:14.5.1.8.30}) holds,
the proof of Lemma~\ref{lem:14.5.1.16.30} is rather simple since 
an essential estimate is already done in Corollary ~\ref{cor:C_10}.
However, when  Condition~\ref{cond:12.6.2.21.13bb} (\ref{item:14.5.1.8.31})  holds,
we need to estimate one more derivative than in Corollary~\ref{cor:C_10}
and this requires  technical computations, see Lemma~\ref{lem:14.3.5.0.17}.
Hence for the moment we postpone the proof of
Lemma~\ref{lem:14.5.1.16.30} 
and first prove Theorem~\ref{thm:dist-four-transf} which in turn 
  clearly is a direct consequence of Lemma~\ref{lem:14.5.1.16.30}
and the following lemma.
\begin{lemma}\label{lem:14.5.4.17.5}
  If for all 
   $g\in C^\infty_{\mathrm c}(S)$  the  vector  $\phi^+[g]$ of
   \eqref{eq:apGen}  (depending on $\lambda>\lambda_0$) satisfies
  \eqref{eq:14.5.1.16.14}, then $F^+\colon \mathcal H_{\mathcal I}\to
  \widetilde{\mathcal H}_{\mathcal I}$ is a unitary diagonalizing
  transform.
\end{lemma}
\begin{proof}  We consider for  $\lambda\in \mathcal I$  the operator 
$F^+(\lambda)\colon B\to\mathcal G$ of 
Proposition~\ref{prop:dist-four-transf}. From
the same result  we know that $F^+$  is  an isometry.

\smallskip
\noindent
\textit{Step I.} First we show that
$F^+(\lambda)$ has dense range. This is equivalent to showing that  $F^+(\lambda)^*:\vG\to B^*$
  is injective, and for that we will   use the representation
  \eqref{eq:pair}  of
  $F^+(\lambda)^*g$ 
 for $g\in C^\infty_{\mathrm c}(S)$. 
For the  term on the right-hand side for such $g$  we claim  the
  bound
  \begin{align}\label{eq:secoTerm}
    (A+a)R(\lambda-\i0)(H-\lambda)\phi^+[g]\in B_0^*.
  \end{align} To obtain \eqref{eq:secoTerm} we use \eqref{eq:2equ},
  reducing the problem to show that 
\begin{align*}
    (A+a)\psi_+,\,(A+a)R(\lambda-\i0)\psi_+\in B_0^*;\quad \psi_+=R(\mathrm i)(H-\lambda)\phi^+[g].
  \end{align*} 
  The interpretation of $\psi_+$ is given as in the proof of
  Theorem~\ref{thm:strong}  which amounts to expanding
  $(H-\lambda)\phi^+[g]$ into a sum of three terms,
  cf. \eqref{eq:decH}.  Each term is in $B \cap\mathcal H^1$, so
  consequently $\psi_+\in B \cap\mathcal H^1$. Note that at this point
  we use \eqref{eq:BasA} and \eqref{eq:14.5.1.16.14}. Using that
  $\psi_+\in B \cap\mathcal H^1$ we deduce \eqref{eq:secoTerm} by  
  Corollary \ref{cor:radiation-conditions}. 

  Now using \eqref{eq:BasA}, \eqref{eq:pair}  and \eqref{eq:secoTerm} we obtain
  \begin{align}\label{eq:uniqFtoG}
    g=\vGlim_{R\to \infty}-\!\!\!\!\!\!\int_R \e^{\i (r-r_0)(\tilde A^{\rm ex }-\tilde
      b^{\rm ex })}\bigl [\tfrac {b^{1/2}\i}{2 a}(a+A)F^+(\lambda)^*g\bigr
    ]_{|S_r}\,\d r,
  \end{align} and since 
  \begin{align*}
    (a+A)F^+(\lambda)^*=(\lambda-\i)\big \{(a+A)R(\mathrm
  i)\big \}F^+(\lambda)^*\in\vB (\vG,B^*),
  \end{align*}
  we conclude \eqref{eq:uniqFtoG} for all $g\in \vG$ by a  continuity
  argument essentially identical with the one given  in the first part of the proof of Proposition~\ref{prop:dist-four-transf}. In particular indeed
  $F^+(\lambda)^*:\vG\to B^*$ is injective.

\smallskip
\noindent
\textit{Step II.}
We prove the unitarity of $F^+\colon \mathcal H_{\mathcal
  I}\to \widetilde{\mathcal H}_{\mathcal I}$. Since we know $F^+ H_\mathcal I\subseteq
M_\lambda F^+$ from Proposition~\ref{prop:dist-four-transf} it then
follows that  $F^+ H_\mathcal I=
M_\lambda F^+$, and the proof is done. 

By using Proposition~\ref{prop:dist-four-transf} (possibly in
combination with \eqref{eq:HeSj}) we obtain that 
 \begin{align}\label{eq:diagf}
   F^+(\lambda)f(H)\psi=f(\lambda)F^+(\lambda)\psi\text{ for all }f\in
   C^\infty_\c(\vI)\text { and }\psi\in B.
\end{align} Assuming   $g(\cdot)\in
 \mathop{\mathrm{ker}}(F^+)^*\subseteq \widetilde{\mathcal
   H}_{\mathcal I}$ it suffices to  show that $g(\lambda)=0$ for a.e.\
 $\lambda\in \mathcal I$. We shall mimic the proof of 
 \cite[Theorem 1.1]{ACH}. For any $f\in
C^\infty_\c(\vI)$ and  $\psi\in B$
\begin{align*}
\int_{\vI}f(\lambda)\bigl\langle g(\lambda),F^+(\lambda)\psi\bigr\rangle_{\mathcal G}\,\mathrm d\lambda
=\bigl\langle (F^+)^*g(\cdot),f(H)\psi\bigr\rangle_{\mathcal H_{\vI}} 
=0.
\end{align*}
We apply this to the elements of a countable and dense
subset, say $\{\psi_k\}_{k=1}^\infty\subseteq B$, and conclude that there exists a set $N\subseteq \mathcal I$ of measure $0$ such that 
\begin{align*}
\bigl\langle g(\lambda),F^+(\lambda)\psi_k\bigr\rangle_{\mathcal
  G}=0\text { for all }k\in \N\text{ and  }\lambda\in \mathcal I\setminus N.
\end{align*}  Since  $\{F^+(\lambda)\psi_k\}_{k=1}^\infty\subseteq
\vG$ is dense  (by Step I) we conclude   that 
 $g(\cdot)=0$.
Hence $F^+\colon \mathcal H_{\mathcal I}\to\widetilde{\mathcal
  H}_{\mathcal I}$ is surjective and therefore unitary.
\end{proof}
\begin{remarks}\label{remark:uniq}
  We used above the representation in terms of  the vectors $\phi^{\pm}[g]\in \vN$ of
  \eqref{eq:apGen} (here stated for both signs)
  \begin{subequations}
  \begin{align}\label{eq:rep1}
    \begin{split}
    \pm\i F^\pm(\lambda)^* g&= \phi^\pm[g]- \psi_\pm-(\lambda-\i)R(\lambda\mp\i 0)\psi_\pm;\\&g\in
    C_\c^\infty(S),\quad \psi_\pm=R(\mathrm i)(H-\lambda)\phi^\pm[g]\in B \cap\mathcal H^1.  
    \end{split}
  \end{align} For
  comparison we obtain using  Corollary \ref{cor:13.9.9.8.23} 
\begin{align*}
    0= \phi^\pm[g]- \psi_\pm-(\lambda-\i)R(\lambda\pm\i 0)\psi_\pm.
  \end{align*} In particular 
  \begin{align*}
    \phi^\pm[g]- (\lambda-\i)R(\lambda\pm \i 0)\psi_\pm\in B_0^*,
  \end{align*}
 which leads to 
  \begin{align}\label{eq:rep3}
    g= (\lambda-\i)F^\pm(\lambda)\psi_\pm;\quad g\in
    C_\c^\infty(S),\quad \psi_\pm=R(\mathrm i)(H-\lambda)\phi^\pm[g].
  \end{align} 
\end{subequations} The  formulas \eqref{eq:rep1} and \eqref{eq:rep3}
will be used in Section \ref{subsec:Scattering matrix and
  characterization of generalized eigenfunctions}, however we stress
that   the vectors of 
  \eqref{eq:apGen} are given with  a cutoff to make them  elements
  of $\vN$. The vectors  $\phi^\pm[g]$  given by
  \eqref{eq:gen1B} do not necessarily enjoy this property, and in fact
  the above formulas might not be valid in
  general when
  $\phi^\pm[g]$  (with $g\in
    C_\c^\infty(S)$) are given by
  \eqref{eq:gen1B}.
\end{remarks}

\begin{proof}[Proof of Theorem~\ref{thm:dist-four-transf}]
  The statement is obvious from Proposition 
  \ref{prop:dist-four-transf} and Lemmas~\ref{lem:14.5.1.16.30} and
  \ref{lem:14.5.4.17.5}.
\end{proof}

Finally we prove Lemma~\ref{lem:14.5.1.16.30}.
We begin with a technical estimate required for the  case 
 of Condition~\ref{cond:12.6.2.21.13bb} (\ref{item:14.5.1.8.31}).

\begin{lemma}\label{lem:14.3.5.0.17} 
Suppose   Condition~\ref{cond:12.6.2.21.13bb} (\ref{item:14.5.1.8.31})
and let $\check \sigma<\min\{\sigma,\tau,\rho\}$.
  Then for any  compact
interval $I\subseteq \mathcal I$ there exists a constant $C>0$ such
that for all $\lambda\in I$, 
   $\check r\geq r_0$ and  $u\in C^2_{\mathrm
    c}(S_{\check r})$,  
  the function 
  $u(r)=\mathrm e^{\mathrm i(\check r-r)(\tilde A-\tilde b)}u\in
   C^2_{\mathrm
    c}(S_{r})$ for   $r\ge \check r$ and 
satisfies
  \begin{align}
  \begin{split}
  \|L_ru(r)\|_{\mathcal G_r}
  &\le C
 \Bigl(
  (1/r)^{\check \sigma}\|u\|_{\mathcal G_{\check r}}
  +(\check r{}^{1/2}/r)^{\check \sigma}\|\nabla'u\|_{\mathcal G_{\check r}}
  +(\check r/r)^{{\check\sigma}}\|\nabla'{}^2u\|_{\mathcal G_{\check r}
  }\Bigr),
    \end{split}
    \label{eq:13.9.23.16.40b}
  \end{align}
where  $\nabla'$ is the Levi--Civita connection associated with the induced 
Riemannian metric on the $r$-sphere $S_{\check r}$, i.e.\ $g_{\check r}=\iota_{\check r}^*g$.
\end{lemma} 
\begin{proof}
\textit{Step I.}\quad
In this proof we make use of the geometric derivatives of a mapping
presented in Subsection~\ref{sec:12.10.11.16.40}.
For any $t\in [\check r-r,0]$ and $x\in \tilde y(-t,S_{r+t})\subseteq S_r$
we define 
the quantity $\bigl(\Delta_r\tilde y^{\alpha}(t,x)\bigr)_{\alpha=2,\dots,d}$
as a tangent of $S_{r+t}$ at $\tilde y(t,x)$
by
\begin{align*}
\Delta_r\tilde y^{\alpha}(t,x)=\ell^{ij}(\nabla^2\tilde y)^\alpha{}_{ij}(t,x).
\end{align*}
Note that here again we adopt the same convention for the Roman and Greek
indices as in the proof of Lemma~\ref{lemma:L_10b}.
Let us abbreviate simply $\tilde y=\tilde y(t,x)$ and use the formula \eqref{eq:160220}.
Then we can write the above quantity more explicitly:
\begin{align}
\begin{split}
\Delta_r\tilde y^{\alpha}(t,x)
=\Delta_r\tilde y^\alpha
=
\ell^{ij}\partial_i\partial_j\tilde y^\alpha
-\ell^{ij}\Gamma_{ij}^k\partial_k\tilde y^\alpha
+\ell^{ij}\Gamma_{\beta\gamma}^\alpha
(\partial_i\tilde y^\beta)
(\partial_j\tilde y^\gamma).
\end{split}
\label{eq:14.3.4.3.37A}
\end{align}
Hereinafter in this proof all the indices run only 
over the angular components in the spherical coordinates.
Whence the distinguished radial index $\alpha,\beta,\gamma,\delta,i,j,k=r$ does
not enter in the summations there.
At this point it should be noted that 
it does not cause a confusion,
because $\Gamma^\alpha_{\beta\gamma}=\Gamma'{}^\alpha_{\beta\gamma}$ for
$\alpha,\beta,\gamma\neq r$ where $\Gamma'$ is the Christoffel symbol for
the submanifold $S_{r+t}$ (and similarly for $\Gamma_{ij}^k$).

Similarly, the quantity $\bigl(L\tilde y^{\alpha}(t,x)\bigr)_{\alpha=2,\dots,d}$
is defined as a tangent of $S_{r+t}$ at $\tilde y(t,x)$ by
\begin{align}
L\tilde y^{\alpha}(t,x)
=
L\tilde y^{\alpha}
&=-\Delta_r\tilde y^{\alpha}
+\tfrac 12 (\partial_i \ln |\d r|^2)
\ell^{ij}\partial_ j\tilde y^{\alpha}, \label{eq:14.3.4.3.37AA}
 \end{align}
cf.\ \eqref{eq:lsubr}.
We note that the quantity \eqref{eq:14.3.4.3.37AA} is defined to 
satisfy that for any function $v\in C^2(S_{r+t})$
\begin{align}\label{eq:basCombAA}
  \begin{split}
  L\bigl[v(\tilde y(t,x))\bigr]
  &
=\bigl[L\tilde y^{\alpha}(t,x)\bigr]
 \bigl(\partial _{\alpha} v\bigr)(\tilde y(t,x))
\\&\phantom{{}={}}
-\ell^{ij}(x)
\bigl[\partial_i\tilde y^{\alpha}(t,x)\bigr]
\bigl[\partial_j\tilde y^{\beta}(t,x)\bigr]
\bigl({\nabla'^2} v\bigr)_{\alpha\beta} (\tilde y(t,x)), 
  \end{split}
\end{align} 
where $\nabla'$ is the covariant derivative on the
submanifold $S_{r+t}$,

\smallskip
\noindent
\textit{Step I\hspace{-.1em}I.}\quad
Next let us estimate the tensor \eqref{eq:14.3.4.3.37A}.
We claim that for any
${\check\sigma< \min\{\sigma,\tau\}}$ there exists $C_1>0$ such that
uniformly in $r\ge \check r\ge r_0$, $t\in [\check r-r,0]$ and 
$x\in \tilde y(r-\check r,S_{\check r})\subseteq S_r$
\begin{align}
F(t):=\ell_{\alpha\beta}
(\Delta_r\tilde y^\alpha)
(\Delta_r\tilde y^\beta)
=|\Delta_r\tilde y|^2\le 
C_1(r+t)^{\check\sigma}/r^{{2\check\sigma}}.\label{eq:14.3.4.3.39}
\end{align} 
To prove \eqref{eq:14.3.4.3.39} we shall  establish  a differential inequality for
$F(t)$, cf.\ the proof of Lemma \ref{lemma:L_10b}. Obviously,  using the
shorthand notation $v^\alpha=\Delta_r\tilde y^\alpha(t,x)$ for 
$\alpha\neq r$,
\begin{align*}
F'(t)
= 
\Bigl((\partial_t \ell_{\alpha\beta})v^\alpha
+
2\ell_{\alpha\beta}\bigl[\partial_t \Delta_r\tilde y^\alpha\bigr]
\Bigr)
v^\beta.
\end{align*} 
Next recall
\begin{subequations}
\begin{align}\label{eq:1G}
  {\Gamma'}_{\delta\gamma}^\alpha=\tfrac 12
  \ell^{\alpha\beta} \parb{\partial_
    \delta\ell_{\beta\gamma}+\partial_
    \gamma\ell_{\beta\delta}-\partial_ \beta\ell_{\delta\gamma}},
\end{align} and similarly for the Christoffel symbol on $M$. The latter yields
\begin{align}\label{eq:2G}
  |\d r|^2\partial_r \ell_{\alpha\beta}=-2\Gamma
  ^r_{\alpha\beta}=2(\nabla^2 r)_{\alpha\beta};\,\alpha,\beta\neq r.
\end{align} 
\end{subequations} 
Then we calculate using \eqref{eq:2G} and that $\tilde
\omega^\eta=\delta_{\eta r}$ (here and below possibly $\eta=r$)
\begin{align} 
\begin{split}
  &(\partial_t \ell_{\alpha\beta})v^\alpha
+
2\ell_{\alpha\beta}[\partial_t \Delta_r\tilde y^\alpha]\\
&= (\partial_\eta \ell_{\alpha\beta})\tilde\omega^\eta v^\alpha
+2\ell_{\alpha\beta}\parbb{\ell^{ij}\partial_i\partial_j\tilde\omega^\alpha
-\bigl[\ell^{ij}\Gamma_{ij}^k\bigr]\partial_k\tilde\omega^\alpha
\\&\phantom{{}={}}
+(\partial_\eta\Gamma_{\delta\gamma}^\alpha)\tilde\omega^\eta
\ell^{ij}[\partial_i\tilde y^\delta]
[\partial_j\tilde y^\gamma]
+2\Gamma_{\delta\gamma}^\alpha
\ell^{ij}
(\partial_i\tilde\omega^{\delta})
[\partial_j\tilde y^\gamma]}
\\
&= (\partial_r \ell_{\alpha\beta}) v^\alpha
+2\ell_{\alpha\beta} (\partial_r\Gamma_{\delta\gamma}^\alpha)
\ell^{ij}[\partial_i\tilde y^\delta]
[\partial_j\tilde y^\gamma]
\\
&= 2 |\d r|^{-2}(\nabla^2
r)_{\alpha\beta}v^\alpha 
+2\ell_{\alpha\beta} (\partial_r\Gamma_{\delta\gamma}^\alpha)
\ell^{ij}[\partial_i\tilde y^\delta]
[\partial_j\tilde y^\gamma].
\end{split}\label{1601311}
\end{align} 
A computation using \eqref{eq:1G}, \eqref{eq:2G}
and \eqref{eq:14.12.27.3.22}  (details skipped) yields
\begin{align*}
  \partial_r{\Gamma'}_{\delta\gamma}^\alpha=(\nabla (|\d
      r|^{-2}\nabla^2 r))_{\delta}{}^{\alpha}{}_\gamma
+(\nabla (|\d r|^{-2}\nabla^2 r))_{\gamma}{}^{\alpha}{}_\delta
-|\d r|^{-2}{(\nabla^3 r)}^{\alpha}{}_{\delta\gamma},
\end{align*} which we can insert on the right-hand side of \eqref{1601311}.  
Using \cS, 
Lemma \ref{lemma:L_10b}  and Condition~\ref{cond:12.6.2.21.13bb}
(\ref{item:14.5.1.8.31}) the contribution from the last 
term of \eqref{1601311}  can
be estimated from below as follows: For any $\epsilon>0$
\begin{align*}
  \cdots \geq-\epsilon(r+t)^{-1}F(t)-\epsilon^{-1}C_2(r+t)^{{2\sigma-\tau-1}}/r^{2\sigma}.
\end{align*} With \eqref{eq:13.9.5.3.30} this leads to the final estimate (seen by
taking $\epsilon\leq \sigma-\check\sigma$)
\begin{align}
F'(t)
\ge {\check\sigma(r+t)^{-1}F(t)
-C_3(r+t)^{2\sigma-\tau-1}/r^{2\sigma}}.
\label{eq:14.5.2.15.18}
\end{align} 
On the other hand we can compute $F(0)=0$ by using
\eqref{eq:14.3.4.3.37A},
and then the Gronwall's inequality yields the claimed estimate
\eqref{eq:14.3.4.3.39}.

\smallskip
\noindent
\textit{Step I\hspace{-.1em}I\hspace{-.1em}I.}\quad
Here we compute and estimate $Lu(r)$.
For simplicity we would like to use 
the spherical coordinates as in \eqref{eq:14.3.5.11.0}.
 To avoid any  confusion concerning base points let us write 
$\sigma\in \tilde y(r-\check r,\mathop{\mathrm{supp}}u)\subseteq S_r$ and 
$\sigma(s)=\tilde y(s-r,\sigma)\in S_s$, $s\in[\check r,r]$.
Then we can write 
\begin{align}
\begin{split}
u(r)(\sigma)
=(\mathrm{e}^{\mathrm{i}(\check r-r)(\tilde A-\tilde b)}u)(\sigma)
=\exp \biggl(\int_{\check r}^r\Bigl(
\mathrm i\tilde b-\tfrac12\mathop{\mathrm{div}}\tilde\omega\Bigr)(s,\sigma(s))
\,\mathrm{d}s\biggr) u(\sigma(\check r)).
\end{split}
\label{eq:1601313}
\end{align}
Using \eqref{eq:1601313} and \eqref{eq:basCombAA}
we obtain the following expression:
\begin{align}
\begin{split}
&Lu(r)(\sigma)\\
&=
[L\sigma^\alpha(\check r)]
(\mathrm{e}^{\mathrm{i}(\check r-r)(\tilde A-\tilde b)}\partial_\alpha u(\check r))(\sigma)
\\&\phantom{={}}
-\ell^{ij}
[\partial_i\sigma^\alpha(\check r)]
[\partial_j\sigma^\beta(\check r)]
\bigl(\mathrm{e}^{\mathrm{i}(\check r-r)(\tilde A-\tilde b)}
(\nabla'{}^2 u(\check r))_{\alpha\beta}\bigr)(\sigma)
\\&\phantom{={}}
+\biggl(\int_{\check r}^{r}
[L\sigma^\alpha(s)]
\bigl[\partial_\alpha
\bigl(\mathrm i\tilde b-\tfrac1{2}
\mathop{\mathrm{div}}\tilde \omega
\bigr)\bigr](s,\sigma(s))\,\mathrm{d}s \biggr)u(r)(\sigma)
\\&\phantom{={}}
-\biggl(\int_{\check r}^{r}
\ell^{ij}[\partial_i\sigma^\alpha(s)]
[\partial_j\sigma^\beta(s)]
\bigl(\nabla'{}^2
\bigl(\mathrm i\tilde b-\tfrac1{2}
\mathop{\mathrm{div}}\tilde \omega
\bigr)\bigr)_{\alpha\beta}(s,\sigma(s))\,\mathrm{d}s \biggr)u(r)(\sigma)
\\&\phantom{={}}
-2\ell^{ij}
\biggl(\int_{\check r}^r
[\partial_i\sigma^\alpha(s)]
\bigl[\partial_\alpha
\bigl( \mathrm i\tilde b-\tfrac1{2}
\mathop{\mathrm{div}}\tilde \omega
\bigr)\bigr](s,\sigma(s))\,\mathrm{d}s \biggr)
\\&\phantom{={}+{}}\cdot
[\partial_j\sigma^\beta(\check r)]
\bigr(\mathrm{e}^{\mathrm{i}(\check r-r)(\tilde A-\tilde b)}\partial_\beta u(\check r)\bigr)(\sigma)
\\&\phantom{={}}
-\ell^{ij}\biggl(\int_{\check r}^{r}
[\partial_i\sigma^\alpha(s)]
\bigl[\partial_\alpha
\bigl(\mathrm i\tilde b-\tfrac1{2}
\mathop{\mathrm{div}}\tilde \omega
\bigr)\bigr](s,\sigma(s))\,\mathrm{d}s \biggr)
\\&\phantom{={}+{}}
\cdot
\biggl(\int_{\check r}^r
[\partial_j\sigma^\alpha(s)]
\bigl[\partial_\alpha
\bigl(\mathrm i\tilde b-\tfrac1{2}
\mathop{\mathrm{div}}\tilde \omega
\bigr)\bigr](s,\sigma(s))\,\mathrm{d}s \biggr) u(r)(\sigma)
.
\end{split}
\label{eq:14.3.13.4.12}
\end{align}
Then by 
\eqref{eq:13:9:34:16:20}, \eqref{eq:14.3.4.3.37AA} 
and \eqref{eq:14.3.4.3.39} 
we can verify \eqref{eq:13.9.23.16.40b}.
 In fact, it is clear that the first and the second terms to the right of \eqref{eq:14.3.13.4.12}
satisfy the desired estimate.
The other terms can be treated more or less in a similar manner, 
and hence we consider only the fifth term.
By the Cauchy--Schwarz inequality, \eqref{eq:13:9:34:16:20}
\begin{align*}
&\biggl|\ell^{ij}
\biggl(\int_{\check r}^r
[\partial_i\sigma^\alpha(s)]
\bigl[\partial_\alpha
\bigl( \mathrm i\tilde b-\tfrac1{2}
\mathop{\mathrm{div}}\tilde \omega
\bigr)\bigr](s,\sigma(s))\,\mathrm{d}s \biggr)
[\partial_j\sigma^\beta(\check r)]
\biggr|
\\
&\le C_4(\check r{}^{1/2}/r)^{\check \sigma}
\int_{\check r}^{r}s^{-1-\min\{\tau,\rho\}/2+\check\sigma/2}\,\mathrm{d}s
\\&
\le C_5(\check r{}^{1/2}/r)^{\check\sigma}.
\end{align*}
This is the desired estimate. We omit the rest of the argument.
\end{proof}

\begin{proof}[Proof of Lemma~\ref{lem:14.5.1.16.30}]
 Write with  $\check r=R_n$ (recall that  possibly $n$ depends on the support of $g$)
 \begin{align*}
   \phi^+&=\bar \chi_n(r) b^{-1/2}\mathrm e^{\mathrm i(\check
     r-r)(\tilde A-\tilde b)}u;\\
u&=\mathrm e^{\mathrm i(r_0-\check
     r)(\tilde A-\tilde b)}g\in  C^1_{\mathrm
    c}(S_{\check r}). 
 \end{align*}
 First let us assume Condition~\ref{cond:12.6.2.21.13bb} (\ref{item:14.5.1.8.30}).
We decompose
\begin{align}
R(\i)L\phi^+=(R(\i)r^{-s}p')(p'r^s\phi^+);\quad s>1/2.
\label{eq:14.5.8.15.51}
\end{align}
Since the first factor is bounded as $\mathcal
H\to \mathcal H_s\cap \mathcal H^1\subseteq B\cap \mathcal H^1$, it
suffices to show that the second factor  belongs to $\mathcal H$ for
some $s>1/2$.  We combine Condition~\ref{cond:12.6.2.21.13bb}
(\ref{item:14.5.1.8.30}) with Lemma \ref{lemma:L_10b}  
and Corollary~\ref{cor:C_10} (these results applied with any  $\check\sigma\in (2,
\min\{\sigma, \tau, \rho\}$) and conclude that indeed $p'r^s\phi^+\in\mathcal H $ for some
$s>1/2$. Of course the conclusion
$R(\i)L\phi^+=\wBlim_{m\to \infty} (R(\i)r^{-s}p')\chi_m(p'r^s\phi^+)$
follows from this argument.

Next,  assuming  Condition~\ref{cond:12.6.2.21.13bb} (\ref{item:14.5.1.8.31}),
 we note that  $u\in C^2_{\mathrm
    c}(S_{\check r})$ and decompose
\begin{align*}
R(\i)L\phi^+=\big (R(\i)r^{-s}\big )\Big (\int \oplus L_r r^s\phi^+\,\d r\Big);\quad s>1/2,
%\label{eq:14.5.8.15.52}
\end{align*}
and proceed   using 
 Lemma \ref{lemma:L_10b} and  Corollary ~\ref{cor:C_10} to bound  
$L_r
  \phi^+-b^{-1/2}L_r
 b^{1/2}\phi^+$ and Lemma~\ref{lem:14.3.5.0.17} to bound  $b^{-1/2}L_r
 b^{1/2}\phi^+$, respectively. Next we  introduce a factor $\chi_m$ as
 above and conclude similarly.
  \end{proof}

\subsection{Scattering matrix and characterization of generalized
  eigenfunctions}\label{subsec:Scattering matrix and characterization
  of generalized eigenfunctions} 
  
In this subsection we prove Theorem~\ref{thm:char-gener-eigenf-1}.
Throughout the subsection we assume Condition~\ref{cond:12.6.2.21.13bb},
and we fix $\lambda\in \mathcal I$.

We begin with a partial uniqueness result.
\begin{lemma}\label{lem:14.5.12.13.42}
Suppose  $\phi\in \mathcal E_\lambda$ and $\xi_\pm\in\mathcal G$ satisfy
\begin{align}\label{eq:decF}
  \phi -\phi^+[\xi_+]+\phi^-[\xi_-]\in B_0^*.
\end{align}
Then $\xi_\pm$ are uniquely determined by $\phi$. Moreover
\begin{subequations}
    \begin{align}\label{eq:aEigenf2wb}
     \|\xi_+\|_{\vG}^2+\|\xi_-\|_{\mathcal
       G}^2&=\lim_{R\to\infty}R^{-1}\int_{B_{2R}\setminus B_{R}}
      b|\phi|^2\,(\det g)^{1/2}\d x,\\
 \|\xi_+\|_{\vG}&=\|\xi_-\|_{\mathcal G}.\label{eq:aEigenf2wb2}
    \end{align}
\end{subequations}
\end{lemma}
\begin{proof} The uniqueness statement follows from
  \eqref{eq:aEigenf2wb}, which in turn is proved as follows:
  \begin{align*}
    &\lim_{R\to\infty}R^{-1}\int_{B_{2R}\setminus B_{R}}
      b|\phi|^2\,(\det g)^{1/2}\d x\\&=\lim_{R\to\infty}R^{-1}\int_{B_{2R}\setminus B_{R}}
      b|\phi^+[\xi_+]-\phi^-[\xi_-]|^2\,(\det g)^{1/2}\d x\\&
=\|\xi_+\|_{\vG}^2+\|\xi_-\|_{\mathcal G}^2-2\Re \lim_{R\to
  \infty}-\!\!\!\!\!\!\int_R \inp{\xi_+, \exp \big(-2 \mathrm i\int_{r_0}^{r}\tilde b(s,{}\cdot{})\,\mathrm{d}s\big)\xi_-}_{\vG}\,\d r
  \end{align*} The last  term vanishes as may be seen by first writing
  \begin{align*}
    \exp \left(-2 \mathrm i\int_{r_0}^{r}\tilde
      b(s,{}\cdot{})\,\mathrm{d}s\right)=(-2\i \tilde b)^{-1} \tfrac{\d}{\d r}\exp \left(-2 \mathrm i\int_{r_0}^{r}\tilde b(s,{}\cdot{})\,\mathrm{d}s\right)
  \end{align*} and then integrate by parts picking up a sum of decaying
  factors. Note that indeed $\tfrac{\d}{\d r}\tilde
  b(r,\cdot)=o(R^{0})$ uniformly in the angle variable (so that \cs
  applies).

  As for \eqref{eq:aEigenf2wb2} first note that $A\phi\in B^*$, which comes from the
representation $A\phi=(\lambda-\i)AR(\i) \phi$ and the fact  that
$AR(\i)\in \vB(B^*)$.
 Then  we compute 
\begin{align*}
    0&=\lim_{n\to \infty}\inp{\i[H, \chi_n]}_\phi\\
&=\lim_{n\to \infty}\inp{A\chi_n'}_\phi\\
&=\lim_{n\to \infty}\inp{A\phi,
  \chi_n'(\phi^+[\xi_+]-\phi^-[\xi_-])}\\
&=\lim_{n\to \infty}\inp{\phi,
  \chi_n'(A\phi^+[\xi_+]-A\phi^-[\xi_-])}\\
&=\lim_{n\to \infty}\inp{\phi,
  \chi_n'(b\phi^+[\xi_+]+b\phi^-[\xi_-])}\\
&=\lim_{n\to \infty}\inp{\phi^+[\xi_+]-\phi^-[\xi_-],
  \chi_n'b(\phi^+[\xi_+]+\phi^-[\xi_-])}\\
&=\|\xi_+\|_{\vG}^2-\|\xi_-\|_{\mathcal G}^2,
\end{align*} where in the last step we integrated by parts as in the
proof of \eqref{eq:aEigenf2wb}.
\end{proof}

Next, we construct $\xi_\pm\in\mathcal G$ from $\phi\in\mathcal
E_\lambda$. Note for comparison that  $F^\pm(\lambda)^*\xi\in \mathcal E_\lambda$ for any $\xi\in
  \vG$ (readily proven by using 
  $F^\pm(\lambda)^*=(\lambda-\i)R(\i)F^\pm(\lambda)^*$, cf. the proof
  of Lemma \ref{lem:14.5.4.17.5}).  
\begin{lemma}\label{lem:14.5.14.3.30A}
  For any $\phi\in\mathcal E_\lambda$ there exist $\xi_\pm\in \vG$
  such that \eqref{eq:aEigenfw} hold.
\end{lemma}
\begin{proof} We  use  the scheme of proof of \cite[Proposition 6.2]
  {Sk}. By the definition of $S(\lambda)$  it
  suffices to show that for any  $\phi\in\mathcal E_\lambda$ the
  representation 
  $\phi= \i F^+(\lambda)^*\xi$ for some $\xi\in
  \vG$ holds.

Pick
  $f\in C^\infty_\c(\R)$ with $f(t)=t$ in neighbourhood of
  $t=\lambda$. Whence 
     $f(H)\phi=\lambda \phi$. We introduce (for a fixed large $m$)
\begin{align*}
    \phi_\pm&= \tfrac{1}{2b}\bar \chi_m(
A\pm b)\phi\in B^*,
\quad
\xi_n=F^+(\lambda)\chi_n \parb{f(H)-\lambda} \phi_+;\, n\in \N.
\end{align*} The sequence $(\xi_n)\subseteq \vG$ is
bounded. Indeed since $F^+(\lambda)\parb{f(H)-\lambda}=0$
(cf. \eqref{eq:diagf})  we compute 
using \eqref{eq:HeSj} (stated below) and estimate uniformly in $n\in\N$ and in $g\in
C_\c^\infty (S)$, $\|g\|_{\vG}=1$,
  \begin{align*}
    \inp{g,\xi_n}_{\vG}&=\i \inp{F^+(\lambda)^*g, \parb{A\chi_n' +\i
        |\d r|^2 \chi_n''/2}f'(H)\phi_+}_{B^*\times B},\\
|\inp{g,\xi_n}_{\vG}|&\leq
C_1|\parb{\|AF^+(\lambda)^*g\|_{B^*}+\|F^+(\lambda)^*g\|_{B^*}}\leq C_2.
  \end{align*} 

Next we  choose a weakly convergent subsequence of $(\xi_n)$, cf.
\cite[Theorem 1 p. 126]{Yo}. Whence,
possibly upon changing notation,  
$\wlim_{n\to\infty}\xi_n=:\xi\in\vG$.
For  this $\xi$ and with $\check f(t):=(f(t)-\lambda) (t-\lambda)^{-1}$ we compute
\begin{align*}
&\i F^+(\lambda)^*\xi\\
  &=\wBstarlim_{n\to \infty}\i
  F^+(\lambda)^*F^+(\lambda)\chi_n \parb{f(H)-\lambda} \phi_+\\
&=\wBstarlim_{n\to \infty} \parb{R(\lambda+\i 0)-R(\lambda-\i 0)}\chi_n \parb{f(H)-\lambda} \phi_+\\
&=\wBstarlim_{n\to \infty}\parbb{R(\lambda+\i 0)\chi_n \parb{f(H)-\lambda} \phi_+
+R(\lambda-\i 0)\chi_n \parb{f(H)-\lambda}  (\phi-\phi_+)}\\
&=\wBstarlim_{n\to \infty}\parbb{\check f(H)\chi_n  \phi
+R(\lambda+\i 0)[\chi_n ,f(H)]\phi_+ 
+R(\lambda-\i 0)[\chi_n ,f(H)]
(\phi-\phi_+)}\\
&=\check f(H)\phi +\wBstarlim_{n\to \infty}\parbb{R(\lambda+\i 0)[\chi_n ,f(H)]\phi_+  +R(\lambda-\i 0)[\chi_n ,f(H)]  (\phi-\phi_+)}.
\end{align*} The first term simplifies as  $\check f(H)\phi=\phi$. To compute the last term we represents in a standard
fashion (in terms of an
almost analytic extension $\tilde f$)
\begin{align}\label{eq:HeSj}
  f(H)=\int_\C R(z)\d\mu (z);\,\d \mu (z)=-(2\pi\i)^{-1}\bar \partial
  \tilde f (z)\d z\d \bar z,
\end{align} allowing us to compute
\begin{align*}
  [\chi_n ,f(H)]=-\i\int _\C R(z)\parb{A\chi_n' +\i
        |\d r|^2 \chi_n''/2}R(z)\d\mu (z).
    \end{align*} Whence due to Corollary \ref{cor:13.9.9.8.23} (note
    also that the
    second term with the factor $\chi_n''$ does not contribute to  the
    limit)
\begin{align*}
  &\wBstarlim_{n\to \infty}R(\lambda+\i 0)[\chi_n ,f(H)]\phi_+\\
&=-\i \wBstarlim_{n\to \infty}\int_\C R(z)R(\lambda+\i
0)A\chi_n'R(z)\d\mu (z)\phi_+\\
&=\i \wBstarlim_{n\to \infty}\int_\C R(z)R(\lambda+\i
0)b\chi_n'R(z)\d\mu (z)\phi_+\\
&=\i \wBstarlim_{n\to \infty}\int_\C R(z)R(\lambda+\i 0)R(z)\d\mu (z)b\chi_n'\phi_+\\
&=-\tfrac\i 2\wBstarlim_{n\to \infty}f'(H)R(\lambda+\i 0)\chi_n'\bar \chi_m(
A+ b)\phi\\
&=-\tfrac\i 2\wBstarlim_{n\to \infty}f'(H)R(\lambda+\i 0)(
A+ b)\chi_n'\phi\\
&= 0.
\end{align*} Similarly, using that $\phi-\phi_+=\chi_m \phi-\phi_-$,
\begin{align*}
  &\wBstarlim_{n\to \infty}R(\lambda-\i 0)[\chi_n ,f(H)](\phi-\phi_+)\\
&=-\tfrac\i 2\wBstarlim_{n\to \infty}f'(H)R(\lambda-\i 0)\chi_n'\bar \chi_m\parb{(
A- b)\phi}\\
&=-\tfrac\i 2\wBstarlim_{n\to \infty}f'(H)R(\lambda-\i 0)(
A- b)\chi_n'\phi\\
&= 0.
\end{align*}
Whence we have shown that  $\phi= \i F^+(\lambda)^*\xi$ for the constructed
$\xi$. 
\end{proof}

\begin{lemma}\label{lem:B_0star} For all $\psi\in B$ and all $\lambda\in\mathcal I$
  \begin{align}
    \label{eq:Basyp}
    \sqrt b R(\lambda \pm \i 0)\psi-1_M\mathrm e^{\mathrm
  i(r_0-r)(\tilde A^{\rm ex }\mp\tilde b^{\rm ex })}F^\pm(\lambda)\psi\in B_0^*.
  \end{align}
  \begin{proof}
    This is obvious from Lemma \ref{lemma:s1} for $\psi\in
    \vH_{1+}$. The general case is treated by an approximation
    argument (as in  the proof of \cite[Corollary 5.5] {Sk}).
  \end{proof}

\end{lemma}
A construction of $\phi\in\mathcal E_\lambda$ from
$\xi_\pm\in \mathcal G$ may   intuitively seem  most feasible  when $\xi_\pm$ satisfies the Dirichlet boundary condition.
We first give such construction for $\xi_\pm\in C^\infty_{\mathrm
  c}(S)$  and shortly  extend it allowing  any $\xi_\pm\in \mathcal G$.
\begin{lemma}\label{lem:14.5.14.3.27}
For any $\xi_-\in C^\infty_{\mathrm c}(S)$ introduce $\phi^-[\xi_-]\in \vN\cap B^*
$    by 
  \eqref{eq:apGen} (rather than by \eqref{eq:gen1B}) and define then $\phi\in\mathcal E_\lambda$ and $\xi_+\in \mathcal G$ by 
\begin{align}
  \begin{split}
\phi&=\psi_-
+(\lambda-\mathrm i)R(\lambda+\mathrm i0)\psi_--\phi^-[\xi_-],\\
\xi_+&=(\lambda-\mathrm i)F^+(\lambda)\psi_-;\quad \psi_-=R(\mathrm i)(H-\lambda)\phi^-[\xi_-]
.\label{eq:14.5.13.18.57}
  \end{split}
\end{align}
Then \eqref{eq:gen1} and \eqref{eq:aEigenfw} hold for $\{\xi_-,\xi_+,  \phi\}$.
\end{lemma}
\begin{proof} 
  Note that $\psi_-\in B$, cf. the proof of Lemma
  \ref{lem:14.5.4.17.5},   and that \eqref{eq:gen1} holds with the
  approximate eigenfunctions of \eqref{eq:gen1B} if the estimate is
  valid for those defined by \eqref{eq:apGen} (obviously the
  difference is in $B_0^*$).  We combine \eqref{eq:rep1} and
  \eqref{eq:rep3} (with the lower sign only) and Lemma
  \ref{lem:B_0star} (with the upper sign).
\end{proof}

Similarly we can first specify  $\xi_+\in C^\infty_{\mathrm c}(S)$
(the proof is similar).
\begin{lemma}\label{lem:14.5.14.3.27b}
For any $\xi_+\in C^\infty_{\mathrm c}(S)$ introduce $\phi^+[\xi_+]\in \vN\cap B^*
$    by 
  \eqref{eq:apGen} and define $\phi\in\mathcal E_\lambda$ and $\xi_-\in \mathcal G$ by 
\begin{align}
  \begin{split}
\phi&=\phi^+[\xi_+]-\psi_+
-(\lambda-\mathrm i)R(\lambda-\mathrm i0)\psi_+,\\
\xi_-&=(\lambda-\mathrm i)F^-(\lambda)\psi_+;\quad \psi_+=R(\mathrm i)(H-\lambda)\phi^+[\xi_+]
.\label{eq:14.5.13.18.57b}    
  \end{split}
\end{align}
Then \eqref{eq:gen1} and \eqref{eq:aEigenfw} hold for $\{\xi_-,\xi_+,  \phi\}$.
\end{lemma}

\begin{proof}[Proof of Theorem~\ref{thm:char-gener-eigenf-1}]
  \begin{comment}
    We claim that $\mathcal E_\lambda$ is closed in $B^*$.
In fact, let $\phi_n\in\mathcal E_\lambda$ be a convergent sequence in $B^*$
with a limit $\phi\in B^*$.
Since the limit $\phi$ obviously satisfies $(H-\lambda)\phi=0$ in the distributional sense,
it suffices to show that $\phi\in\mathcal N$, i.e.\ $\chi_m\phi\in\mathcal H^1$ for any $m\ge 0$.
We compute
\begin{align*}
&\bigl\|p\chi_m\phi_n-p\chi_m\phi_k\bigr\|^2_{\mathcal H}\\
&=2\mathop{\mathrm{Re}}\bigl\langle (\phi_n-\phi_k),\chi_m^2H_0(\phi_n-\phi_k)\bigr\rangle
-\bigl\langle \phi_n-\phi_k,[[H_0,\chi_m],\chi_m](\phi_n-\phi_k)\bigr\rangle\\
&=2\bigl\langle \phi_n-\phi_k,(E-V)\chi_m^2(\phi_n-\phi_k)\bigr\rangle
-\bigl\langle \phi_n-\phi_k,[[H_0,\chi_m],\chi_m](\phi_n-\phi_k)\bigr\rangle\\
&\le C\|\chi_{m+1}(\phi_n-\phi_k)\|^2_{\mathcal H},
\end{align*}
and obtain that $\chi_m\phi_n\in\mathcal H^1$ converges in $\mathcal H^1$.
Thus we have $\chi_m\phi\in\mathcal H^1$, and the claim follows.
In particular, we may think of $\mathcal E_\lambda$ as equipped with
the induced topology from $B^*$.
Since $\mathcal E_\lambda$ is closed in $B^*$,
  \end{comment}
Let any $\xi_-\in\mathcal G$ be given, and choose a sequence $\xi_{-,n}\in C^\infty_{\mathrm c}(S)$
such that $\xi_{-,n}\to \xi_-$ in $\mathcal G$ as $n\to\infty$.
By Lemma~\ref{lem:14.5.14.3.27} we have 
\begin{align*}
\i F^-(\lambda)^*\xi_{-,n}-\phi^+[S(\lambda)\xi_{-,n}]+\phi^-[\xi_{-,n}]\in B^*_0
\end{align*} (with  the approximate
eigenfunctions of \eqref{eq:gen1B}).
By the continuity of $F^-(\lambda)^*$, $S(\lambda)$ and $\phi^\pm[{}\cdot{}]$
we obtain, letting $n\to\infty$, 
\begin{align}\label{eq:brep}
  \i F^-(\lambda)^*\xi_--\phi^+[S(\lambda)\xi_-]+\phi^-[\xi_-]\in
  B^*_0.
\end{align} Whence \eqref{eq:gen1} and \eqref{eq:aEigenfw} hold for
$\{\xi_-,S(\lambda)\xi_-,  \i F^-(\lambda)^*\xi_-\}$, and  the existence part of \ref{item:14.5.13.5.40} follows when $\xi_-\in\mathcal G$ is given first.
We can proceed  similarly using Lemma~\ref{lem:14.5.14.3.27b} when
$\xi_+\in\mathcal G$ is given first,
and whence, with Lemma~\ref{lem:14.5.14.3.30A},
 the existence part of \ref{item:14.5.13.5.40} is completed.
In addition, the correspondences for  either $\xi_-\in\mathcal G$ or
$\xi_+\in\mathcal G$  given first are  given by 
\eqref{eq:aEigenfw}.

To complete \ref{item:14.5.13.5.40} it remains  to prove the uniqueness
part. 
Note that we already have a partial result in
Lemma~\ref{lem:14.5.12.13.42} (for 
$\phi$  given first).
 Let $\xi_-\in\mathcal G$ be given and suppose  that 
$\phi-\phi^+[\xi_+]+\phi^-[\xi_-]\in B_0^*$
 for some $\phi\in\mathcal E_\lambda$ and
$\xi_+\in\mathcal G$.  By linearity we may assume that $\xi_-=0$, and
it suffices to show that $\xi_+=0$ and $\phi=0$. Clearly by 
Lemma~\ref{lem:14.5.12.13.42} the vector $\xi_+=0$ and whence $\phi\in
B^*_0$. By Theorem \ref{thm:13.6.20.0.10} it then follows that
$\phi=0$. We can argue similarly if $\xi_+\in\mathcal G$ is given. We
have shown \ref{item:14.5.13.5.40}  and the formulas \eqref{eq:aEigenfw}.
The assertion \eqref{eq:aEigenfwB} for the upper sign
follows  from \eqref{eq:uniqFtoG}. We can
argue similarly
for the lower sign. Whence \ref{item:14.5.13.5.41} is shown.  

The
formulas  \eqref{eq:aEigenf2w} are
 immediate consequences of \eqref{eq:aEigenf2wb} and
\eqref{eq:aEigenf2wb2}, and in combination  with \ref{item:14.5.13.5.40} and
\ref{item:14.5.13.5.41} we conclude that indeed  $F^\pm(\lambda)^*\colon\vG\to \vE_\lambda\,(\subseteq B^*)$
    are bi-continuous. We have shown \ref{item:14.5.13.5.42}.

Finally, since $F^\pm(\lambda)^*$ are injective and have closed range
in $ B^*$ 
(by \ref{item:14.5.13.5.42}),
Banach's closed range theorem \cite[Theorem p. 205]{Yo} implies that
the range of $F^\pm(\lambda)$ for both signs coincides with $\mathcal
G$.  We conclude that  the
range of $\delta(H-\lambda)=(2\pi)^{-1}F^\pm(\lambda)^*F^\pm(\lambda)$
 coincides with $\mathcal E_\lambda$.  Hence
\ref{item:14.5.14.4.17} is shown.
\end{proof}

\subsection{Counter examples, open problems}
\label{subsec:Parabolic example} 
We consider modifications of the model of Example
\ref{ex:countex} and show that the asymptotics of the 
generalized eigenfunctions in $B^*$  for these models are {\it not} given by \eqref{eq:gen1}.  Fix
$\kappa\in (0,1)$, let $\theta:=xy^{-\kappa}$ for $y>0$ and let 
$r^2:=\kappa x^2+y^2$. Consider $M\subset \R^2$ with an end described as
\begin{align*}
  E=\{(x,y)\in\R\times \R_+|\quad r>r_0, \quad -1<\theta<1\},
\end{align*} which is  a cylinder in the variables $r$ and
$\theta$. The (inverse) 
metric in these coordinates are
\begin{align*}
  g^{rr}=N_r:=| \d r|^2,\quad g^{\theta\theta}=N_\theta:=| \d\theta|^2,\quad  g^{r\theta}=0.
\end{align*} 
Using the short-hand notation $|g|=\det g=N^{-1}_rN^{-1}_\theta$ we compute
\begin{align*}
  |g|^{1/4} \Delta |g|^{-1/4}&=\partial_rN_r\partial_r
  +\partial_\theta N_\theta\partial_\theta +W_r+W_\theta;\\
W_r&=-N_r(\partial_r \ln |g|)^2/16- (\partial_r N_r\partial_r \ln
|g|)/4,\\
W_\theta&=-N_\theta(\partial_\theta \ln |g|)^2/16- (\partial_\theta N_\theta\partial_\theta \ln |g|)/4.
\end{align*} 
We also compute
\begin{align*}
  &\partial_r x=\tfrac {\kappa x}{ rN_r},\quad \partial_\theta x=\tfrac 1
  { y^\kappa N_\theta},\quad \partial_r y=\tfrac
  y{rN_r},\quad \partial_\theta y=-\tfrac
  {\kappa \theta}{yN_\theta },
\\&N_r=1-(\kappa-\kappa^2) \theta^2 y^{2\kappa}r^{-2},\quad N_\theta=y^{-2\kappa}+\tfrac {\kappa^2\theta^2}{y^2},
\\&\partial_r N_r=2(1-\kappa N_r^{-1})(\kappa-\kappa^2) \theta^2
y^{2\kappa}r^{-3}= O(r^{2\kappa-3}), 
\quad  \partial^2_r
  N_r=O(r^{2\kappa-4}),\\&\quad \partial_\theta N_r=-2(1-\kappa^2 N_\theta^{-1}\theta^2y^{-2})(\kappa-\kappa^2) \theta
y^{2\kappa}r^{-2}=O(r^{2\kappa-2}),\quad \partial^2_\theta N_r=O(r^{2\kappa-2}),
\\&\partial_r N_\theta= -\tfrac{2\kappa}{rN_r}\parb{y^{-2\kappa}+\kappa
\theta^2y^{-2}}=O(r^{-1-2\kappa}),\quad  \partial^2_r
N_\theta=O(r^{-2-2\kappa}),\\& \partial_\theta N_\theta=\tfrac{2\kappa^2
\theta}{y^{2}N_\theta }\parb{y^{-2\kappa} +\kappa \theta^2y^{-2}}+\tfrac{2\kappa^2
\theta}{y^{2}}=\tfrac{4\kappa^2
\theta}{r^{2}}\parb{1+O\parb{r^{2\kappa-2}}}
,\\&\quad\, \partial^2_\theta N_\theta=\tfrac{4\kappa^2}{r^{2}}\parb{1+O\parb{r^{2\kappa-2}}}.
\end{align*} 
Using these formulas and 
$\partial_*\ln |g|=-(\partial_* N_r)/ N_r-(\partial_*  N_\theta)/N_\theta$ we get
\begin{align*}
  W_r=O(r^{-2}),\quad W_\theta=O(r^{-2}).
\end{align*}

We consider for $\kappa\in(0, 1/2]$ the approximate
outgoing eigenfunction (corresponding to any $\lambda>0$ and here with
$r_0=r_0(\lambda)$ chosen big enough)
\begin{align}\label{eq:appEiyb}
  \begin{split}
 \phi^+&:=\bar \chi _n|g|^{-1/4}b^{-1/2}\e^{\i \int^r_{r_0} b \,\d r} u(\theta)\\
&\approx C(\lambda)r^{
-\kappa/2}\e^{\i \int^r_{r_0} b \,\d r} u(\theta).   
  \end{split}
\end{align}
 Here
 \begin{align*}
   b=\sqrt{ 2(\lambda-\tfrac{
    \mu(\lambda)}{r^{2\kappa}})}\approx \sqrt{ 2\lambda}-\tfrac{
    \mu(\lambda)}{\sqrt{ 2\lambda}}r^{-2\kappa},
 \end{align*} $u=u(\theta)$ is any  Dirichlet eigenstate of the
 operator on $L^2((-1,1), \d \theta)$ given by 
 \begin{align*}
   H_D&:=-\tfrac 12\partial^2_\theta\text{ for }\kappa<1/2,\\
H_D&:=-\tfrac 12\partial^2_\theta-\tfrac{\lambda\theta^2}{4}\text{ for }\kappa=1/2,
 \end{align*} and $\mu(\lambda)$ is the  corresponding eigenvalue. To see
 why this is an  approximate eigenfunction we first note that $\phi^+\in
 \vN \cap B^*$. We claim that in fact 
 \begin{align*}
   &(H-\lambda)\phi^+\in r^{2\kappa-2}B^*\subseteq B\text{ for }\kappa<1/2,\\
 &(H-\lambda)\phi^+\in r^{-2}B^*\subseteq B\text{ for }\kappa=1/2.
 \end{align*} We compute for $\kappa=1/2$ (skipping the details for $\kappa<1/2$)
 \begin{align*}
   \partial_rN_r\partial_r\parb{b^{-1/2}\e^{\i \int^r_{r_0} b \,\d r} u}
   &=\parb{-b^2 +\tfrac{\lambda\theta^2}{2r}+O(r^{-2})}b^{-1/2}\e^{\i
     \int^r_{r_0} b \,\d r} u,\\
\partial_\theta N_\theta\partial_\theta\parb{b^{-1/2}\e^{\i \int^r_{r_0} b \,\d r} u}
   &=b^{-1/2}\e^{\i \int^r_{r_0} b \,\d r} \parb{\tfrac{1}{r}\partial^2_\theta u+O(r^{-2})}.
 \end{align*} In the first identity we substitute $b^2=2(\lambda-\tfrac{
    \mu(\lambda)}{r})$. Then we collect our computations and indeed obtain
\begin{align*}
   (H-\lambda)\phi^+=(H-\lambda)\phi^+-|g|^{-1/4}b^{-1/2}\e^{\i \int^r_{r_0} b \,\d r}
   r^{-1}\parb{H_D-\mu(\lambda)}u\in  r^{-2}B^*.
 \end{align*} 

Next we define
\begin{align*}
  \phi_u= \phi^+-R(\lambda-\i 0)(H-\lambda)\phi^+.
\end{align*} This $\phi_u$ is in $\vE_\lambda$ with
non-trivial prescribed outgoing asymptotics. If we look
at all eigenstates of $H_D$, say numbered by $k\in\N$,  we
obtain 
several generalized eigenfunctions this way. Note that for $\kappa=1/2$
\begin{align*}
  \e^{\i \int^r_{r_0} b \,\d r} \approx \e^{\i \sqrt{2\lambda} r} \exp\parb
{-\i \tfrac{
    \mu(\lambda)}{\sqrt{2\lambda}}\ln r}.
\end{align*}
Due to the
non-trivial factor
\begin{align*}
  \exp\parb
{-\i \tfrac{
    \mu(\lambda,k)}{\sqrt{2\lambda}}\ln r}
\end{align*} the asymptotics \eqref{eq:gen1} is
readily seen to be {\it incorrect}
(seen by using just two of the constructed  generalized
eigenfunctions). By a similar reasoning this conclusion is also valid
for $\kappa<1/2$.

The methods of this paper (in combination with other ingredients) should yield a modification of Theorem
\ref{thm:char-gener-eigenf-1} where the asymptotics of {\it any} $\phi\in
\vE_\lambda$ should be provided by  functions of the form \eqref{eq:appEiyb} and
their 
incoming counterparts, say in combination denoted by
$\{\phi_k^{\pm}|k\in\N\}$. This would intuitively yield the identification of
the  limiting space as $\vG=l^2(\N)$, but we shall 
not elaborate at this point.

For $\kappa\in  (1/2, 1)$ we do not know how to construct  approximate
outgoing eigenfunctions in $\vN \cap B^*$. If for example  we take
$b=\sqrt{2\lambda}$ and $u$ any nonzero function in the domain of the
Dirichlet Laplacian on $(-1,1)$ in
\eqref{eq:appEiyb} we obtain
\begin{align*}
  (H-\lambda N_r)\phi^+\in r^{-2\kappa}B^*\subseteq B,
\end{align*} which shows that
\begin{align*}
  (H-\lambda)\phi^+\notin B,
\end{align*} since $1-N_r\approx
(\kappa-\kappa^2)\theta^2r^{2\kappa-2}$ is long-range for $\kappa\in
(1/2, 1)$. The reader might think that a better approximation to the
eikonal equation than $\sqrt{2\lambda}r$ could be given to construct  concrete approximate
outgoing eigenfunctions in $\vN \cap B^*$ to cure this deficiency, however a closer examination
 indicates that this is not feasible 
(note that the forward flow property is a severe
restriction). The ansatz \eqref{eq:gen1B} has a similar deficiency. Whence
the asymptotics of the generalized eigenfunctions in $\vE_\lambda$ is
not known to us for $\kappa\in  (1/2, 1)$.

\begin{comment}
 
 The case $\kappa\in (1/2, 1)$ might     be treated in terms of
a better approximation
 to the eikonal equation than done in this paper? In fact we can find
 a better one of the form
 \begin{align*}
   R= r(1+\sum^M_{m=1}c_m z^{2m}),
 \end{align*} where $z=r^{\kappa-1}\theta$ and
 \begin{align*}
   8c_1=-(4\kappa-2)\pm \sqrt{(4\kappa-2)^2+16(\kappa-\kappa^2)}=2\pm 2-4\kappa.
 \end{align*} To see this we ``solve'' the eikonal equation $1=|\d R|^2=N_r
 (\partial_rR)^2+N_\theta
 (\partial_\theta R)^2$ recursively. Note that
 \begin{align*}
   N_r=N_r(z^2)\mand  r^{2\kappa}N_\theta\text{ similarly is smooth in
   }z^2.
 \end{align*} We need the choice of ``minus''. Whence
 \begin{align*}
  c_1=- \tfrac 12\kappa.
 \end{align*}

Next we consider
\begin{align*}
  \phi\approx J^{-1/2}\e^{\i \sqrt{2\lambda} R} u, 
\end{align*} where $u$ needs to be a function on the curve $R=r_0$ and
$J$ a certain Jacobian. Unfortunately this $R$ probably does not fulfill the
convexity requirement: it seems that $R=y$ (or some
approximation). The ``plus'' probably corresponds to an approximation to
$\sqrt{x^2+y^2}$ which does not fulfill forward completeness.

\end{comment}

\end{document}